\documentclass[graybox]{svmult}

\usepackage{mathptmx}       
\usepackage{helvet}         
\usepackage{courier}        
\usepackage{type1cm}        
\usepackage{graphicx}        
\usepackage{multicol}        
\usepackage[bottom]{footmisc}

\makeindex             

\usepackage{amsmath,amsfonts,amssymb}

\usepackage{mathtools}
\usepackage[framemethod=tikz]{mdframed}                         
\usepackage{bbm}                                                
\usepackage{xfrac}                                              
\usepackage{bm}                                                 %
\usepackage[noend]{algpseudocode}                               %
\usepackage{algorithm,algorithmicx}                             %

\floatname{algorithm}{Procedure}
\usepackage{placeins}

\usepackage{enumerate}
\usepackage[colorlinks=true, linktocpage=false]{hyperref}       %
\hypersetup{
	citecolor = {blue},
	linkcolor = {cyan!90!black}
}
\makeatletter
\def\ProblemSpecBox{
	\@ifnextchar[\ProblemSpecBox@opt{\ProblemSpecBox@noopt}}
\def\ProblemSpecBox@opt[#1]#2{
	\protected@edef\@currentlabelname{#1}
	\protected@edef\@currentlabel{#1}
	\begin{mdframed}[
		innerlinewidth=0.5pt,
		innerleftmargin=10pt,
		innerrightmargin=10pt,
		innertopmargin = 10pt,
		innerbottommargin=10pt,
		skipabove=\dimexpr\topsep+\ht\strutbox\relax,
		roundcorner=5pt,
		frametitle={#2},
		frametitlerule=true,
		frametitlerulewidth=1pt]
	}
	\def\ProblemSpecBox@noopt#1{
		\ProblemSpecBox@opt[#1]{#1}
	}
	\def\endProblemSpecBox{
	\end{mdframed}
}
\makeatother

\newcommand{\ru}[1]{{#1}roblem {RU}} 
\newcommand{\gop}[1]{{#1}roblem {GO}} 
\newcommand{\rpe}[1]{{#1}roblem {RPE}} 

\newcommand{\pisiSE}{$\mathrm{\Pi\Sigma}$}
\newcommand{\rpisiSE}{$\mathrm{R}\mathrm{\Pi\Sigma}$}
\newcommand{\aE}{$\mathrm{A}$}
\newcommand{\rE}{$\mathrm{R}$}
\newcommand{\apE}{$\mathrm{AP}$}
\newcommand{\rpiE}{$\mathrm{R\Pi}$}
\newcommand{\pE}{$\mathrm{P}$}
\newcommand{\piE}{$\mathrm{\Pi}$}
\newcommand{\sigmaE}{$\mathrm{\Sigma}$}
\newcommand{\ii}{\mathbbm{i}} 
\newcommand{\ee}{\mathbbm{e}} 
\newcommand{\sm}{\setminus}
\newcommand{\s}{\sigma}
\newcommand{\bs}{\boldsymbol}

\newcommand{\myt}{\mathfrak{t}}

\newcommand{\CC}{\mathbb{C}}

\newcommand{\EE}{\mathbb{E}}
\newcommand{\HH}{\mathbb{E}}
\newcommand{\KK}{\mathbb{K}}
\newcommand{\NN}{\mathbb{N}}
\newcommand{\ZZ}{\mathbb{Z}}
\newcommand{\FF}{\mathbb{F}}
\newcommand{\QQ}{\mathbb{Q}}
\newcommand{\XX}{\mathbb{X}}
\newcommand{\VV}{\mathbb{V}}

\renewcommand{\AA}{\mathbb{A}}
\newcommand{\zvs}[1]{\{ {\bf 0}_{#1} \}} 
\newcommand{\dField}[2][\sigma]{( {#2}, {#1})} 
\newcommand{\KarrModule}[2]{{\bf M} \left( {#1}, {#2} \right)} 
\newcommand{\ringOfSeqs}{\mathbb{K}^{\mathbb{N}}} 
\newcommand{\ringOfEquivSeqs}[1][\KK]{\mathcal{S}({#1})} 

\DeclareMathOperator{\lcm}{lcm}
\DeclareMathOperator{\sign}{sign}
\DeclareMathOperator{\ev}{ev} 
\DeclareMathOperator{\ProdExpr}{ProdE}
\DeclareMathOperator{\Prod}{Prod}

\newcommand{\geno}[1]{\left\langle {#1} \right\rangle}
\newcommand{\ProdLst}[4]{ {{#1}^{{#4}_{#2}}_{#2}}\cdots{{#1}^{{#4}_{#3}}_{#3}} } 
\newcommand{\prodLst}[5]{ \prod_{{#1}={#2}}^{{#3}} {{#4}_{#1}}^{{#5}_{#1}} } 
\newcommand{\fourargsexpsubscript}[4]{ {#1}^{{#2}_{{#3},{#4}}}_{{#3},{#4}} } 
\newcommand{\fiveargsexpsubscript}[5]{ {#1}^{{#2}_{{#3},{{#4}_{#5}}}}_{{#3},{{#4}_{#5}}} } 
\newcommand{\vectorexpsubscript}[3]{{{\bs{#1}}^{{\bs{#2}}_{\bs{#3}}}_{\bs{#3}}}} 
\newcommand{\twoargstwosamesubscript}[2]{{#1}_{{#2},{#2}}} 
\newcommand{\threeargstwosamesubscript}[3]{ {#1}_{{#2},{{#3}_{#2}}} } 
\newcommand{\threeargssubscript}[3]{ {#1}_{{#2},{#3}} } 
\newcommand{\fourargssubscript}[4]{{#1}_{{#2},{{#3}_{#4}}}} 
\newcommand{\shp}{\mathrm{gcd}_{\sigma}} 
\newcommand{\const}{\mathrm{const}} 
\newcommand{\zv}[1]{  {\bf 0}_{#1} }
\newcommand{\zs}{\{ 0 \}} 
\newcommand{\abs}[1]{| {#1} |} 
\newcommand{\norm}[1]{\lVert { #1 } \rVert} 
\newcommand{\seqA}[4][0]{\big\langle {#2}_{#3}{(#4)} \big\rangle_{{#4} \geq {#1}} }
\newcommand{\constSeqA}[3][0]{{\big\langle {#2}^{#3} \big\rangle}_{{#3} \geq {#1}}} 
\newcommand{\funcSeqA}[3][0]{\big\langle {#2} \big\rangle_{{#3} \geq {#1}}} 
\newcommand{\polySeq}[2]{\big({#1}\big)^{#2}}
\newcommand{\rootOfUnitySeq}[2]{\big({(-1)^{\frac{1}{#1}}\big)^{#2}}}
\newcommand{\rootOfUnitySeqExp}[3]{\big(\big({(-1)^{\frac{1}{#1}}\big)^{#2}}\big)^{#3}}
\newcommand{\algSeq}[2]{\big(\sqrt{#1}\big)^{#2}}
\newcommand{\algSeqExp}[3]{\big(\big(\sqrt{#1}\big)^{#2}\big)^{#3}}
\newcommand{\intSeq}[2]{{#1}^{#2}} 
\newcommand{\intSeqExp}[3]{\big({#1}^{#2}\big)^{#3}} 
\newcommand{\polySeqExp}[3]{\big({\big({#1}\big)^{#2}}\big)^{#3}}
\newcommand{\myProduct}[4]{\smashoperator{\prod_{\mathclap{\substack{{#1}={#2}}}}^{{#3}}} {#4}} 
\newcommand{\myProdWithInnerUnderbrace}[5]{\smashoperator{\prod_{\mathclap{\substack{{#1}={#2}}}}^{{#3}}} 
	\underbrace{#4}_{=: \, #5}}
\newcommand{\Lst}[3]{{#1}_{#2},\dots,{#1}_{#3}} 
\newcommand{\Lstc}[4]{{#1}_{{#2},{#3}},\dots,{#1}_{{#2},{#4}}} 
\newcommand{\Lsth}[4]{{#1}_{#2},\dots,{#1}_{#3}, {#4}} 
\makeatletter
\newsavebox\myboxA
\newsavebox\myboxB
\newlength\mylenA
\newcommand{\xoverline}[2][0.75]{%
	\sbox{\myboxA}{$\m@th#2$}%
	\setbox\myboxB\null
	\ht\myboxB=\ht\myboxA%
	\dp\myboxB=\dp\myboxA%
	\wd\myboxB=#1\wd\myboxA
	\sbox\myboxB{$\m@th\overline{\copy\myboxB}$}
	\setlength\mylenA{\the\wd\myboxA}
	\addtolength\mylenA{-\the\wd\myboxB}%
	\ifdim\wd\myboxB<\wd\myboxA%
	\rlap{\hskip 0.5\mylenA\usebox\myboxB}{\usebox\myboxA}%
	\else
	\hskip -0.5\mylenA\rlap{\usebox\myboxA}{\hskip 0.5\mylenA\usebox\myboxB}%
	\fi}
\makeatother
\newcommand{\dr}{difference ring} 
\newcommand{\df}{difference field}
\newcommand{\drE}{difference ring extension}
\newcommand{\dfE}{difference field extension}

\spnewtheorem{mtheorem}{Theorem}[section]{\bfseries}{\itshape}
\spnewtheorem{mdefinition}{Definition}[section]{\bfseries}{}
\spnewtheorem{mlemma}{Lemma}[section]{\bfseries}{\itshape}
\spnewtheorem{mexample}{Example}[section]{\bfseries}{}
\spnewtheorem{mcorollary}{Corollary}[section]{\bfseries}{\itshape}
\spnewtheorem{mproposition}{Proposition}[section]{\bfseries}{\itshape}

\begin{document}

\title*{Representing (\texorpdfstring{$\boldsymbol{q}$}{q}--)hypergeometric products and mixed versions in difference rings\thanks{Supported by the Austrian Science Fund (FWF) grant SFB F50 (F5009-N15)}
\bigskip
{\small \text{Dedicated to Sergei A. Abramov on the occasion of his 70th birthday}}}
\titlerunning{Representing (\texorpdfstring{$q$}{q}--)hypergeometric products and mixed versions in difference rings}
\author{Evans Doe Ocansey and Carsten Schneider}
\institute{Evans Doe Ocansey \at Research Institute for Symbolic Computation (RISC), Johannes Kepler University, Altenbergerstra\ss e 69, 4040, Linz, Austria, \email{eocansey@risc.jku.at}
\and Carsten Schneider \at Research Institute for Symbolic Computation (RISC), Johannes Kepler University, Altenbergerstra\ss e 69, 4040, Linz, Austria, \email{cschneid@risc.jku.at}}
%
%
\maketitle

\abstract{
	In recent years, Karr's difference field theory has been extended to the so-called \rpisiSE-extensions in which one can  represent not only indefinite nested sums and products that can be expressed by transcendental ring extensions, but one can also handle algebraic products of the form \texorpdfstring{$\alpha^n$}{alpha\^n} where \texorpdfstring{$\alpha$}{alpha} is a root of unity. In this article we supplement this summation theory substantially by the following building block. We provide new algorithms that represent a finite number of hypergeometric or mixed \texorpdfstring{$(\Lst{q}{1}{e})$}{(q1,...,qe)}-multibasic hypergeometric products in such a difference ring. This new insight provides a complete summation machinery that enables one to formulate such products and indefinite nested sums defined over such products in \rpisiSE-extensions fully automatically. As a side-product, one obtains compactified expressions where the products are algebraically independent among each other, and one can solve the zero-recognition problem for such products.}

\section{Introduction}

Symbolic summation in difference fields has been introduced by Karr's groundbreaking work~\cite{karr1981summation,karr1985theory}. 
He defined the so-called \pisiSE-fields $\dField{\FF}$ which are composed by a field $\FF$ and a field automorphism $\s:\FF\to\FF$. Here the field $\FF$ is built by a tower of transcendental field extensions whose generators either represent sums or products where the summands or multiplicands are elements from the field below.
In particular, the following problem has been solved: given such a \pisiSE-field $\dField{\FF}$ and given $f\in\FF$. Decide algorithmically, if there exists a $g\in\FF$ with 
\begin{equation}\label{Equ:DFTele}
f=\s(g)-g.
\end{equation}
Hence if $f$ and $g$ can be rephrased to expressions $F(k)$ and $G(k)$ in terms of indefinite nested sums and products, one obtains the telescoping relation
\begin{equation}\label{Equ:SeqTele}
F(k)=G(k+1)-G(k).
\end{equation}
Then summing this telescoping equation over a valid range, say $a\leq k\leq b$, one gets the 
identity $\sum_{k=a}^bF(k)=G(b+1)-G(a).$\\
In a nutshell, the following strategy can be applied: (I) construct an appropriate  \pisiSE-field $\dField{\FF}$ in which a given summand $F(k)$ in terms of indefinite nested sums and products is rephrased by $f\in\FF$; (II) compute $g\in\FF$ such that~\eqref{Equ:DFTele} holds; (III) rephrase $g\in\FF$ to  an expression $G(k)$ such that~\eqref{Equ:SeqTele} holds.

In the last years various new algorithms and improvements of Karr's difference field theory have been developed in order to obtain a fully automatic simplification machinery for nested sums. Here the key observation is that a sum can be either expressed in the existing difference field $\dField{\FF}$ by solving the telescoping problem~\eqref{Equ:SeqTele} or --if this is not possible-- it can be adjoined as a new extension on top of the already constructed field $\FF$ yielding again a \pisiSE-field; see Theorem~\ref{thm:rpiE-Criterion}(3) below. By a careful construction of $\dField{\FF}$ one can simplify sum expressions such that the nesting depth is minimized~\cite{Schneider:08c}, or the number~\cite{Schneider:15} or the degree~\cite{Schneider:07d} of the objects arising in the summands are optimized.

In contrast to sums, representing products in \pisiSE-fields is not possible in general. In particular, the alternating sign $(-1)^k$, which arises frequently in applications, can be represented properly only in a ring with zero divisors introducing relations such as
$(1-(-1)^k)(1+(-1)^k)=0$.
In~\cite{schneider2005product} and a streamlined version worked out in~\cite{Schneider:14}, this situation has been cured for the class of hypergeometric products of the form $\prod_{i=l}^kf(i)$ with $l\in\NN$ and $f(k)\in\QQ(k)$ being a rational function with coefficient from the rational numbers: namely, a finite number of such products can be always represented in a \pisiSE-field adjoined with the element $(-1)^k$.
In particular, nested sums defined over such products can be formulated automatically in difference rings built by the so-called~\rpisiSE-extensions~\cite{schneider2016difference,schneider2017summation}. This means that the difference rings are constructed by transcendental ring extensions and algebraic ring extensions with generators of the form $\alpha^n$
where $\alpha$ is a primitive root of unity. Within this setting~\cite{schneider2016difference,schneider2017summation}, one can then solve the telescoping problem for indefinite sums (see Eq.~\eqref{Equ:DFTele}) and more generally the creative telescoping problem~\cite{petkovvsek1996b} to compute linear recurrences for definite sums. Furthermore one can simplify the so-called d'Alembertian~\cite{Petkov:92,Abramov:94,Abramov:96,vanHoeij:99} 
or Liouvillian solutions~\cite{van2006galois,Petkov:2013} of linear recurrences which are given in terms of nested sums defined over hypergeometric products. For many problems coming, e.g., from combinatorics or particle physics (for the newest applications see~\cite{Sulzgruber:16} or~\cite{Schneider:16b}) this difference ring machinery with more than 100 extension variables works fine. But in more general cases, one is faced with nested sums defined not only over hypergeometric but also over mixed multibasic products. Furthermore, these products might not be expressible in $\QQ$ but only in an algebraic number field, i.e., in a finite algebraic field extension of $\QQ$.

In this article we will generalize the existing product algorithms~\cite{schneider2005product,Schneider:14} to  cover also this more general class of products.

\begin{mdefinition}\label{defn:prodObjs}
	Let $\KK = K(\Lst{q}{1}{e})$ be a rational function field over a field $K$ and let $\FF=\KK(x,\Lst{t}{1}{e})$ be a rational function field over $\KK$. $\prod_{k=l}^{n}f(k,q_{1}^{k},\dots,q_{e}^{k})$ is a \emph{mixed $(q_1,\dots,q_e)$-multibasic hypergeometric product} in $n$, if $f(x,\Lst{t}{1}{e})\in \FF\setminus\{0\}$ and $l\in\NN$ is chosen big enough (see Ex.~\ref{exa:DiffFieldOfqMixedSeqs} below) such that $f(\ell,q_{1}^{\ell},\dots,q_{e}^{\ell})$ has no pole and is non-zero for all $\ell\in\NN$ with $\ell\ge l$. If $f(\Lst{t}{1}{e})\in\FF$ which is free of $x$, then $\prod_{k=l}^{n}f(q_{1}^{k},\dots,q_{e}^{k})$ is called a $(q_1,\dots,q_e)$-\emph{multibasic hypergeometric product} in $n$. If $e=1$, then it is called a \emph{basic} or \emph{$q$-hypergeometric product} in $n$ where $q=q_{1}$. If $e=0$ and $f\in\KK(x)$, then $\prod_{k=l}^{n}f(k)$ is called a \emph{hypergeometric product} in $n$. Finally, if $f\in\KK$, it is called \emph{constant or geometric product} in $n$.\\ 
	Let $\bs{q}^{n}$ denote $q_{1}^{n},\dots,q_{e}^{n}$ and
	$\bs{t}$ denote $(t_1,\dots,t_e)$.
	Further, we define the set of ground expressions\footnote{Their elements are considered as expressions that can be evaluated for sufficiently large $n\in\NN$.} $\KK(n)=\{f(n)\mid f(x)\in\KK(x)\}$, $\KK(\bs{q}^{n})=\{f(\bs{q}^n)\mid f(\bs{t})\in\KK(\bs{t})\}$ and $\KK(n,\bs{q}^{n})=\{f(n,\bs{q}^n)\mid f(x,\bs{t})\in\KK(x,\bs{t})\}$.
        Moreover, we define       
        $\Prod(\XX)$ with $\XX\in\{\KK, \KK(n), \KK(\bs{q}^{n}),\KK(n,\bs{q}^{n})\}$ as 
	the set of all such products where the multiplicand is taken from $\XX$. Finally, we introduce the set of product expressions $\ProdExpr(\XX)$ as the set of all elements	
	\begin{equation}\label{Equ:ProdEDef}
	  \smashoperator{\sum_{\mathclap{\substack{{(\Lst{\nu}{1}{m})\in S}{}}}}^{{}}}{a_{(\Lst{\nu}{1}{m})}(n)\,P_{1}(n)^{\nu_{1}}\cdots P_{m}(n)^{\nu_{m}}}
	\end{equation}
	with $m\in\NN$, $S \subseteq \ZZ^{m}$ finite, $a_{(\Lst{\nu}{1}{m})}(n)\in\XX$ and $P_{1}(n),\dots,P_{m}(n)\in\Prod(\XX)$.
\end{mdefinition}

\noindent For this class where the subfield $K$ of $\KK$ itself can be a rational function field over an algebraic number field, we will solve the following problem.

\vspace*{-0.3cm}
\begin{ProblemSpecBox}[\rpe{P}]{
		{\bf \rpe{P}: Representation of Product Expressions.}
	}\label{prob:ProblemRPE}
	{
		Let $\XX_{\KK}\in\{\KK,\KK(n), \KK({\bs q}^{n}),\KK(n,{\bs q}^{n})\}$. \emph{Given} $P(n)\in\ProdExpr(\XX_{\KK})$;\\ \emph{find} $Q(n)\in\ProdExpr(\XX_{\KK^{\prime}})$ with $\KK^{\prime}$ a finite algebraic field extension\footnote{If $\KK=K(\Lst{\kappa}{1}{u})(\Lst{q}{1}{e})$ is a rational function field over an algebraic number field $K$, then in worst case $\KK$ is extended to $\KK^{\prime}=K^{\prime}(\kappa_1,\dots,\kappa_u)(q_1,\dots,q_e)$ where $K^{\prime}$ is an algebraic extension of $K$. Subsequently, all algebraic field extensions are finite.} of $\KK$ and a natural number $\delta$ with the following properties: 
		\begin{enumerate} 
			\item $P(n) = Q(n)$ for all $n \in \NN$ with $n \geq \delta$;
			\item The product expressions in $Q(n)$ (apart from products over roots of unity) are algebraically independent among each other.
			\item The zero-recognition property holds, i.e., $P(n)=0$ holds for all $n$ from a certain point on if and only if $Q(n)$ is the zero-expression. 
		\end{enumerate}
	}
\end{ProblemSpecBox}
Internally, the multiplicands of the products are factorized and the monic irreducible factors, which are shift-equivalent, are rewritten in terms of one of these factors; compare~\cite{abramov1971summation,Paule:95,schneider2005product,Petkov:10,ZimingLi:11}. Then using results of~\cite{schneider2010parameterized,Singer:08} we can conclude that products defined over these irreducible factors can be rephrased as transcendental difference ring extensions. Using similar strategies, one can treat the content coming from the monic irreducible polynomials, and obtains finally an \rpisiSE-extension in which the products can be rephrased. We remark that the normal forms presented in~\cite{ZimingLi:11} are closely related to this representation and enable one to check, e.g., if the given products are algebraically independent. Moreover, there is an algorithm~\cite{KZ:08} that can compute all algebraic relations for $c$-finite sequences, i.e., it finds certain ideals from $\ProdExpr(\KK)$ whose elements evaluate to zero.  Our main focus is different. We will compute alternative products which are by construction algebraically independent among each other and which enable one to express the given products in terms of the algebraic independent products. In particular, we will make this algebraic independence statement (see property (2) of \ref{prob:ProblemRPE}) very precise by embedding the constructed \rpisiSE-extension explicitly into the ring of sequences~\cite{petkovvsek1996b} by using results from~\cite{schneider2017summation}. The derived algorithms implemented in Ocansey's Mathematica package \texttt{NestedProducts} supplement the summation package \texttt{Sigma}~\cite{Sigma}
and enable one to formulate nested sums over such general products in the setting of \rpisiSE-extensions. As a consequence, it is now possible to apply completely automatically the summation toolbox~\cite{Petkov:10,bronstein2000solutions,karr1981summation,schneider2005product,Schneider:08c,Schneider:07d,schneider2010parameterized,Sigma,Schneider:14,Schneider:15,schneider2016difference,schneider2016symbolic,schneider2017summation} for simplification of indefinite and definite nested sums defined over such products.

The outline of the article is as follows. In Section~\ref{sec:RingOfSeqsDiffRingsAndDiffFields} we define the basic notions of \rpisiSE-extensions and present the main results to embed a difference ring built by \rpisiSE-extensions into the ring of sequences. In Section~\ref{sec:mainResult} our \ref{prob:ProblemRPE} is reformulated to Theorem~\ref{thm:ProblemRMHPE} in terms of these notions, and the basic strategy of how this problem will be tackled is presented. In Section~\ref{Sec:AlgProperties} the necessary properties of the constant field are worked out that enable one to execute our proposed algorithms. Finally, in Sections~\ref{Sec:HypergeometricCase} and~\ref{Sec:mixedHypergeometricCase} the hypergeometric case and afterwards the mixed multibasic case are treated. A conclusion is given in Section~\ref{Sec:Conclusion}.

\section{Ring of sequences, difference rings and difference fields}\label{sec:RingOfSeqsDiffRingsAndDiffFields}

In this section, we discuss the algebraic setting of {\dr s} (resp. fields) and the ring of sequences as they have been elaborated in \cite{karr1981summation,schneider2016difference,schneider2017summation}. In particular, we demonstrate how sequences generated by expressions in ${\ProdExpr(\KK(n))}$ (resp. ${\ProdExpr(\KK(n, {\bs q}^{n}))}$) can be modeled in this algebraic framework. 


\subsection{Difference fields and difference rings}\label{subsec:DiffFieldAndDiffRings}
A {\dr} (resp. field) $\dField{\AA}$ is a ring (resp. field) $\AA$ together with a ring (resp. field) automorphism $\s: \AA \rightarrow \AA$. Subsequently, all rings (resp. fields) are commutative with unity; in addition they contain the set of rational numbers $\QQ$, as a subring (resp. subfield). The multiplicative group of units of a ring (resp. field) $\AA$ is denoted by $\AA^{*}$.  A ring (resp. field) is computable if all of it's operations are computable. A difference ring (resp. field) $\dField{\AA}$ is computable if $\AA$ and $\s$ are both computable. Thus, given a computable difference ring (resp. field), one can decide if $\s(c) = c$. The set of all such elements for a given difference ring (resp. field) denoted by $$\const{\dField{\AA}}=\{c\in\AA\,|\,\s(c)=c\}$$ forms a subring (resp. subfield) of $\AA$. In this article, $\const{\dField{\AA}}$ will always be a field called the constant field of $\dField{\AA}$. Note that it contains $\QQ$ as a subfield. For any difference ring (resp. field) we shall denote the constant field by $\KK$.

The construction of difference rings/fields will be accomplished by a tower of difference ring/field extensions.
A {\dr} $\dField[\tilde{\s}]{\tilde{\AA}}$ is said to be a \emph{\drE} \index{\drE} of a {\dr} $\dField{\AA}$ if $\AA$ is a subring of $\tilde{\AA}$ and  for all $a \in \AA $, $\tilde{\s}(a) = \s(a)$ (i.e., $\tilde{\s} |_{\AA} = \s$). The definition of a \emph{\dfE} \index{\dfE} is the same by only replacing the word ring with field. In the following we do not distinguish anymore between $\s$ and $\tilde{\s}$. 

In the following we will consider two types of product extensions.
Let $\dField{\AA}$ be a {\dr} (in which products have already been defined by previous extensions). Let $\alpha \in \AA^{*}$ be a unit and consider the ring of Laurent polynomials $\AA[t, t^{-1}]$ (i.e., $t$ is transcendental over $\AA$). Then there is a unique {\drE} $\dField{\AA[t, t^{-1}]}$ of $\dField{\AA}$ with $\s(t)=\alpha\,t$ and $\s(t^{-1}) = \alpha^{-1}\,t^{-1}$. The extension here is called a \emph{product-extension} (in short \pE-extension) and the generator $t$ is called a \pE-monomial. 
Suppose that $\AA$ is a field and $\AA(t)$ is a rational function field (i.e., $t$ is transcendental over $\AA$). Let $\alpha \in \AA^{*}$. Then there is a unique \dfE\, $\dField{\AA(t)}$ of $\dField{\AA}$ with $\s(t)=\alpha\,t$. We call the extension a \emph{\pE-field extension} and $t$ a \pE-monomial. In addition, we get the chain of extensions $\dField{\AA} \le \dField{\AA[t,t^{-1}]} \le \dField{\AA(t)}$.

Furthermore, we consider extensions which model algebraic objects like $\zeta^{k}$ where $\zeta$ is a $\lambda$-th root of unity for some $\lambda \in \NN$ with $\lambda > 1$. Let $\dField{\AA}$ be a {\dr} and let $\zeta\in\AA^{*}$ be a primitive $\lambda$-th root of unity, (i.e., $\zeta^{\lambda}=1$ and $\lambda$ is minimal). Take the {\drE} $\dField{\AA[y]}$ of  $\dField{\AA}$ with $y$ being transcendental over $\AA$ and $\s(y) = \zeta\,y$. Note that this construction is also unique. Consider the ideal $I:=\geno{y^{\lambda}-1}$ and the quotient ring $\EE := \AA[y]/I$. Since $I$ is closed under $\s$ and $\s^{-1}$ i.e., $I$ is a reflexive difference ideal, we have a ring automorphism $\s : \EE \to \EE$ defined by $\s(h+I) = \s(h) + I$. In other words, $\dField{\EE}$ is a {\dr}. Note that by this construction the ring $\AA$ can naturally be embedded into the ring $\EE$ by identifying $a \in \AA$ with $a + I \in \EE$, i.e., $a \mapsto a + I$. Now set $\vartheta := y + I$. Then
$\dField{\AA[\vartheta]}$ is a {\drE} of $\dField{\AA}$ subject to the relations $\vartheta^{\lambda} = 1$ and $\s(\vartheta) = \zeta\,\vartheta$. This extension is called an algebraic extension (in short \aE-extension) of order $\lambda$. The generator, $\vartheta$ is called an \aE-monomial and we define 
$\lambda=\min\{n>0\,|\,\zeta^{n}=1\}$ as its order. Note that the \aE-monomial $\vartheta$, with the relations $\vartheta^{\lambda}=1$ and $\s(\vartheta) = \zeta\,\vartheta$ models $\zeta^{k}$ with the relations $(\zeta^{k})^{\lambda} = 1$ and $\zeta^{k+1}=\zeta\,\zeta^{k}$. In addition, the ring $\AA[\vartheta]$ is not an integral domain (i.e., it has zero-divisors) since $(\vartheta - 1)\,(\vartheta^{\lambda-1}+\cdots+\vartheta+1) = 0$ but $(\vartheta - 1) \neq 0 \neq (\vartheta^{\lambda-1}+\cdots+\vartheta+1)$. 

We introduce the following notations for convenience. Let $\dField{\EE}$ be a {\drE} of $\dField{\AA}$ with $t\in\EE$.  $\AA\langle t \rangle$ denotes the ring of Laurent polynomials $\AA[t,\tfrac{1}{t}]$ (i.e., $t$ is transcendental over $\AA$) if $\dField{\AA[t,\tfrac{1}{t}]}$ is a \pE-extension of $\dField{\AA}$. Lastly, $\AA\langle t \rangle$ denotes the ring $\AA[t]$ with $t\notin\AA$ but subject to the relation $t^{\lambda}=1$ if $\dField{\AA[t]}$ is an \aE-extension of $\dField{\AA}$ of order $\lambda$. 	We say that the {\drE} $\dField{\AA\langle t \rangle}$ of $\dField{\AA}$ is an \apE-extension (and $t$ is an \apE-monomial) if it is an \aE- or a \pE-extension. Finally, we call $\dField{\AA\langle t_{1} \rangle\dots\langle t_{e} \rangle}$ a (nested) \apE-extension/\pE-extension of $\dField{\AA}$ it is built by a tower of such extensions.

Throughout this article, we will restrict ourselves to the following classes of extensions as our base field.

\begin{mexample}\label{exa:RatDF}
	Let $\KK(x)$ be a rational function field and define the field automorphism $\s:\KK(x)\to\KK(x)$ with $\s(f)=f|_{x\,\mapsto\, x+1}$. We call $\dField{\KK(x)}$ the \emph{rational {\df}} over $\KK$.
\end{mexample}

\begin{mexample}\label{exa:qMixedDF}
        Let $\KK=K(q_1,\dots,q_e)$ be a rational function field (i.e., the $q_i$ are transcendental among each other over the field $K$ and let $\dField{\KK(x)}$ be the rational {\df} over $\KK$. Consider a \pE-extension $\dField{\EE}$ of $\dField{\KK(x)}$ with $\EE=\KK(x)[t_{1},\frac{1}{t_{1}}]\dots[t_{e},\frac{1}{t_{e}}]$ and $\s(t_{i})=q_{i}\,t_{i}$ for $1\le i\le e$. Now consider the field of fractions $\FF=Q(\EE) = \KK(x)(t_{1})\dots(t_{e})$. We also use the shortcut $\bs{t}=(t_1,\dots,t_e)$ and write $\FF=\KK(x)(\bs{t})=\KK(x,\bs{t})$.      
        Then $\dField{\FF}$ is a \pE-field extension of the {\df} $\dField{\KK(x)}$. It is also called the \emph{mixed ${\bs q}$-multibasic {\df}} over $\KK$. If $\FF = \KK(t_{1})\dots(t_{e})=\KK(\bs{t})$ which is free of $x$, then $\dField{\FF}$ is called the \emph{${\bs q}$-multibasic {\df}} over $\KK$. Finally, if $e = 1$, then $\FF = \KK(t_{1})$ and $\dField{\FF}$ is called a \emph{$q$-} or a \emph{basic {\df}} over $\KK$.
\end{mexample}

\noindent Based on these ground fields we will define now our products. In the first sections we will restrict to the hypergeometric case.

\begin{mexample}\label{exa:A-ExtInDiffRingSetting}
	Let $\KK=\QQ\big(\ii,\,(-1)^{\frac{1}{6}}\big)$ and let $\dField{\KK(x)}$ be a rational {\df}. Then the product expressions 
	\begin{equation}\label{eqn:prdtsExprsOverRootsOfUnits}
	\myProduct{k}{1}{n}{(-1)^{\frac{1}{6}}}, \quad \myProduct{k}{1}{n}{(-1)^{\frac{1}{2}}}
	\end{equation}
	from $\Prod(\KK(n))$ can be represented in an \aE-extension as follows. Here $\ii$ is the complex unit which we also write as $(-1)^{\frac{1}{2}}$. Now take the \aE-extension $\dField{\KK(x)[\vartheta_{1}]}$ of $\dField{\KK(x)}$ with $\s(\vartheta_{1}) = (-1)^{\frac{1}{6}}\,\vartheta_{1}$ of order $12$. The \aE-monomial $\vartheta_{1}$ models $\rootOfUnitySeq{6}{n}$ with the shift-behavior 
	$S_{n}\rootOfUnitySeq{6}{n} =\rootOfUnitySeq{6}{n+1}= (-1)^{\frac{1}{6}}\,\rootOfUnitySeq{6}{n}.$ 
	Further, $\dField{\KK(x)[\vartheta_{1}][\vartheta_{2}]}$ is also an \aE-extension of $\dField{\KK(x)[\vartheta_{1}]}$ with $\s(\vartheta_{2})=\ii\,\vartheta_{2}$ of order $4$. The generator $\vartheta_{2}$ models $(\ii)^{n}$ with $S_{n}(\ii)^{n}=(\ii)^{n+1}=\ii\,(\ii)^{n}$.
\end{mexample}	

\begin{mexample}\label{exa:P-ExtOverConstPolysInDiffRingSetting}
	The product expressions
	\begin{equation}\label{eqn:constHyperGeoPrdtsExprs}
	\myProduct{k}{1}{n}{\sqrt{13}},\quad \myProduct{k}{1}{n}{7}, \quad \myProduct{k}{1}{n}{169}
	\end{equation}
	from $\Prod(\KK(n))$ with $\KK=\QQ(\sqrt{13})$ are represented in a \pE-extension of the rational {\df} $\dField{\KK(x)}$ with $\s(x)=x+1$ as follows. 
	\begin{enumerate} 
		\item Consider the \pE-extension $\dField{\AA_{1}}$ of $\dField{\KK(x)}$ with $\AA_{1}=\KK(x)[y_{1},\tfrac{1}{y_{1}}]$,  $\s(y_{1}) = (\sqrt{13})\,y_{1}$ and $\s(\tfrac{1}{y_{1}}) = \frac{1}{\sqrt{13}}\,\tfrac{1}{y_{1}}$. In this ring, we can model polynomial expressions in $\algSeq{13}{n}$ and $\algSeq{13}{-n}$ with the shift behavior $S_{n}\algSeq{13}{n} = \sqrt{13}\,\algSeq{13}{n}$ and $S_{n}\frac{1}{(\sqrt{13})^{n}} = \frac{1}{\sqrt{13}}\,\frac{1}{(\sqrt{13})^{n}}$. Here, $\algSeq{13}{n}$ and $\frac{1}{(\sqrt{13})^{n}}$ are rephrased by $y_{1}$ and $\frac{1}{y_{1}}$, respectively.
		\item Constructing the \pE-extension $\dField{\AA_{2}}$ of $\dField{\AA_{1}}$ with $\AA_{2}=\AA_{1}[y_{2},\frac{1}{y_{2}}]$, $\s(y_{2}) = 7\,y_{2}$ and $\s(\frac{1}{y_{2}}) = \frac{1}{7}\,\frac{1}{y_{2}}$, we are able to model polynomial expressions in $\intSeq{7}{n}$ and $\intSeq{7}{-n}$ with the shift behavior $S_{n}\intSeq{7}{n} = 7\,\intSeq{7}{n}$ and $S_{n}\frac{1}{\intSeq{7}{n}} = \frac{1}{7}\,\frac{1}{\intSeq{7}{n}}$ by rephrasing ${\intSeq{7}{n}}$ and $\frac{1}{\intSeq{7}{n}}$ with $y_{2}$ and $\frac{1}{y_{2}}$, respectively.
		\item Introducing the \pE-extension $\dField{\AA_{3}}$ of $\dField{\AA_{2}}$ with $\AA_{3}=\AA_{2}[y_{3},\tfrac{1}{y_{3}}]$, $\s(y_{3}) = 169\,y_{3}$ and $\s(\frac{1}{y_{3}}) = \frac{1}{169}\,\frac{1}{y_{3}}$, one can model polynomial expressions in $\intSeq{\left(169\right)}{n}$ and $\intSeq{\left(169\right)}{-n}$ with the shift behavior $S_{n}\intSeq{(169)}{n} =169\,\intSeq{(169)}{n}$ and $S_{n}\frac{1}{\intSeq{(169)}{n}}= \frac{1}{169}\,\frac{1}{\intSeq{\left(169\right)}{n}}$ by rephrasing $\intSeq{\left(169\right)}{n}$ and $\intSeq{\left(169\right)}{-n}$ by $y_{3}$ and $\frac{1}{y_{3}}$, respectively. 
	\end{enumerate}
\end{mexample}

\begin{mexample}\label{exa:P-ExtOverNonConstPolysInDiffRingSetting}
	The hypergeometric product expressions
	\begin{equation}\label{eqn:nonConstHyperGeoPrdtsExprs}
	P_{1}(n) = \myProduct{k}{1}{n}{k}, \quad P_{2}(n) = \myProduct{k}{1}{n}{\big(k+2\big)}
	\end{equation}
	from $\Prod(\QQ(n))$ can be represented in a \pE-extension defined over the rational {\df} $\dField{\QQ(x)}$ in the following way. Take the \pE-extension $\dField{\QQ(x)[z_{1},\tfrac{1}{z_{1}}]}$ of $\dField{\QQ(x)}$ with $\s(z_{1})=(x+1)\,z_{1}$ and $\s(\tfrac{1}{z_{1}})=\frac{1}{(x+1)}\,\tfrac{1}{z_{1}}$. In this extension, one can model polynomial expressions in the product expression $P_{1}(n)$ with the shift behavior $S_{n}P_{1}(n) =(n+1)\,P_{1}(n)$ and $S_{n} \frac{1}{P_{1}(n)} = \frac{1}{(n+1)}\,\frac{1}{P_{1}(n)}$ by rephrasing $P_{1}(n)$ and $\frac{1}{P_{1}(n)}$ by $z_{1}$ and $\frac{1}{z_{1}}$. Finally, taking the \pE-extension $\dField{\QQ(x)[z_{1},\tfrac{1}{z_{1}}][z_{2},\tfrac{1}{z_{2}}]}$ of $\dField{\QQ(x)[z_{1},\tfrac{1}{z_{1}}]}$ with  $\s(z_{2})=(x+3)\,z_{2}$ and $\s(\tfrac{1}{z_{2}})=\frac{1}{(x+3)}\,\frac{1}{z_{2}}$, we can represent polynomial expressions in the product expression $P_{2}(n)$ with the shift-behavior  $S_{n}P_{2}(n)=(n+3)\,P_{2}(n)$ and $S_{n} \frac{1}{P_{2}(n)}=\frac{1}{(n+3)}\,\frac{1}{P_{2}(n)}$ by rephrasing $P_{2}(n)$ and $\frac{1}{P_{2}(n)}$ by $z_{2}$ and $\frac{1}{z_{2}}$, respectively. 
\end{mexample}

In order to solve \rpe{P} within the next sections, we rely on a refined construction of $AP$-extensions. More precisely, we will represent our products in \rpiE-extensions.

\begin{mdefinition}\label{defn:rpisiSE}
An \apE-extension (\aE- or \pE-extension) $\dField{\AA\langle t_{1} \rangle\dots\langle t_{e} \rangle}$ of $\dField{\AA}$ is called an \rpiE-extension (\rE- or \piE-extension) if $\const\dField{\AA\langle t_{1} \rangle\dots\langle t_{e} \rangle}=\const\dField{\AA}$. Depending on the type of extension, we call $t_i$ an \rE-/\piE-/\rpiE-monomial. Similarly, let $\AA$ be a field. Then we call a tower of \pE-extensions $\dField{\AA(t_1)\dots(t_e)}$ of $\dField{\AA}$ a \piE-extension if $\const\dField{\AA(t_1)\dots(t_e)}=\const\dField{\AA}$.
\end{mdefinition}

\noindent We concentrate mainly on product extensions and skip the sum part that has been mentioned in the introduction. Still, we need to handle the very special case of the rational difference field $\dField{\KK(x)}$ with $\sigma(x)=x+1$ or the mixed ${\bs q}$-multibasic version. Thus it will be convenient to introduce also the field version of \sigmaE-extensions~\cite{karr1981summation,schneider2001symbolic}.

\begin{mdefinition}
Let $\dField{\FF(t)}$ be a difference field extension of $\dField{\FF}$ with $t$ transcendental over $\FF$ and $\sigma(t)=t+\beta$ with $\beta\in\FF$. This extension is called a \sigmaE-extension\footnote{In~\cite{karr1981summation} also the more general case $\sigma(t)=\alpha\,t+\beta$ with $\alpha\in\FF^*$ is treated. In the following we restrict to the simpler case $\alpha=1$ which is less technical but general enough to cover all problems that we observed in practical problem solving.} if $\const\dField{\FF(t)}=\const\dField{\FF}$. In this case $t$ is also called a \sigmaE-monomial. $\dField{\FF(t)}$ is called a \pisiSE-extension of $\dField{\FF}$ if it is either a \piE- or a \sigmaE-extension. I.e., $t$ is transcendental over $\FF$, $\const\dField{\FF(t)}=\const\dField{\FF}$ and $t$ is a \piE-monomial ($\sigma(t)=\alpha\,t$ for some $\alpha\in\FF^*$) or $t$ is a \sigmaE-monomial ($\sigma(t)=t+\beta$ for some $\beta\in\FF$). $\dField{\FF(t_1)\dots(t_e)}$ is a \pisiSE-extension of $\dField{\FF}$ if it is a tower of \pisiSE-extensions.
\end{mdefinition}

\noindent Note that there exist criteria which can assist in the task to check if during the construction the constants remain unchanged. The reader should see \cite[Proof $3.16$, $3.22$ and  $3.9$]{schneider2016difference} for the proofs. For the field version, see also \cite{karr1981summation}.

\begin{mtheorem}\label{thm:rpiE-Criterion}
	Let $\dField{\AA}$ be a \dr. Then the following statements hold.
	\begin{enumerate} 
		\item Let $\dField{\AA[t,\tfrac{1}{t}]}$ be a \pE-extension of $\dField{\AA}$ with $\s(t)=\alpha\,t$ where $\alpha\in\AA^{*}$. Then this is a \piE-extension (i.e., $\const{\dField{\AA[t,\tfrac{1}{t}]}} = \const{\dField{\AA}}$) iff there are no $g\in\AA\sm\zs$ and $v\in\ZZ\sm\zs$ with $\s(g)=\alpha^{v}\,g$. 
		\item Let $\dField{\AA[\vartheta]}$ be an \aE-extension of $\dField{\AA}$ of order $\lambda>1$ with $\s(\vartheta)=\zeta\,\vartheta$ where $\zeta\in\AA^{*}$. Then this is an \rE-extension (i.e., $\const{\dField{\AA[\vartheta]}} = \const{\dField{\AA}}$) iff there are no $g\in\AA\sm\zs$ and $v\in\{1,\dots,\lambda-1\}$ with $\s(g)=\zeta^{v}\,g$. If it is an \rE-extension, $\zeta$ is a primitive $\lambda$th root of unity.
		\item Let $\AA$ be a field and let $\dField{\AA(t)}$ be a difference field extension of $\dField{\FF}$ with $t$ transcendental over $\FF$ and $\sigma(t)=t+\beta$ with $\beta\in\FF$. Then this is a \sigmaE-extension (i.e., $\const\dField{\FF(t)}=\const\dField{\FF}$) if there is no $g\in\FF$ with $\sigma(g)=g+\beta$.
	\end{enumerate} 
\end{mtheorem}

Concerning our base case difference fields (see Examples~\ref{exa:RatDF} and~\ref{exa:qMixedDF}) the following remarks are relevant. The rational difference field $\dField{\KK(x)}$ is a \sigmaE-extension of $\dField{\KK}$ by part~(3) of Theorem~\ref{thm:rpiE-Criterion} and using the fact that there is no $g\in\KK$ with $\sigma(g)=g+1$. Thus 
$\const\dField{\KK(x)}=\const\dField{\KK}=\KK$.  Furthermore, the mixed ${\bs q}$-multibasic difference field $\dField{\FF}$ with $\FF=\KK(x)(t_{1})\dots(t_{e})$ is a \piE-field extension of $\dField{\KK(x)}$. See Corollary~\ref{cor:MixedMultibasicDiffFieldAsPiExt} below. As a consequence, we have that $\const\dField{\FF}=\dField{\KK(x)}=\KK$. 

\noindent We give further examples and non-examples of \rpiE-extensions.

\begin{mexample}\label{exa:rpisiSE}
	\begin{enumerate} 
		\item In Example~\ref{exa:A-ExtInDiffRingSetting}, the \aE-extension $\dField{\KK(x)[\vartheta_{1}]}$ is an \rE-extension of $\dField{\KK(x)}$ of order $12$ since there are no $g\in\KK(x)^{*}$ and $v\in\{1,\dots,11\}$ with $\s(g)=\big((-1)^{\sfrac{1}{6}}\big)^{v}\,g$. However, the \aE-extension $\dField{\KK(x)[\vartheta_{1}][\vartheta_{2}]}$ is not an \rE-extension of $\dField{\KK(x)[\vartheta_{1}]}$ since with $g = \vartheta_{1}^{3}\in\KK(x)[\vartheta_{1}]$ and $v=1$, we have $\s(g)=\ii\,g$. In particular, we get $c = \vartheta_{1}^{3}\,\vartheta_{2}\in \const\dField{\KK(x)[\vartheta_{1}][\vartheta_{2}]}\sm\const\dField{\KK(x)[\vartheta_{1}]}$.
		\item The \pE-extension $\dField{\AA_{1}}$ of $\dField{\KK(x)}$ in Example~\ref{exa:P-ExtOverConstPolysInDiffRingSetting}(1) with $\s(y_{1})=\sqrt{13}\,y_{1}$ is a \piE-extension of $\dField{\KK(x)}$ as there are no $g\in\KK(x)^{*}$ and $v\in\ZZ\sm\zs$ with $\s(g)=\big(\sqrt{13}\big)^{v}\,g$. Similarly, the \pE-extension $\dField{\AA_{2}}$ in Example~\ref{exa:P-ExtOverConstPolysInDiffRingSetting}(2) with $\s(y_{2})=7\,y_{2}$ is also a \piE-extension of $\dField{\AA_{1}}$ since there does not exist a $g\in\AA_{1}\sm\zs$ and a $v\in\ZZ\sm\zs$ with $\s(g)=\intSeq{7}{v}\,g$. However, the \pE-extension $\dField{\AA_{3}}$ in part (3) of Example~\ref{exa:P-ExtOverConstPolysInDiffRingSetting} is not a \piE-extension of $\dField{\AA_{2}}$ since with $g=y_{1}^{4}\in\AA_{2}$ we have $\s(g)=(169)\,g$. In particular, $w=y_{1}^{-4}\,y_{3}\in\const\dField{\KK(x)[y_{1},\tfrac{1}{y_{1}}][y_{2},\tfrac{1}{y_{2}}][y_{3},\tfrac{1}{y_{3}}]}\sm\const\dField{\KK(x)[y_{1},\tfrac{1}{y_{1}}][y_{2},\tfrac{1}{y_{2}}]}$.
		\item Finally, in Example~\ref{exa:P-ExtOverNonConstPolysInDiffRingSetting} the \pE-extension $\dField{\QQ(x)[z_{1},\tfrac{1}{z_{1}}]}$ is a \piE-extension of $\dField{\QQ(x)}$ with $\s(z_{1})=(x+1)\,z_{1}$ since there are no $g\in\QQ(x)^{*}$ and $v\in\ZZ\sm\zs$ with $\s(g)=(x+1)^{v}\,g$. But the \pE-extension $\dField{\QQ(x)[z_{1},\tfrac{1}{z_{1}}][z_{2},\tfrac{1}{z_{2}}]}$ with $\s(z_{2})=(x+3)\,z_{2}$ is not a \piE-extension of $\dField{\QQ(x)[z_{1},\tfrac{1}{z_{1}}]}$ since with $g=(x+2)\,(x+1)\,z_{1}$ and $v=1$ we have $\s(g)=(x+3)\,g$. In particular, we get $c=\frac{g}{z_{2}}\in\const\dField{\QQ(x)\geno{z_{1}}\geno{z_{2}}}\sm\const\dField{\QQ(x)\geno{z_{1}}}$.
	\end{enumerate}
\end{mexample}

\noindent We remark that in~\cite{karr1981summation} and~\cite{schneider2016difference} algorithms have been developed that can carry out these checks if the already designed difference ring is built by properly chosen \rpiE-extensions. In this article we are more ambitious. We will construct \apE-extensions carefully such that they are automatically \rpiE-extensions and such that the products under consideration can be rephrased within these extensions straightforwardly. In this regard, we will utilize the following lemma.

\begin{lemma}\label{Lemma:RExtOverPiSi}
Let $\dField{\FF}$ be a \pisiSE-extension of $\dField{\KK}$ with $\const\dField{\KK}=\KK$. Then the \aE-extension $\dField{\FF[\vartheta]}$ of $\dField{\FF}$ with order $\lambda>1$ is an \rE-extension. 
\end{lemma}
\begin{proof}
By \cite[Lemma $3.5$]{karr1985theory} we have $\const{\dField[\s^{k}]{\FF}}=\const{\dField{\FF}}$ for all $k\in\NN\setminus\{0\}$. Thus with \cite[Proposition $2.20$]{schneider2017summation}, $\dField{\FF[\vartheta]}$ is an \rE-extension of $\dField{\FF}$.	 
\end{proof}

\subsection{Ring of sequences}\label{subsec:RingOfSeqs}

We will elaborate how \rpiE-extensions can be embedded into the difference ring of sequences~\cite{schneider2017summation}; compare also \cite{van2006galois}. Precisely this feature will enable us to handle condition (2) of~\ref{prob:ProblemRPE}.

Let $\KK$ be a field containing $\QQ$ as a subfield and let $\NN$ be the set of non-negative integers. We denote by $\ringOfSeqs$ the set of all sequences $\seqA{a}{}{n} = \geno{a(0),\,a(1),\,a(2),\,\dots\,}$ whose terms are in $\KK$. With component-wise addition and multiplication, $\ringOfSeqs$ forms a commutative ring. The field $\KK$ can be naturally embedded into $\ringOfSeqs$ as a subring, by identifying $c \in \KK$ with the constant sequence $\geno{c,c,c,\dots\,} \in \ringOfSeqs$. Following the construction in \cite[Sec. $8.2$]{petkovvsek1996b}, we turn the shift operator
\[
S : \begin{cases}
\ringOfSeqs & \rightarrow \ \ringOfSeqs \\
\geno{a(0),\,a(1),\,a(2),\,\dots\,} & \mapsto \ \geno{a(1),\,a(2),\,a(3),\,\dots\,}
\end{cases}
\]
into a ring automorphism by introducing an equivalence relation $\sim$ on sequences in $\ringOfSeqs$. Two sequences $\seqA{a}{}{n}$ and $\seqA{b}{}{n}$ are said to be equivalent if and only if there exists a natural number $\delta$ such that $a(n) = b(n)$ for all $n \ge \delta$. The set of equivalence classes form a ring again with component-wise addition and multiplication which we will denote by $\mathcal{S}(\KK)$. For simplicity, we denote the elements of $\ringOfEquivSeqs$ (also called germs) by the usual sequence notation as above. Now it is obvious that $S:\ringOfEquivSeqs \to \ringOfEquivSeqs$ is a ring automorphism. Therefore, $\dField[S]{\ringOfEquivSeqs}$ forms a {\dr} called the  (\emph{difference}) \emph{ring of sequences} (\emph{over} $\KK$). 
\begin{mexample}\label{exa:APS-ExtInSeqSetting}
	The hypergeom.\ products in~\eqref{eqn:prdtsExprsOverRootsOfUnits}, \eqref{eqn:constHyperGeoPrdtsExprs} and \eqref{eqn:nonConstHyperGeoPrdtsExprs} yield the sequences 	{\fontsize{9}{0}\selectfont 		
	$\mathcal{M}=\big\{ \funcSeqA{\rootOfUnitySeq{6}{n}}{n},\,\,
	\funcSeqA{\ii^{n}}{n},\,\,  \funcSeqA{\algSeq{13}{n}}{n},\, \funcSeqA{\intSeq{7}{n}}{n},\,\funcSeqA{\intSeq{\left(169\right)}{n}}{n},\,\, \langle{P_{1}(n)}\rangle_{n\ge 0},\,\, \langle{P_{2}(n)}\rangle_{n\ge 0}\big\}$}
	with $S\funcSeqA{a_{n}}{n} := \funcSeqA{S_{n}\,a_{n}}{n}$ for $a_{n}\in\mathcal{M}$.
\end{mexample}

\begin{mdefinition}\label{defn:DiffRingHomOrMon}
	Let $\dField{\AA}$ and $\dField[\sigma^{\prime}]{\AA^{\prime}}$ be two \dr s. We say that $\tau: \AA \to \AA^{\prime}$ is a \emph{{\dr} homomorphism} between the \dr s $\dField{\AA}$ and $\dField[\sigma^{\prime}]{\AA^{\prime}}$ if $\tau$ is a ring homomorphism and for all $f \in \AA$, $\tau(\s(f)) = \s^{\prime}(\tau(f))$. If $\tau$ is injective then it is called a \emph{{\dr} monomorphism} or a \emph{{\dr} embedding}. In this case $\dField{\tau(\AA)}$ is a sub-difference ring of $\dField[\sigma^{\prime}]{\AA^{\prime}}$ where $\dField{\AA}$ and $\dField{\tau(\AA)}$ are the same up to renaming with respect to $\tau$. If $\tau$ is a bijection, then it is a \emph{{\dr} isomorphism} and we say $\dField{\AA}$ and $\dField[\sigma^{\prime}]{\AA^{\prime}}$ are isomorphic.\\
	Let $\dField{\AA}$ be a {\dr} with constant field $\KK$. A {\dr} homomorphism (resp. monomorphism) $\tau:\AA \to \ringOfEquivSeqs$ is called $\KK$-\emph{homomorphism} (resp. -\emph{monomorphism}) if for all $c \in \KK$ we have that $\tau(c) = {\bs c} := \geno{c,c,c,\dots\,}$.
\end{mdefinition}

The following lemma is the key tool to embed {\dr s} constructed by \rpiE-extensions into the ring of sequences.

\begin{mlemma}\label{lem:injectiveHom}
	Let $\dField{\AA}$ be a {\dr} with constant field $\KK$. Then:
	\begin{enumerate} 
		\item The map $\tau:\AA\to\ringOfEquivSeqs$ is a $\KK$-homomorphism if and only if there is a map $\ev:\AA\times\NN\to\KK$ with $\tau(f) = \geno{\ev(f,0),\,\ev(f,1),\dots}$	for all $f\in\AA$ satisfying the following properties:
		\begin{enumerate} 
			\item for all $c \in \KK$, there is a natural number $\delta \geq 0$ such that
			\begin{align*}
			\forall\,n\geq\delta\,:\,\ev(c,n)=c;
			\end{align*}
			\item for all $f,g\in\AA$ there is a natural number $\delta \geq 0$ such that 
			\begin{align*}
			\forall\,n\geq\delta\,&:\,\ev(f\,g,n)=\ev(f,n)\,\ev(g,n),
			\\
			\forall\,n\geq\delta\,&:\,\ev(f+g,n)=\ev(f,n)+\ev(g,n);
			\end{align*}
			\item for all $f\in\AA$ and $i\in\ZZ$, there is a natural number $\delta \geq 0$ such that
			\begin{align*}
			\forall\,n\geq\delta\,:\,\ev(\s^{i}(f),n)=\ev(f,n+i).
			\end{align*}
		\end{enumerate}
		\item Let $\dField{\AA\langle t \rangle}$ be an \apE-extension of $\dField{\AA}$ with $\s(t)=\alpha\,t$ ($\alpha\in\AA^*$) and suppose that $\tau:\AA\to\ringOfEquivSeqs$ as given in part~(1) is a $\KK$-homomorphism.\\	
		Take some big enough $\delta\in\NN$ such that $\ev(\alpha,n)\neq0$ for all $n\ge\delta$. Further, take $u\in\KK^{*}$; if $t^{\lambda}=1$ for some $\lambda>1$, we further assume that $u^{\lambda}=1$ holds. Consider the map $\tau^{\prime}:\AA\langle t \rangle \to \ringOfEquivSeqs$ with $\tau(f)=\geno{\ev(f,n)}_{n\ge0}$ where the evaluation function $\ev^{\prime}:\AA\langle t \rangle\times\NN\to\KK$ is defined by 
		\[
		\ev^{\prime}(\sum_{i}h_{i}\,t^{i}, n)=\sum_{i}\ev(h_{i},n)\,\ev^{\prime}(t,n)^{i}
		\]
		with 
		\[
		\ev^{\prime}(t,n)=u\prod_{k=\delta}^{n}\ev(\alpha,k-1).
		\]
		Then $\tau$ is a $\KK$-homomorphism.
		\item If $\dField{\AA}$ is a field and $\dField{\EE}$ is a (nested) \rpiE-extension of $\dField{\AA}$, then any $\KK$-homomorphism $\tau: \EE\to \ringOfEquivSeqs$ is injective.
	\end{enumerate}
\end{mlemma}

\begin{proof} \
	\begin{enumerate} 
		\item The proof follows by \cite[Lemma $2.5.1$]{schneider2001symbolic}.
		\item The proof follows by \cite[Lemma $5.4$(1)]{schneider2017summation}.
		\item By \cite[Theorem $3.3$]{schneider2017summation}, $\dField{\EE}$ is simple that is, any ideal of $\EE$ which is closed under $\s$ is either $\EE$ or $\zs$ . Thus by \cite[Lemma $5.8$]{schneider2017summation} $\tau^{\prime}$ is injective.\qed 
	\end{enumerate}
\end{proof}

\noindent In this article, we will apply part $(2)$ of Lemma~\ref{lem:injectiveHom} iteratively. As base case, we will use the following {\df s} that can be embedded into the ring of sequences. 

\begin{mexample}\label{exa:DiffFieldOfRatSeqs}
	Take the rational {\df} $\dField{\KK(x)}$ over $\KK$ defined in Example~\ref{exa:RatDF} and consider the map
	$\tau : \KK(x) \to \ringOfEquivSeqs$ defined by $\tau(\frac{a}{b})=\langle\ev\big(\frac{a}{b}, n \big)\big)\rangle_{n \geq 0}$ with $a,b\in\KK[x]$ and $b\neq0$  where 
	\begin{equation}\label{Equ:EvRatDefinition}
	\ev\big(\tfrac{a}{b}, n\big) := 
			\begin{cases}
			0, & \text{if } b(n) = 0 \\
			\frac{a(n)}{b(n)}, & \text{if } b(n) \neq 0.
			\end{cases}\end{equation} 
	Then by Lemma~\ref{lem:injectiveHom}(1) it follows that $\tau:\KK(x)\to \ringOfEquivSeqs$ is a $\KK$-homomorphism. 
	We can define the function:
	\begin{equation}\label{eqn:hyperGeoShiftBoundedFxns}
	Z(p) = \max\big(\{k \in \NN\,|\,p(k)=0\}\big)+1 \ \text{ for any } p \in \KK[x]
	\end{equation}
	with $\max(\varnothing)=-1$. Now let $f=\tfrac{a}{b}\in\KK(x)$ where $a,b\in\KK[x]$, $b\neq0$. Since $a(x)$, $b(x)$ have only finitely many roots, it follows that $\tau(\frac{a}{b}) = {\bs 0}$ if and only if $\frac{a}{b}=0$. Hence $\ker(\tau)=\zs$ and thus $\tau$ is injective. Summarizing, we have constructed a $\KK$-embedding, $\tau:\KK(x)\to\ringOfEquivSeqs$ where the {\df} $\dField{\KK(x)}$ is identified in the {\dr} of $\KK$-sequences $\dField[S]{\ringOfEquivSeqs}$ as the {sub-\dr} of $\KK$-sequences $\dField[S]{\tau(\KK(x))}$. We call $\dField[S]{\tau(\KK(x))}$ the \emph{\df\, of rational sequences}\index{\df\, of rational sequences}.
\end{mexample}

\begin{mexample}\label{exa:DiffFieldOfqMixedSeqs}
	Take the mixed ${\bs q}$-multibasic {\df} $\dField{\FF}$ with $\FF = \KK(x,{\bs t})$ defined in Example~\ref{exa:qMixedDF}. Then, $\tau:\FF \to \ringOfEquivSeqs$ defined by $\tau(\frac{a}{b})=\langle\ev\big(\frac{a}{b}, n \big)\big\rangle_{n \geq 0}$ with $a,b\in\KK[x,{\bs t}]$ and $b\neq0$ where 
			\begin{equation}\label{Equ:EvqMixed}
			\ev\big(\tfrac{a}{b}, n\big) := 
			\begin{cases}
			0, & \text{if } b(n, {\bs q}^{n}) = 0 \\
			\frac{a(n, \, {\bs q}^{n})}{b(n, \, {\bs q}^{n})}, & \text{if } b(n, \, {\bs q}^{n}) \neq 0
			\end{cases} 
			\end{equation}	
	is a $\KK$-homomorphism. We define the function
	\begin{equation} \label{eqn:mixedHyperGeoShiftBoundedFxns}
	Z(p) = \max\big(\{k\in\NN\,|\,p(k,{\bs q}^{k})= 0\}\big)+1 \ \text{ for any } p \in \KK[x,{\bs t}]
	\end{equation}
	with $\max(\varnothing)=-1$. We will use the fact that this set of zeros is finite if $p\neq0$ and that $Z(p)$ can be computed; see~\cite[Sec. 3.2]{bauer1999multibasic}.  For any rational function, $f=\frac{g}{h} \in \FF\setminus\{0\}$ with $g,h\in\KK[x,{\bs t}]$, let $\delta = \max(\{Z(g), Z(h)\})\in\NN$. Then $f(n)\neq 0$ for all $n\ge \delta$ and thus $\tau(f)\neq{\bs 0}$. Hence $\ker(\tau)=\zs$ and thus $\tau$ is injective. In summary, we have constructed a $\KK$-embedding $\tau:\FF\to\ringOfEquivSeqs$ where the {\df} $\dField{\FF}$ is identified in $\dField[S]{\ringOfEquivSeqs}$ as $\dField[S]{\tau(\FF)}$ which we call the \emph{{\df} of mixed ${\bs q}$-multibasic rational sequences}\index{{\df} of mixed ${\bs q}$-multibasic rational sequences}. 
\end{mexample}

\section{Main result}\label{sec:mainResult}
We shall solve \ref{prob:ProblemRPE} algorithmically by proving the following main result in Theorem~\ref{thm:ProblemRMHPE}. Here the specialization $e=0$ covers the hypergeometric case. Similarly, taking ${\bs q}$-multibasic hypergeometric products in~\eqref{eqn:MixedHypergeoPdts} and suppressing $x$ yield the multibasic case. Further, setting $e=1$ provides the $q$-hypergeometric case.

\begin{mtheorem}\label{thm:ProblemRMHPE}
Let $\KK=K(\kappa_1,\dots,\kappa_u)(q_1,\dots,q_e)$ be a rational function field over a field $K$ and consider the mixed ${\bs q}$-multibasic hypergeometric products
		\begin{equation}\label{eqn:MixedHypergeoPdts}
		P_{1}(n) = \myProduct{k}{\ell_{1}}{n}{h_{1}(k, \bs{q}^{k})},\,\,\,\dots,\,\,\,P_{m}(n) = \myProduct{k}{\ell_{m}}{n}{h_{m}(k,\bs{q}^{k})}\ \in\,\Prod(\KK(n,{\bs q}^{n}))
		\end{equation}
		with $\ell_{i} \in \NN$ and $h_{i}(x,{\bs t})\in\KK(x,{\bs t})$ s.t.\ $h_{i}(k,{\bs q}^{k})$ has no pole and is non-zero for $k \geq \ell_{i}$.\\ 
		Then there exist irreducible monic polynomials $\Lst{f}{1}{s}\in\KK[x,{\bs t}]\sm\KK$, nonnegative integers $\ell'_1,\dots,\ell'_s$ and a finite algebraic field extension $K'$ of $K$ with
		a $\lambda$-th root of unity $\zeta\in K'$ and elements $\Lst{\alpha}{1}{w}\in {K^{\prime}}^{*}$ which are not roots of unity with the following properties.\\ 
		One can choose natural numbers $\mu_{i},\delta_{i}\in\NN$ for $1\le i \le m$, integers $u_{i,j}$ with $1\le i \le m$, $1\le j \le w$, integers $v_{i,j}$ with $1\le i \le m$, $1\le j \le s$ and rational functions $r_{i}\in\KK(x,{\bs t})^{*}$ for $1\le i \le m$ such that the following holds:
		\begin{enumerate} 
			\item For all $n \in \NN$ with $n \geq \delta_{i}$, 
			\fontsize{9.4}{0}\selectfont
				\begin{equation}\label{eqn:simpMixedHyperGeoPrdts}
				\hspace*{-0.6cm}P_{i}(n)=\big(\zeta^{n}\big)^{\mu_{i}} \big({\alpha^{n}_{1}}\big)^{u_{i,1}}\hspace*{-0.1cm}\cdots \big({\alpha^{n}_{w}}\big)^{u_{i,w}}{r_{i}(n, {\bs q}^{n})} \left( {\smashoperator{\prod_{k=\ell^{\prime}_{1}}^{n}} f_{1}(k, {\bs q}^{k})} \right)^{\hspace*{-0.1cm}{v_{i,1}}} \hspace*{-0.25cm}\cdots \left( {\smashoperator{\prod_{k=\ell^{\prime}_{s}}^{n}} f_{s}(k, {\bs q}^{k})} \right)^{\hspace*{-0.1cm}{v_{i,s}}}\hspace*{-0.2cm}.
				\end{equation}
			\normalsize
			\item The sequences with entries from the field $\KK'=K'(\kappa_1,\dots,\kappa_u)(q_1,\dots,q_e)$, 
			\begin{equation}\label{set:algIndepSeqs}
			\funcSeqA{\alpha_{1}^{n}}{n}, \dots, \funcSeqA{\alpha_{w}^{n}}{n}, \funcSeqA{\myProduct{k}{\ell_{1}^{\prime}}{n}{f_{1}(k,{\bs q}^{k})}}{n},\dots,\funcSeqA{\myProduct{k}{\ell_{s}^{\prime}}{n}{f_{s}(k,{\bs q}^{k})}}{n}, 
			\end{equation}
			are among each other algebraically independent over  $\tau\big(\KK^{\prime}(x,{\bs t})\big)\big[ \constSeqA{\zeta}{n} \big]$; here $\tau : \KK^{\prime}(x,{\bs t}) \to \ringOfEquivSeqs[\KK^{\prime}]$ is a difference ring monomorphism  where $\tau(\frac{a}{b})=\langle\ev\big(\frac{a}{b}, n \big)\big\rangle_{n \geq 0}$ for $a,b\in\KK^{\prime}[x,{\bs t}]$ is defined by~\eqref{Equ:EvqMixed}.
		\end{enumerate}
		If $K$ is a strongly $\s$-computable field (see Definition~\ref{defn:stonglySigmaComputable} below), then the components in \eqref{eqn:simpMixedHyperGeoPrdts} are computable.
\end{mtheorem}

Namely, Theorem~\ref{thm:ProblemRMHPE} provides a solution to \ref{prob:ProblemRPE} as follows.
Let $P(n)\in\ProdExpr(\KK(n,{\bs q}^{n}))$ be defined as in~\eqref{Equ:ProdEDef}
with $S \subseteq \ZZ^{m}$ finite, $a_{(\Lst{\nu}{1}{m})}(n)\in\KK(n,{\bs q}^n)$ and where the products $P_{i}(n)$ are given as in \eqref{eqn:MixedHypergeoPdts}.
Now assume that we have computed all the components as stated in Theorem~\ref{thm:ProblemRMHPE}. Then determine $\lambda\in\NN$ such that all $a_{(\Lst{n}{1}{m})}(n)$ have no pole for $n\ge\lambda$, and set $\delta=\max(\lambda,\Lst{\delta}{1}{m})$. Moreover, replace all $P_{i}$ with $1\leq i\leq m$ by their right-hand sides of \eqref{eqn:simpMixedHyperGeoPrdts} in the expression $P(n)$ yielding the expression $Q(n)\in\ProdExpr(\KK'(n,{\bs q}^{n}))$. Then by this construction we have $P(n)=Q(n)$ for all $n\ge\delta$. Furthermore, part~(2) of Theorem~\ref{thm:ProblemRMHPE} shows part~(2) of \ref{prob:ProblemRPE}.\\ 
Finally, we look at the zero-recognition statement of part~(3) of \ref{prob:ProblemRPE}. If $Q=0$, then $P(n)=0$ for all $n\geq\delta$ by part (1) of \ref{prob:ProblemRPE}. Conversely, if $P(n)=0$ for all $n$ from a certain point on, then also $Q(n)=0$ holds for all $n$ from a certain point on by part (1). 
Since the sequences~\eqref{set:algIndepSeqs} are algebraically independent over $\tau(\KK^{\prime}(x,{\bs t}))[\langle\zeta^n\rangle_{n\geq0}]$, the expression $Q(n)$ must be free of these products. 
Consider the mixed ${\bs q}$-multibasic difference field $\dField{\KK'(x,\bs{t})}$ and the \aE-extension $\dField{\KK'(x,\bs{t})[\vartheta]}$ of $\dField{\KK'(x,\bs{t})}$ of order $\lambda$ with $\sigma(\vartheta)=\zeta\,\vartheta$. 
By Corollary~\ref{cor:MixedMultibasicDiffFieldAsPiExt} below it follows that the mixed ${\bs q}$-multibasic difference field $\dField{\KK'(x,\bs{t})}$ is a \pisiSE-extension of $\dField{\KK'}$ with $\const\dField{\KK'}=\KK'$. Thus by Lemma~\ref{Lemma:RExtOverPiSi} it follows that the \aE-extension is an \rE-extension. In particular, it follows by Lemma~\ref{lem:injectiveHom} that the homomorphic extension of $\tau$ from $\dField{\KK'(x,\bs{t})}$ to $\dField{\KK'(x,\bs{t})[\vartheta]}$ with $\tau(\vartheta)=\langle\zeta^n\rangle_{n\geq0}$ is a $\KK'$-embedding. Since $Q(n)$ is a polynomial expression in $\zeta^n$ with coefficients from $\KK'(n,{\bs q}^{n})$ ($\zeta^n$ comes from~\eqref{eqn:simpMixedHyperGeoPrdts}), we can find an $h(x,\bs{t},\vartheta)\in\KK'(x,\bs{t})[\vartheta]$ such that the expression $Q(n)$ equals $h(n,\bs{q}^n,\zeta^n)$.
Further observe that $\tau(h)$ and the produced sequence of $Q(n)$ agree from a certain point on. Thus $\tau(h)=\bs{0}$ and since $\tau$ is a $\KK'$-embedding, $h=0$. Consequently, $Q(n)$ must be the zero-expression.

We will provide a proof (and an underlying algorithm) for Theorem~\ref{thm:ProblemRMHPE} by tackling the following subproblem formulated in the {\dr} setting.

\begin{svgraybox}
	\emph{Given} a mixed ${\bs q}$-multibasic {\df} $\dField{\FF}$ with $\FF=\KK(x)(t_{1})\dots(t_{e})$ where $\s(x)=x+1$ and $\s(t_{\ell})=q_{\ell}\,t_{\ell}$ for $1\le\ell\le e$; given $\Lst{h}{1}{m}\in\FF^{*}$. \emph{Find} an \rpiE-extension $\dField{\AA}$ of $\dField{\KK^{\prime}(x)(t_{1})\dots(t_{e})}$ where $\KK^{\prime}$ is an algebraic field extension of $\KK$ and $\Lst{g}{1}{m}\in\AA\sm\zs$ where $\s(g_{i})=\s(h_{i})\,g_{i}$ for $1\leq i\leq m$.
\end{svgraybox}

\noindent Namely, taking the special case $\FF=\KK(x)$ with $\s(x)=x+1$, we will tackle the above problem in Theorem~\ref{thm:RPiMonomsReps}, and we will derive the general case in Theorem~\ref{imptThm:MixedCaseCombinedPeriodZeroShiftCoprimeAndContentsAsRPIExtension}. Then based on the particular choice of the $g_i$ this will lead us directly to Theorem~\ref{thm:ProblemRMHPE}. 

\noindent We will now give a concrete example of the above strategy for hypergeometric products. An example for the mixed ${\bs q}$-multibasic situation is given in Example~\ref{exa:MixedHyperGeoPrdtsInSeqSetting}.

\begin{mexample}\label{exa:ConstructRPiExt}
Take the rational function field $\KK=K(\kappa)$ defined over the algebraic number field $K=\QQ\big((\ii+\sqrt{3}),\sqrt{-13}\big)$ and take the rational function field $\KK(x)$ defined over $\KK$. Now consider the hypergeometric product expressions
	\begin{equation}\label{eqn:HyperGeoPrdtsInDiffRing}
	P(n) =  
	\myProduct{k}{1}{n}{h_{1}(k)}\, + \,
	\myProduct{k}{1}{n}{h_{2}(k)}\, + \\ 
	\myProduct{k}{1}{n}{h_3(k)}\in\ProdExpr(\KK(n))
	\end{equation}
with
\begin{equation}\label{Equ:RatXh}
h_1(x)=\tfrac{-13\,\sqrt{-13}\,\kappa}{x},\,
h_2(x)=\tfrac{-784\,(\kappa+1)^{2}\,x}{13\,\sqrt{-13}\,(\ii+\sqrt{3})^{4}\,\kappa\,(x+2)^{2}},\,
h_3(x)=\tfrac{-17210368\,(\kappa+1)^{5}\,x}{13\,\sqrt{-13} \,(\ii + \sqrt{3})^{10}\,\kappa\,(x+2)^{5}}
\end{equation}
where $h_1,h_2,h_3\in\KK(x)$. With our algorithm (see
	Theorem~\ref{imptThm:CombinedShiftCoprimeAndContentsAsRPIExtension} below) we construct the algebraic field extension $K'=\QQ((-1)^{\frac{1}{6}},\sqrt{13})$ of $K$, take the rational function field $\KK^{\prime}=K'(\kappa)$ and define on top the rational {\df} $\dField{\KK'(x)}$ with $\s(x)=x+1$. Based on this, we obtain the \rpiE-extension $\dField{\AA}$ of $\dField{\KK'(x)}$ with
	\begin{equation}\label{alg:RPiExtension}
		\AA=\KK^{\prime}(x)[\vartheta][y_{1}, y^{-1}_{1}][{y_{2}}, y^{-1}_{2}][y_{3}, y^{-1}_{3}][y_{4}, y^{-1}_{4}][z, z^{-1}]
	\end{equation}
	and the automorphism $\s: \AA \rightarrow \AA$ defined by $\s(\vartheta)=(-1)^{\frac{1}{6}}\vartheta$, $\s(y_{1})=\sqrt{13}\,y_{1}$, $\s(y_{2})=7\,y_{2}$, $\s(y_{3})=\kappa\,y_{3}$, $\s(y_{4})=(\kappa+1)\,y_{4}$, and $\s(z)=(x+1)\,z$; note that $\const{\dField{\AA}}=\KK^{\prime}$. Now consider the difference ring homomorphism $\tau:\AA\to\ringOfEquivSeqs[\KK^{\prime}]$ which we define as follows. For the base field $\dField{\KK^{\prime}(x)}$ we take the difference ring embedding  $\tau(\frac{a}{b})=\langle\ev\big(\frac{a}{b}, n \big)\rangle_{n \geq 0}$ for  $a,b\in\KK^{\prime}[x]$ where $\ev$ is defined in~\eqref{Equ:EvRatDefinition}. Further, applying iteratively part~(2) of Lemma~\ref{lem:injectiveHom} we obtain the difference ring homomorphism $\tau:\AA\to\ringOfEquivSeqs[\KK^{\prime}]$ determined by $\tau(\vartheta)=\langle\rootOfUnitySeq{6}{n} \rangle_{n\geq0}$, $\tau(y_{1})=\langle\algSeq{13}{n} \rangle_{n\geq0}$, $\tau(y_{2})=\langle\intSeq{7}{n} \rangle_{n\geq0}$,  $\tau(y_{3})=\langle\kappa^{n}\rangle_{n\geq0}$,$\tau(y_{4})=\langle(\kappa+1)^{n}\rangle_{n\geq0}$ and $\tau(z)=\langle n! \rangle_{n\geq0}$.
	In addition, since $\dField{\AA}$ is an \rpiE-extension of $\dField{\KK^{\prime}(x)}$, it follows by part (3) of Lemma~\ref{lem:injectiveHom} that $\tau$ is a $\KK^{\prime}$-embedding. 
	This implies that  
	$\tau(\KK^{\prime}(x))[\tau(\vartheta)][\tau(y_{1}),\tau(y_{1}^{-1})]\dots[\tau(y_{4}),\tau(y_{4}^{-1})][\tau(z),\tau(z^{-1})]$	is a Laurent polynomial ring over the ring $\tau(\KK^{\prime}(x))[\tau(\vartheta)]$ with $\tau(\vartheta) = \constSeqA[\delta]{\big((-1)^{\frac{1}{6}}\big)}{\hspace*{-0.06cm}n}$.
	Further, we find
		\begin{equation}\label{eqn:SimpHyperGeoInDiffRing}
		Q^{\prime}=\underbrace{\frac{\vartheta^{9}\,y_{1}^{3}\,y_{3}}{z}}_{=:\,g_{1}} + 4\,\underbrace{\frac{\vartheta^{11}\,y_{2}^{2}\,y_{4}^{2}}{(x+1)^{2}\,(x+2)^{2}\,y_{1}^{3}\,y_{3}\,z}}_{=:\,g_{2}} + 32\,\underbrace{\frac{\vartheta^{5}\,y_{2}^{5}\,y_{4}^{5}}{(x+1)^{5}\,(x+2)^{5}\,y_{1}^{3}\,y_{3}\,z^{4}}}_{=:\,g_{3}}
		\end{equation}	
	\normalsize 
	where $\s(g_{i}) = \s(h_{i})\,g_{i}$ for $i=1,2,3$. Thus the $g_{i}$ model the shift behaviors of the hypergeometric products with the multiplicands $h_{i}\in\KK(x)$. In particular, we have defined $Q'$ such that $\tau(Q')=\langle P(n)\rangle_{n\geq 0}$ holds.		
	Rephrasing  $x \leftrightarrow n$, $\vartheta \leftrightarrow \rootOfUnitySeq{6}{n}$, $y_{1} \leftrightarrow \algSeq{13}{n}$, $y_{2} \leftrightarrow \intSeq{7}{n}$, $y_{3} \leftrightarrow \intSeq{\kappa}{n}$, $y_{4} \leftrightarrow \polySeq{\kappa+1}{n}$and $ z \leftrightarrow  n!$ in  \eqref{eqn:SimpHyperGeoInDiffRing} we get
	\fontsize{9.3}{0}\selectfont
 		\begin{multline}\label{eqn:SimpHyperGeoPrdtsSeqSetting}
		\hspace*{-0.3cm}Q(n) = \frac{\rootOfUnitySeqExp{6}{n}{9}\,\algSeqExp{13}{n}{3}\,\kappa^{n}}{n!} + \frac{4\,\rootOfUnitySeqExp{6}{n}{11}\,\intSeqExp{7}{n}{2}\,\polySeqExp{\kappa+1}{n}{2}}{(n+1)^{2}\,(n+2)^{2}\,\algSeqExp{13}{n}{3}\,\kappa^{n}\,n!} \\ + \frac{32\,\rootOfUnitySeqExp{6}{n}{5}\,\intSeqExp{7}{n}{5}\,\polySeqExp{\kappa+1}{n}{5}}{(n+1)^{5}\,(n+2)^{5}\,\algSeqExp{13}{n}{3}\,\kappa^{n}\,\polySeq{n!}{4}}\in\ProdExpr(\KK(n)).
		\end{multline}	
     \normalsize
     Note: $n!$ and $a^n$ with $a\in\KK'^*$ are just shortcuts for $\prod_{k=1}^nk$ and $\prod_{k=1}^na$, respectively.
Based on the corresponding proof of Theorem~\ref{thm:ProblemRMHPE} at the end of Subsection~\ref{SubSec:HypergeoemtricCase} we can ensure that $P(n) = Q(n)$ holds for all $n \in \NN$ with $n\ge1$. Further details on the computation steps can be found in Examples~\ref{exa:SimpNonConstHyperGeoTermsInDiffRing} and~\ref{exa:constRPiExts} below.
 
	\end{mexample}

\section{Algorithmic preliminaries: strongly \texorpdfstring{$\sigma$}{sigma}-computable fields}\label{Sec:AlgProperties}

In Karr's algorithm~\cite{karr1981summation} and all the improvements~\cite{KS:06,Schneider:07d,Schneider:08c,Petkov:10,Schneider:15,schneider2016difference,schneider2017summation} one relies on certain algorithmic properties of the constant field $\KK$. Among those, one needs to solve the following problem.

\begin{ProblemSpecBox}[\gop{P}]{ 
		{\gop{P} for $\Lst{\alpha}{1}{w} \in K^{*}$} 
	}\label{prob:ProblemGO}
	{
		Given a field $K$ and $\Lst{\alpha}{1}{w} \in K^{*}$. Compute a basis of the submodule \vspace*{-0.1cm}
		\[
		\VV := \big\{(\Lst{u}{1}{w})\in\ZZ^{w}\,\Big|\,\prodLst{i}{1}{w}{\alpha}{u}=1\big\} \text{ of } \ZZ^{w} \text{ over } \ZZ. 
		\]
		\vspace*{-0.7cm}} 
\end{ProblemSpecBox}

\noindent In~\cite{schneider2005product} it has been worked out that \ref{prob:ProblemGO} is solvable in any rational function field $\KK=K(\kappa_1,\dots,\kappa_u)$ provided that one can solve \ref{prob:ProblemGO} in $K$ and that one can factor multivariate polynomials over $K$.
In this article we require the following stronger assumption: \ref{prob:ProblemGO} can be solved not only in  $K$ ($K$ with this property was called $\sigma$-computable in~\cite{schneider2005product,KS:06}) but also in any algebraic extension of it. 
\begin{mdefinition}\label{defn:stonglySigmaComputable}
A field $K$ is strongly $\sigma$-computable if the standard operations in $K$ can be performed, multivariate polynomials can be factored over $K$ and \ref{prob:ProblemGO} can be solved for $K$ and any finite algebraic field extension of $K$.
\end{mdefinition}

\noindent 

\noindent Note that Ge's algorithm~\cite{ge1993algorithms} solves \ref{prob:ProblemGO} over an arbitrary number field $K$. 
Since any finite algebraic extension of an algebraic number field is again an algebraic number field, it follows with Ge's algorithm, that any number field $K$ is $\sigma$-computable. 

Summarizing, in this article we can turn our theoretical results to algorithmic versions, if we assume that $\KK=K(\kappa_1,\dots,\kappa_u)$ is a rational function field over a field $K$ which is strongly $\sigma$-computable. In particular, the underlying algorithms are implemented in the package \texttt{NestedProducts} for the case that $K$ is a finite algebraic field extension of $\QQ$. 

Besides these fundamental properties of the constant field, we rely on further (algorithmic) properties that can be ensured by difference ring theory. Let $\dField{\FF[t]}$ be a {\dr} over the field $\FF$ with $t$ transcendental over $\FF$ and  $\s(t)=\alpha\,t+\beta$ where $\alpha\in\FF^{*}$ and $\beta\in\FF$. Note that for any $h\in\FF[t]$ and any $k\in\ZZ$ we have $\sigma^k(h)\in\FF[t]$. Furthermore, if $h$ is irreducible, then also $\s^{k}(h)$ is irreducible.\\
Two polynomials $f,\,h\in\FF[t]\sm\zs$ are said to be \emph{shift co-prime}\index{shift co-prime}, also denoted by $\shp(f,h)=1$, if for all $k\in\ZZ$ we have that $\gcd(f,\s^{k}(h))=1$. Furthermore, we say that $f$ and $h$ are shift-equivalent, 
denoted by $f\sim_{\s}h$, if there is a $k\in\ZZ$ with $\frac{\sigma^k(f)}{h}\in\FF$. If there is no such $k$, then we also write $f\nsim_{\s}h$.

It is immediate that $\sim_{\s}$ is an equivalence relation. In the following we will focus mainly on irreducible polynomials $f,h\in\FF[t]$. Then observe that $f\sim_{\s}h$ holds if and only if $\shp(f,h)\neq1$ holds. In the following it will be important to determine such a $k$. Here
we utilize the following property of \pisiSE-extensions whose proof can be found in~\cite[Thm.~4]{karr1981summation} (~\cite[Cor.~1,2]{bronstein2000solutions}
or~\cite[Thm.~2.2.4]{schneider2001symbolic}).

\begin{mlemma}\label{Lemma:Period}
Let $\dField{\FF(t)}$ be a \pisiSE-extension of $\dField{\FF}$ and $f\in\FF(t)^*$. Then $\frac{\sigma^k(f)}{f}\in\FF$ for some $k\neq0$ iff $\frac{\sigma(t)}{t}\in\FF$  and $f=u\,t^m$ with $u\in\FF^*$ and $m\in\ZZ$. 
\end{mlemma}

\noindent Namely, using this result one can deduce when such a $k$ is unique.

\begin{mlemma}\label{lem:KarrsSpecification}
	Let $\dField{\FF(t)}$ be a \pisiSE-extension of $\dField{\FF}$ with $\s(t)=\alpha\,t+\beta$ for $\alpha\in\FF^{*}$ and $\beta\in\FF$. Let $f,h\in\FF[t]$ be irreducible with $f\,\sim_{\s}\,h$. Then there is a unique $k\in\ZZ$ with $\frac{\s^{k}(f)}{h}\in\FF^{*}$ iff $\frac{\sigma(t)}{t}\notin\FF$ or $f=a\,t$ and $h=b\,t$ for some $a,b\in\FF^*$.
\end{mlemma}

\begin{proof}
	Suppose on the contrary that $\beta=0$ and $f=t=h$. Then $\frac{\s^{k}(f)}{h}\in\FF^{*}$ for all $k\in\ZZ$ and thus $k$ is not unique. Conversely, suppose that $\s^{k_{1}}(f)=u\,h$ and $\s^{k_{2}}(f)=v\,h$ with $k_{1}>k_{2}$. Then $\frac{\s^{k_{1}-k_{2}}(f)}{f}=\frac{u}{v}\in\FF^{*}$. Thus by Lemma~\ref{Lemma:Period}, $\frac{\sigma(t)}{t}\in\FF$ and $f=a\,t$ for some $a\in\FF^*$. Thus also $h=b\,t$ for some $b\in\FF^*$. \qed
\end{proof}

Consider the rational {\df} $\dField{\KK(x)}$ with $\s(x)=x+1$. Note that $x$ is a \sigmaE-monomial. Let $f, h\in\KK[x]\sm\KK$ be  irreducible polynomials. If $f\,\sim_{\s}\,h$, then there is a unique $k\in\ZZ$ with $\frac{\s^{k}(f)}{h}\in\KK$. Similarly for the mixed ${\bs q}$-multibasic difference field $\dField{\KK(x)(t_{1})\dots(t_{e})}$ with $\s(x)=x+1$ and $\s(t_{i})=q_{i}\,t_{i}$ for $1\le i\le e$ we note that the $t_i$ are \piE-monomials; see Corollary~\ref{cor:MixedMultibasicDiffFieldAsPiExt} below. For $1\leq i\leq e$ and $\EE=\KK(x)(t_{1})\dots(t_{i-1})$, let  $f, h\in\EE[t_{i}]$ be monic irreducible polynomials. If $f\,\sim_{\s}\,h$, then there is a unique $k\in\ZZ$ with $\frac{\s^{k}(f)}{h}\in\EE$ if and only $f\neq t_{i}\neq h$. In both cases, such a unique $k$ can be computed if one can perform the usual operations in $\KK$; \cite[Thm.~$1$]{KS:06}. Optimized algorithms for theses cases can be found in \cite[Section~$3$]{bauer1999multibasic}. In addition, the function $Z$ given in~\eqref{eqn:hyperGeoShiftBoundedFxns} or in~\eqref{eqn:mixedHyperGeoShiftBoundedFxns} can be computed due to~\cite{bauer1999multibasic}.
Summarizing, the following properties hold.

\begin{mlemma}
Let $\dField{\FF}$ be the rational or mixed $\textbf{q}$-multibasic difference field over $\KK$ as defined in Examples~\ref{exa:RatDF} and~\ref{exa:qMixedDF}. Suppose that the usual operations\footnote{This is the case if $\KK$ is strongly $\sigma$-computable, or if $\KK$ is a rational function field over a strongly $\sigma$-computable field.} in $\KK$ are computable. Then one compute
\begin{enumerate}
 \item the $Z$-functions given in~\eqref{eqn:hyperGeoShiftBoundedFxns} or in~\eqref{eqn:mixedHyperGeoShiftBoundedFxns};
 \item one can compute for shift-equivalent irreducible polynomials $f,h$ in $\KK[x]$ (or in $\KK(x)(t_1,\dots,t_{i-1})[t_i]$) a $k\in\ZZ$ with $\frac{\sigma^k(f)}{h}\in\KK$ (or $\frac{\sigma^k(f)}{h}\in\KK(x)(t_1,\dots,t_{i-1})$).
\end{enumerate}
\end{mlemma}

For further considerations, we introduce the following Lemma which gives a relation between two polynomials that are shift-equivalent. 
\begin{mlemma}\label{lem:ConstructShiftEquivPoly}
	Let $\dField{\FF(t)}$ be a {\df} over a field $\FF$ with $t$ transcendental over $\FF$ and $\s(t)=\alpha\,t+\beta$ where $\alpha\in\FF^{*}$ and $\beta\in\FF$. Let $f,\,h\in\FF[t]\sm\FF$ be monic and $f\sim_{\s}h$. Then there is a $g\in\FF(t)^{*}$ with $h= \frac{\s(g)}{g}\,f$.
\end{mlemma}

\begin{proof}
	Since $f\,\sim_{\s}\,h$, there is a $k\in\ZZ$ and $u\in\FF^{*}$ with $\s^{k}(f)=u\,h$. Note that $\deg(f)=\deg(h)=m$. By comparing coefficients of the leading terms and using that $f,h$ are monic, we get $u\,t^m=\s^{k}(t^{m})$. If $k\ge0$, set $g:=\prod_{i=0}^{k-1}\s^{i}(t^{-m}f)$. Then $\frac{\s(g)}{g}=\frac{\s^{k}(t^{-m}f)}{t^{-m}f}=\frac{\s^{k}(f)\,t^{m}}{f\,\s^{k}(t^{m})}=\frac{h\,u\,t^{m}}{f\,\s^{k}(t^{m})}=\frac{h}{f}$. Thus, $h= \frac{\s(g)}{g}\,f$. If $k<0$, set $g:=\prod_{i=1}^{-k}\s^{-i}(\frac{t^{m}}{f})$. Then $\frac{\s(g)}{g}=\frac{\s^{k}(t^{m}f^{-1})}{t^{m}f^{-1}}=\frac{t^{m}\,\s^{k}(f)}{\s^{k}(t^{m})\,f}=\frac{h\,u\,t^{m}}{f\,\s^{k}(t^{m})}=\frac{h}{f}$. Hence again $h= \frac{\s(g)}{g}\,f$. \qed
\end{proof}



\section{Algorithmic construction of \texorpdfstring{\rpiE}{RPi}-extensions for \texorpdfstring{$\ProdExpr(\KK(n))$}{ProdE(KK(n))}}\label{Sec:HypergeometricCase}

In this section we will provide a proof for Theorem~\ref{thm:ProblemRMHPE} for the case $\ProdExpr(\KK(n))$. Afterwards, this proof strategy will be generalized for the case $\ProdExpr(\KK(n, {\bs q}^{n}))$ in Section~\ref{Sec:mixedHypergeometricCase}. In both cases, we will need the following set from~\cite[Definition $21$]{karr1981summation}.
\begin{mdefinition}\label{defn:MKarrZModule}
	For a difference field $\dField{\FF}$ and ${\bs f} = (\Lst{f}{1}{s})\in(\FF^{*})^{s}$ we define 
	\[
	\KarrModule{{\bs f}}{\FF} = \left\{(\Lst{v}{1}{s})\in\ZZ^{s}\,\big|\,\tfrac{\s(g)}{g}=\ProdLst{f}{1}{s}{v}\,\text{ for some } g\in\FF^{*}\right\}. 
	\]
\end{mdefinition} 

\noindent Note that $\KarrModule{{\bs f}}{\FF}$ is a $\ZZ$-submodule of $\ZZ^{s}$ which has finite rank.
We observe further that for the special case $\const\dField{\AA}=\AA$ we have $\frac{\sigma(g)}{g}=1$ for all $g\in\AA^*$. Thus
\[
	\KarrModule{{\bs f}}{\AA} = \{(\Lst{v}{1}{s})\in\ZZ^{s}\,\vert\,\ProdLst{f}{1}{s}{v}=1\}
\]
which is nothing else but the set in~\ref{prob:ProblemGO}.

\noindent Finally, we will heavily rely on the following lemma that ensures if a \pE-extension forms a \piE-extension; compare also~\cite{Singer:08}.

\begin{mlemma}\label{lem:transcendentalCriterionForPrdts}
	Let $\dField{\FF}$ be a {\df} and let ${\bs f} = (\Lst{f}{1}{s})\in{(\FF^{*})}^{s}$. Then the following statements are equivalent. 
	\begin{enumerate} 
		\item There are no $(\Lst{v}{1}{s}) \in \ZZ^{s}\sm\zvs{s}$ and $g\in\FF^{*}$ with \eqref{eqn:multTelescoping}, i.e., $\KarrModule{{\bs f}}{\FF}=\zvs{s}$.
		\item One can construct a \piE-field extension $\dField{\FF(z_{1})\dots(z_{s})}$ of $\dField{\FF}$ with $\s(z_{i})=f_{i}\,z_{i}$, for $1\le i\le s$.
		\item One can construct a \piE-extension $\dField{\FF[z_{1},z^{-1}_{1}]\dots[z_{s},z^{-1}_{s}]}$ of $\dField{\FF}$ with $\s(z_{i})=f_{i}\,z_{i}$, for $1\le i\le s$.
	\end{enumerate}
\end{mlemma}

\begin{proof}
	$(1)\Leftrightarrow(2)$ is established by \cite[Theorem 9.1]{schneider2010parameterized}.  $(2) \Longrightarrow (3)$ is obvious while $(3) \Longrightarrow (2)$ follows by iterative application of \cite[Corollary $2.6$]{schneider2017summation}. \qed
\end{proof}

Throughout this section, let $\dField{\KK(x)}$ be the rational difference field over a constant field $\KK$, where $\KK=K(\kappa_1,\dots,\kappa_u)$ is a rational function field over a field $K$. For algorithmic reasons we will assume in addition that $K$ is strongly $\s$-computable (see Definition~\ref{Sec:AlgProperties}).
In Subsection~\ref{Subsec:K} we will treat Theorem~\ref{thm:ProblemRMHPE} first for the special case $\ProdExpr(K)$. Next, we treat the case $\ProdExpr(\KK)$ in Subsection~\ref{Sec:ConstantRat}.
In Subsection~\ref{Sec:NonConstantRat} we present simple criteria to check if a tower of \piE-monomials $t_i$ with $\sigma(t_i)/t_i\in\KK[x]$ forms a \piE-extension. Finally, in Subsection~\ref{SubSec:HypergeoemtricCase} we will utilize this extra knowledge to construct \piE-extensions for the full case $\ProdExpr(\KK(n))$.

\subsection{Construction of \texorpdfstring{\rpiE}{RPi}-extensions for \texorpdfstring{$\ProdExpr(K)$}{ProdE(K)}}\label{Subsec:K}

Our construction is based on the following theorem.

\begin{mtheorem} \label{thm:RPiMonomsReps}
	Let $\Lst{\gamma}{1}{s} \in K^{*}$. Then there is an algebraic field extension $K^{\prime}$ of $K$ together with a $\lambda$-th root of unity $\zeta\in K^{\prime}$ and elements ${\bs \alpha} = (\Lst{\alpha}{1}{w}) \in {K^{\prime}}^{w}$ with $\KarrModule{{\bs \alpha}}{K^{\prime}} = \zvs{w}$ such that for all $i=1,\dots,s$,
	\begin{equation}\label{eqn:rootOfUnityAndNoIntRelsAlgNums}
	\gamma_{i}=\zeta^{\mu_{i}}\,\alpha_{1}^{u_{i,{1}}}\cdots\alpha_{w}^{u_{i,{w}}}
	\end{equation}
	holds for some $1 \leq\mu_{i} \leq \lambda$ and $(u_{i,{1}},\dots,u_{i,{w}})\in~\ZZ^{w}$.\\ 
	If $K$ is strongly $\sigma$-computable, then $\zeta$, the $\alpha_i$ and the $\mu_i,u_{i,j}$ can be computed.
\end{mtheorem}

\begin{proof}
	We prove the Theorem by induction on $s$. The base case $s=0$ obviously holds. Now assume that there are a $\lambda$-th root of unity $\zeta$, elements ${\bs \alpha}=(\Lst{\alpha}{1}{w}) \in ({K^{\prime}}^{*})^{w}$ with $\KarrModule{{\bs \alpha}}{K^{\prime}}=\zvs{w}$, $1 \leq \mu_{i}\leq \lambda$ and $(v_{i,{1}}, \dots v_{i,{w}}) \in \ZZ^{w}$ such that $\gamma_{i} = \zeta^{\mu_{i}}\,\alpha_{1}^{v_{i,{1}}}\cdots\alpha_{w}^{v_{i,{w}}}$	holds for all $1\leq i\leq s-1$.\\ 
	Now consider in addition $\gamma_s\in K^*$.
	First suppose the case $\KarrModule{(\Lst{\alpha}{1}{w}, \gamma_{s})}{K^{\prime}}=\zvs{w+1}$. With $\alpha_{w+1}:=\gamma_{s}$, we can write $\gamma_{s}$ as $\gamma_{s} = \zeta^{\lambda}\,\ProdLst{\alpha}{1}{w}{\upsilon}\,\alpha_{w+1}$ with $\lambda=\upsilon_{1}=\cdots=\upsilon_{w}=0$. Further, with $v_{i,w+1}=0$, we can write $\gamma_{i} = \zeta^{\mu_{i}}\,\alpha_{1}^{v_{i,{1}}}\cdots\alpha_{w}^{v_{i,{w}}}\,\alpha_{w+1}^{v_{i,w+1}}$ for all $1 \leq i \leq s-1$. This completes the proof for this case.\\ 
	Otherwise, suppose that $\KarrModule{(\Lst{\alpha}{1}{w},\gamma_{s})}{K^{\prime}}\neq\zvs{w+1}$ and take $(\Lsth{\upsilon}{1}{w}{u_{s}})\in\KarrModule{(\Lst{\alpha}{1}{w},\gamma_{s})}{K^{\prime}}\sm\zvs{w+1}$. Note that $u_{s}\neq0$ since $\KarrModule{{\bs \alpha}}{K^{\prime}}=\zvs{w}$. Then take all the non-zero integers in $(\Lsth{\upsilon}{1}{w}{u_{s}})$ and define $\delta$ to be their least common multiple. Define $\bar{\alpha}_{j}:= \alpha^{\sfrac{1}{|u_{s}|}}_{j}\in K^{\prime\prime}$ for $1\leq j \leq w$ where $K^{\prime\prime}$ is some algebraic field extension of $ K^{\prime}$ and let $\lambda^{\prime} = \lcm(\delta, \lambda)$. Take a primitive  $\lambda^{\prime}$-th root of unity $\zeta^{\prime}:= \ee^{\frac{2\,\pi\,\ii}{\lambda^{\prime}}}$. Then we can express $\gamma_{s}$ in terms of $\Lst{\bar{\alpha}}{1}{w}$ by
	\begin{equation}\label{Eq:usRoot}
	\gamma_{s}={(\zeta^{\prime})}^{\nu_{s}}\myProduct{j}{1}{w}{\alpha_{j}^{-\frac{\upsilon_{j}}{u_{s}}}} = {(\zeta^{\prime})}^{\nu_{s}}\myProduct{j}{1}{w}{(\xoverline{\alpha}_{j}) ^{-\upsilon_{j}\,\cdot\,\sign(u_{s})}}
	\end{equation}
	with $1\le\nu_{s}\le\lambda^{\prime}$. Note that for each $j, -\upsilon_{j} \cdot \sign(u_{s}) \in \ZZ$. Thus we have been able to represent $\gamma_{s}$ as a power product of $\zeta^{\prime}$ and elements $\xoverline{\bs \alpha} = (\Lst{\xoverline{\alpha}}{1}{w}) \in ({K^{\prime\prime}}^{*})^{w}$ which are not roots of unity. Consequently, we can write $\gamma_{i}=(\zeta^{\prime})^{\mu_{i}}\,\bar{\alpha}_{1}^{u_{i,{1}}} \cdots \bar{\alpha}_{w}^{u_{i,{w}}}$ for $1 \leq i \leq s-1$,  where $u_{i,{j}} = \abs{u_{s}} \, v_{i,{j}}$ for $1 \leq j \leq w$ and $1 \leq \mu_{i} \leq \lambda^{\prime}$. Now suppose that $\KarrModule{\xoverline{\bs \alpha}}{K^{\prime\prime}} \neq \zvs{w}$. Then there is a $(\Lst{m}{1}{w}) \in \ZZ^{w} \sm \zvs{w}$ such that 
	\[
	1=\myProduct{j}{1}{w} {(\bar{\alpha}_{j})^{m_{j}}} = \myProduct{j}{1}{w}{\Big( \alpha^{\frac{1}{|u_{s}|}}_{j} \Big)^{m_{j}}} \Longrightarrow \myProduct{j}{1}{w}{\Big( \alpha^{\frac{1}{|u_{s}|}}_{j} \Big)^{|u_{s}|\,m_{j}}} = 1^{|u_{s}|} \iff \myProduct{j}{1}{w}{\alpha^{m_{j}}_{j}} = 1
	\]
	with $(\Lst{m}{1}{w}) \neq \zv{w}$; contradicting the assumption that $\KarrModule{{\bs \alpha}}{K^{\prime}} = \zvs{w}$ holds. Consequently, $\KarrModule{\xoverline{\bs \alpha}}{K^{\prime\prime}}  = \zvs{w}$ which completes the induction step.\\
	Suppose that $K$ is strongly $\s$-computable. Then one can decide if $\KarrModule{{\bs \alpha}}{K^{\prime}}$ is the zero-module, and if not one can compute a non-zero integer vector. All other operations in the proof rely on basic operations that can be carried out. \qed
\end{proof}

\begin{remark}\label{Remark:PiExtforConstantCase}
Let $\Lst{\gamma}{1}{s} \in K^{*}$ and suppose that the ingredients
	$\zeta$, $\Lst{\alpha}{1}{w}$ and the $\mu_i$ and $u_{i,j}$ are given as stated in Theorem~\ref{thm:RPiMonomsReps}.
        Let $n\in\NN$.
	Then by \eqref{eqn:rootOfUnityAndNoIntRelsAlgNums} we have that
	\[
	\gamma_{i}^{n}=\myProduct{k}{1}{n}{\gamma_{i}}= \myProduct{k}{1}{n}{\zeta^{\mu_{i}}}\,\myProduct{k}{1}{n}{\alpha_{1}^{u_{i,{1}}}}\cdots\myProduct{k}{1}{n}{\alpha_{w}^{u_{i,{w}}}} =  \big(\zeta^{n}\big)^{\mu_{i}}\,\big( \alpha^{n}_{1}\big)^{u_{i,{1}}} \cdots \big( \alpha^{n}_{w}\big)^{u_{i,{w}}}.
	\]
	The following remarks are relevant.
	\begin{enumerate} 
		\item Since $\KarrModule{{\bs \alpha}}{K^{\prime}}=\zvs{w}$, we know that there are no $g\in{K^{\prime}}^{*}$, and $(\Lst{u}{1}{w})\in\ZZ^{w}\sm\zvs{w}$ with $1 = \frac{\s(g)}{g} = \ProdLst{\alpha}{1}{w}{u}$. In short we say that $\Lst{\alpha}{1}{w}$ satisfy no integer relation. Thus it follows by Lemma~\ref{lem:transcendentalCriterionForPrdts} that there is a \piE-extension $\dField{\EE}$ of $\dField{K^{\prime}}$ with $\EE=K^{\prime}[y_{1},y^{-1}_{1}]\dots[y_{w}, y^{-1}_{w}]$ and $\s(y_{j})=\alpha_{j}\,y_{j}$ for $j = 1,\dots,w$.
		\item Consider the \aE-extension $\dField{\EE[\vartheta]}$ of $\dField{\EE}$  with $\s(\vartheta)=\zeta\,\vartheta$ of order $\lambda$. By Lemma~\ref{Lemma:RExtOverPiSi} this is an \rE-extension. (Take the quotient field of $\EE$, apply Lemma~\ref{Lemma:RExtOverPiSi}, and then take the corresponding subring.)	
		\item Summarizing, the product expressions $\gamma_{1}^{n},\dots,\gamma_{s}^{n}$ can be rephrased in the \rpiE-extension $\dField{K^{\prime}[y_{1},y^{-1}_{1}]\dots[y_{w}, y^{-1}_{w}][\vartheta]}$ of $\dField{K^{\prime}}$. Namely, we can represent $\alpha_{j}^{n}$ by $y_{j}$ and $\zeta^{n}$ by $\vartheta$.
	\item If $K=\QQ$ (or if $K$ is the quotient field of a certain unique factorization domain), this result can be obtained without any extension, i.e., $K=K'$; see~\cite{schneider2005product}.

	\end{enumerate} 
\end{remark}

So far, Ocansey's Mathematica package \texttt{NestedProducts} contains the algorithmic part of Theorem~\ref{thm:RPiMonomsReps} if $K$ is an algebraic number field, i.e., a finite algebraic field extension of the field of rational numbers $\QQ$. More precisely, the field is given in the form $K=\QQ(\theta)$ together with an irreducible polynomial $f(x)\in\QQ[x]$ with $f(\theta)=0$ such that the degree $n:=\deg f$ is minimal ($f$ is also called the minimal polynomial of $\theta$).  Let $\Lst{\theta}{1}{n}\in\CC$ be the roots of the minimal polynomial $f(x)$. Then the mappings $\varphi_{i}:\QQ(\theta)\to\CC$ defined as $\varphi_{i}(\sum_{j=0}^{n-1}\gamma_{j}\,\theta^{j})=\sum_{j=0}^{n-1}\gamma_{j}\,\theta_{i}^{j}$ with $\gamma_{j}\in\QQ$ are the embeddings of $\QQ(\theta)$ into the field of complex numbers $\CC$ for all $i=1,\dots,n$. Note that any finite algebraic extension $K^{\prime}$ of $K$ can be also represented in a similar manner and can be embedded into $\CC$. 
Subsequently, we consider algebraic numbers as elements in the subfield $\varphi_{i}(\QQ(\theta))$ of $\CC$ for some $i$.

Now let $K$ be such a number field. Applying the underlying algorithm of Theorem~\ref{thm:RPiMonomsReps} to given $\Lst{\gamma}{1}{s} \in K^{*}$ might lead to rather complicated algebraic field extensions in which the $\alpha_i$ are represented. It turned out that the following strategy improved this situation substantially. Namely, consider the map, $\norm{\, \,}:K\to\mathbb{R}$ where $\mathbb{R}$ is the set of real numbers with $\gamma \mapsto \geno{\gamma, \gamma}^{\sfrac{1}{2}}$ where $\geno{\gamma, \gamma}$ denotes the product of $\gamma$ with its complex conjugate. In this setting, one can solve the following problem.
\vspace*{-0.2cm}
\begin{ProblemSpecBox}[\ru{P}]{ 
		\bf \ru{P} for $\gamma\in K^{*}$.
	}\label{prob:ProblemRU}
	{
		\emph{Given} $\gamma\in K^{*}$. \emph{Find}, if possible, a root of unity $\zeta$ such that
		$\gamma = \norm{\gamma}\,\zeta$
		holds.
	}
\end{ProblemSpecBox}

\begin{mlemma}\label{lem:RootOfUnity}
	If $K$ is an algebraic number field, then \ref{prob:ProblemRU} for $\gamma \in
K^{*}$ is solvable in $K$ or some finite algebraic extension $K^{\prime}$ of $K$.
\end{mlemma}

\begin{proof}
	Let $\gamma\in K=\QQ(\alpha)$ where $p(x)$ is the minimal polynomial of
$\alpha$. We consider two cases. Suppose that $\norm{\gamma}\notin K$.
Then using the Primitive Element Theorem (see, e.g., 
\cite[pp. 145]{winkler2012polynomial}) we can construct a new minimal polynomial which represents the algebraic field extension $K'$ of $K$  
with $\norm{\gamma}\in K'$. Define
$\zeta:=\frac{\gamma}{\norm{\gamma}}\in K^{\prime}$. Note that
$\norm{\zeta}=1$. It remains to check if $\zeta$ is a root of unity\footnote{$\zeta$ lies on the unity circle. However, not every algebraic number on the unit circle is a root of unity: Take for instance $\frac{1-\sqrt{3}}{2}+\frac{3^{\frac{1}{4}}}{\sqrt{2}}\,\ii$ and its complex conjugate; they are on the unit circle, but they are roots of the polynomial $x^{4}-2\,x^{3}-2\,x^{2}+1$ which is irreducible in $\QQ[x]$ and which is not a cyclotomic polynomial. For details on number fields containing such numbers see \cite{parry1975units}.},
i.e., if there is an $n\in\NN$ with $\zeta^{n}=1$. This is constructively decidable 
since $K^{\prime}$ is strongly $\s$-computable. In the second case we have
$\norm{\gamma}\in K$, and thus $\zeta := \frac{\gamma}{\norm{\gamma}}\in K$.
Since $K$ is strongly $\s$-computable, one can decide again constructively if there is an $n\in\NN$
with $\zeta^{n}=1$.
\end{proof}

As preprocessing step (before we actually apply Theorem~\ref{thm:RPiMonomsReps}) we check algorithmically if we can solve  \ref{prob:ProblemRU}
for each of the algebraic numbers $\Lst{\gamma}{1}{s}$.
Extracting their roots of unity and applying Proposition~\ref{pro:CommonRootOfUnity}, we can compute a common 
$\lambda$-th root of unity that will represent all the other roots of unity. 

\begin{mproposition}\label{pro:CommonRootOfUnity}
	Let $a$ and $b$ be distinct primitive roots of unity of order
$\lambda_{a}$ and $\lambda_{b}$, respectively. Then there is a primitive
$\lambda_{c}$-th root of unity $c$ such that for some $0 \le m_{a},m_{b}
< \lambda_{c}$ we have $c^{m_{a}} = a$ and $c^{m_{b}}= b$.
\end{mproposition}

\begin{proof}
	Take primitive roots of unity of orders $\lambda_{a}$ and $\lambda_{b}$,
say, $\alpha=\ee^{\frac{2\,\pi\,\ii}{\lambda_{a}}}$ and
$\beta=\ee^{\frac{2\pi\ii}{\lambda_{b}}}$. Let $a=\alpha^{u}$ and
$b=\beta^{v}$ for $0\le u < \lambda_{a}$ and $0\le v < \lambda_{b}$.
Define $\lambda_{c}:= \lcm(\lambda_{a}, \lambda_{b})$ and take a
primitive $\lambda_{c}$-th root of unity,
$c=\ee^{\frac{2\pi\ii}{\lambda_{c}}}$. Then with
$m_{a}=u\,\frac{\lambda_{c}}{\lambda_{a}}\mod{\lambda_{c}}$ and
$m_{b}=v\,\frac{\lambda_{c}}{\lambda_{b}} \mod{\lambda_{c}}$ the Proposition is
proven. \qed
\end{proof}

\begin{mexample}\label{exa:SimpProdOverRootOfUnityAndAlgebraicNumsInDiffRing}
	With $K=\QQ(\ii+\sqrt{3},\sqrt{-13})$, we can extract the following products
	\begin{equation}\label{set:productOverAlgebraicNumbers}
		\myProdWithInnerUnderbrace{k}{1}{n}{-13\sqrt{-13}}{\gamma'_1},\quad \myProdWithInnerUnderbrace{k}{1}{n}{\frac{-784}{13\,\sqrt{-13}\,(\ii+\sqrt{3})^4}}{\gamma'_2},\quad \myProdWithInnerUnderbrace{k}{1}{n}{\frac{-17210368}{13\,\sqrt{-13}\,(\ii+\sqrt{3})^{10}}}{\gamma'_3}
	\end{equation}	
	from \eqref{eqn:HyperGeoPrdtsInDiffRing}. Let $\gamma_{1} = -13$, $\gamma_{2} = \sqrt{-13}$, $\gamma_{3} = -784$, $\gamma_{4} = 13$, $\gamma_{5} = (\ii+\sqrt{3})$ and $\gamma_{6} = -17210368$. Applying \ref{prob:ProblemRU} to each $\gamma_{i}$ we get the roots of unity $1, -1, \ii, \frac{\ii+\sqrt{3}}{2}$ with orders $1, 2, 4, 12$, respectively. By Proposition~\ref{pro:CommonRootOfUnity}, the order of the common root of unity is $12$. Among all the possible $12$-th root of unity, we take  $\zeta := (-1)^{\sfrac{1}{6}} = \ee^{\frac{\pi\,\ii}{6}}$. Note that we can express the other roots of unity with less order in terms of our chosen root of unity, $\zeta$. In particular, we can write $1, -1, \ii$ as $\zeta^{12}, \zeta^{6}, \zeta^{3}$, respectively. Consequently, \eqref{set:productOverAlgebraicNumbers} simplifies to
	\fontsize{9.32837}{0}\selectfont
	\begin{equation}\label{eqn:productOverAlgebraicNumbersWithRootOfUnitySimplified}
	\rootOfUnitySeqExp{6}{n}{9}\,\myProduct{k}{1}{n}{13\,\sqrt{13}},\quad\rootOfUnitySeqExp{6}{n}{11}\, \myProduct{k}{1}{n}{\frac{49}{13\,\sqrt{13}}},\quad\rootOfUnitySeqExp{6}{n}{5}\,\myProduct{k}{1}{n}{\frac{16807}{13\,\sqrt{13}}}.
	\end{equation}
	\normalsize	 
	The pre-processing step yields the numbers $\gamma_{1}^{*}=\sqrt{13}$, $\gamma_{2}^{*}=13$, $\gamma_{3}^{*} = 49$ and $\gamma_{4}^{*}=~16807$ which are not roots of unity. Now we carry out the steps worked out in the proof of Theorem~\ref{thm:RPiMonomsReps}. The package \texttt{NestedProducts} uses Ge's algorithm~\cite{ge1993algorithms} to given $\alpha_{1}=\sqrt{13}$ and $\alpha_{2}'=49$ and finds out that there is no integer relation, i.e, $\KarrModule{(\alpha_{1},\alpha_{2}')}{K'}=\zvs{2}$ with $K'=\QQ\big((-1)^{\frac{1}{6}},\sqrt{13}\big)$. For the purpose of working with primes whenever possible, we write $\alpha_{2}'=\alpha_{2}^{2}$ where $\alpha_{2}=7$. Note that, $\KarrModule{(\alpha_{1},\alpha_{2})}{K'}=\zvs{2}$.  
	Now take the \apE-extension $\dField{K^{\prime}[\vartheta][y_{1},y^{-1}_{1}][y_{2}, y^{-1}_{2}]}$ of $\dField{K^{\prime}}$ with $\s(\vartheta)=(-1)^{\frac{1}{6}}$, $\s(y_{1})=\sqrt{13}\,y_{1}$ and $\s(y_{2})=7\,y_{2}$. By our construction and Remark~\ref{Remark:PiExtforConstantCase} it follows that the \apE-extension is an \rpiE-extension. Further, with $\alpha_{1}$ and $\alpha_{2}$ we can write $13 = \big(\sqrt{13}\big)^{2} \cdot 7^{0}$, $49 = \big(\sqrt{13}\big)^{0} \cdot 7^{2}$ and $16807 = \big(\sqrt{13}\big)^{0} \cdot 7^{5}$. Hence for $a'_1=\vartheta^{9}\,y_{1}^{3}$, $a'_2=\frac{\vartheta^{11}\,y_{2}^{2}}{y_{1}^{3}}$, $a'_3=\frac{\vartheta^{5}\,y_{2}^{5}}{y_{1}^{3}}$ 
	we get $\sigma(a'_i)=\gamma'_i\,a'_i$ for $i=1,2,3$. Thus the shift behavior of the products in~\eqref{set:productOverAlgebraicNumbers} is modeled by $a'_1,a'_2,a'_3$, respectively. In particular, the products in~\eqref{set:productOverAlgebraicNumbers} can be rewritten to
	\fontsize{9.2}{0}\selectfont
	\begin{equation}
		\hspace*{-0.15cm}\rootOfUnitySeqExp{6}{n}{9}\,\algSeqExp{13}{n}{3},\quad\rootOfUnitySeqExp{6}{n}{11}\, {\frac{\intSeqExp{7}{n}{2}}{\algSeqExp{13}{n}{3}}},\quad\rootOfUnitySeqExp{6}{n}{5}\,{\frac{\intSeqExp{7}{n}{5}}{\algSeqExp{13}{n}{3}}}. \vspace{-0.25cm}
	\end{equation}
\end{mexample}

\subsection{Construction of \texorpdfstring{\rpiE}{RPi}-extensions for \texorpdfstring{$\ProdExpr(\KK)$}{ProdE(KK)}}\label{Sec:ConstantRat}

Next, we treat the case that $\KK=K(\kappa_1,\dots,\kappa_u)$ is a rational function field where we suppose that $K$ is strongly $\s$-computable.

\begin{mtheorem} \label{thm:RPiMonomsRepsKRat}
Let $\KK=K(\kappa_1,\dots,\kappa_u)$ be a rational function field over a field $K$ and let
$\Lst{\gamma}{1}{s} \in \KK^{*}$. Then there is an algebraic field extension $K^{\prime}$ of $K$ together with a $\lambda$-th root of unity $\zeta\in K^{\prime}$ and elements ${\bs \alpha} = (\Lst{\alpha}{1}{w}) \in {K^{\prime}(\kappa_1,\dots,\kappa_u)}^{w}$ with $\KarrModule{{\bs \alpha}}{{K^{\prime}(\kappa_1,\dots,\kappa_u)}} = \zvs{w}$ such that for all $i=1,\dots,s$ we have~\eqref{eqn:rootOfUnityAndNoIntRelsAlgNums} for some $1 \le\mu_{i}\le\lambda$ and $(u_{i,{1}},\dots,u_{i,{w}})\in~\ZZ^{w}$.\\ 
If $K$ is strongly $\sigma$-computable, then $\zeta$, the $\alpha_i$ and the $\mu_i, u_{i,j}$ can be computed.
\end{mtheorem}

\begin{proof}
There are monic irreducible\footnote{It would suffice to require that the $f_i\in K[\kappa_1,\dots,\kappa_u]\setminus K$ are monic and pairwise co-prime. For practical reasons we require in addition that the $f_i$ are irreducible. For instance, suppose we have to deal with $(\kappa(\kappa+1))^n$. Then we could take $f_1=\kappa(\kappa+1)$ and can adjoin the \piE-monomial $\sigma(t)=f_1\,t$ to model the product. However, if in a later step also the unforeseen products $\kappa^n$ and $(\kappa+1)^n$ arise, one has to split $t$ into two monomials, say $t_1,t_2$, with $\sigma(t_1)=\kappa\,t_1$ and $\sigma(t_2)=(\kappa+1)\,t_2$. Requiring that the $f_i$ are irreducible avoids such undesirable redesigns of an already constructed \rpiE-extension.} pairwise different polynomials $f_1,\dots,f_s$ from $K[\kappa_1,\dots,\kappa_u]$ and elements $c_1,\dots,c_s\in K^*$ such that for all $i$ with $1\leq i\leq s$ we have
\begin{equation}\label{Equ:PolyRep}
\gamma_i=c_i\,f_1^{z_{i,1}}\,f_2^{z_{i,2}}\dots f_s^{z_{i,s}}
\end{equation}
with $z_{i,j}\in\ZZ$. By Theorem~\ref{thm:RPiMonomsReps} there exist $\bs \alpha=(\alpha_1,\dots,\alpha_w)\in ({K^{\prime}}^*)^w$ in an algebraic field extension $K^{\prime}$ of $K$ with $\KarrModule{{\bs \alpha}}{K^{\prime}} = \zvs{w}$ and a root of unity $\zeta\in K^{\prime}$ such that
\begin{equation}\label{Equ:ConstantRep}
c_i=\zeta^{\mu_{i}}\,\alpha_{1}^{u_{i,{1}}}\cdots\alpha_{w}^{u_{i,{w}}}
\end{equation}
holds for some $m_i,u_{i,j}\in\NN$. Hence $\gamma_i=\zeta^{\mu_{i}}\,\alpha_{1}^{u_{i,{1}}}\cdots\alpha_{w}^{u_{i,{w}}}f_1^{z_{i,1}}\,f_2^{z_{i,2}}\dots f_s^{z_{i,s}}.$ Now let $(\nu_1,\dots,\nu_w,\lambda_1,\dots,\lambda_s)\in\ZZ^{w+s}$ with $1=\alpha_1^{\nu_1}\,\alpha_2^{\nu_2}\dots \alpha_w^{\nu_w}\,f_1^{\lambda_1}f_2^{\lambda_2}\dots f_s^{\lambda_s}.$ Since the $f_i$ are all irreducible and the $\alpha_i$ are from $K'\sm\zs$, it follows that $\lambda_1=\dots=\lambda_s=0$. Note that 
$\alpha_1^{\nu_1}\,\alpha_2^{\nu_2}\dots \alpha_w^{\nu_w}=1$ holds in $K^{\prime}$ if and only if it holds in $K^{\prime}(\kappa_1,\dots,\kappa_u)$. Thus by 
$\KarrModule{{\bs \alpha}}{K^{\prime}} = \zvs{w}$ we conclude that $\nu_1=\dots=\nu_w=0$. Consequently,
$\KarrModule{(\alpha_1,\dots,\alpha_w,f_1,\dots,f_s)}{K^{\prime}(\kappa_1,\dots,\kappa_u)}=\zvs{w+s}$.\\
Now suppose that the computational aspects hold. Since one can factorize polynomials in $K[\kappa_1,\dots,\kappa_u]$, the representation~\eqref{Equ:PolyRep} is computable. In particular, the representation~\eqref{Equ:ConstantRep} is computable by Theorem~\ref{thm:RPiMonomsReps}.
This completes the proof. \qed
\end{proof}

\noindent Note that again Remark~\ref{Remark:PiExtforConstantCase} is relevant where $K'(\Lst{\kappa}{1}{u})$ takes over the role of $K'$: using Theorem~\ref{thm:RPiMonomsRepsKRat} in combination with 
Lemma~\ref{lem:transcendentalCriterionForPrdts} we can construct a \piE-extension in which we can rephrase products defined over $\KK$.
Further, we remark that the package~\texttt{NestedProducts} implements this machinery for the case that $K$ is an
algebraic number field. Summarizing, we allow products that depend on extra parameters. This will be used for the multibasic case with $\KK=K(q_1,\dots,q_e)$ for a field $K$ ($K$ might be again, e.g., a rational function field defined over an algebraic number field). We remark further that for the field $\KK=\QQ(\kappa_1,\dots,\kappa_u)$ this result can be accomplished without any field extension, i.e., $\KK'=\KK$; see~\cite{schneider2005product}.

\begin{mexample}[Cont. \ref{exa:SimpProdOverRootOfUnityAndAlgebraicNumsInDiffRing}]\label{exa:SimpProdOverConstantFieldWithParamenters}
	Let $\KK'=K'(\kappa)$ with $K'=\QQ\big((-1)^{\frac{1}{6}},\sqrt{13}\big)$ and consider
	\begin{equation}\label{Equ:ContentProd}
		\myProdWithInnerUnderbrace{k}{1}{n}{-13\sqrt{-13}\,\kappa}{\gamma_1},\quad \myProdWithInnerUnderbrace{k}{1}{n}{\frac{-784\,(\kappa+1)^{2}}{13\,\sqrt{-13}\,(\ii+\sqrt{3})^{4}\,\kappa}}{\gamma_2},\quad \myProdWithInnerUnderbrace{k}{1}{n}{\frac{-17210368\,(\kappa+1)^{5}}{13\,\sqrt{-13}\,(\ii+\sqrt{3})^{10}\,\kappa}}{\gamma_3}
	\end{equation}
	which are instances of the products from \eqref{eqn:HyperGeoPrdtsInDiffRing}. By Example~\ref{exa:SimpProdOverRootOfUnityAndAlgebraicNumsInDiffRing} the products in~\eqref{set:productOverAlgebraicNumbers} can be modeled in the \rpiE-extension $\dField{K'[\vartheta][y_{1},y^{-1}_{1}][y_{2}, y^{-1}_{2}]}$ of $\dField{K'}$. Note that $\kappa, (\kappa+1)\in K[\kappa]\sm K$ are both irreducible over $K$. Thus $\KarrModule{(\sqrt{13},7, \kappa,\kappa+1)}{\KK'}=\zvs{4}$ holds. Consequently by Remark~\ref{Remark:PiExtforConstantCase}, $\dField{\KK'[\vartheta]\geno{y_{1}}\geno{y_{2}}\geno{y_{3}}\geno{y_{4}}}$ is an \rpiE-extension of $\dField{\KK'}$ with $\s(y_{3})=\kappa\,y_{3}$ and $\s(y_{4})=(\kappa+1)\,y_{4}$. Here the \piE-monomials $y_{3}$ and $y_{4}$ model $\kappa^{n}$ and $(\kappa+1)^{n}$, respectively. In particular, for $i=1,2,3$	we get $\sigma(a_i)=\gamma_i\,a_i$ with 
	\begin{align}\label{Equ:GammaShiftRat}
	a_1=\vartheta^{9}\,y_{1}^{3}\,y_3,&& a_2=\tfrac{\vartheta^{11}\,y_{2}^{2}\,y_4^2}{y_{1}^{3}\,y_3},&& a_3=\tfrac{\vartheta^{5}\,y_{2}^{5}\,y_4^5}{y_{1}^{3}\,y_3}.
	\end{align}
In short, $a_1,a_2,a_3$ model the shift behaviors of the products in~\eqref{Equ:ContentProd}, respectively.
\end{mexample}

\subsection{Structural results for single nested \texorpdfstring{\piE-extensions}{Pi-extensions}}\label{Sec:NonConstantRat}

Finally, we focus on products where non-constant polynomials are involved. Similar to 
Theorem~\ref{thm:RPiMonomsRepsKRat} we will use irreducible factors as main building blocks to define our \piE-extensions. The crucial refinement is that these factors are also shift co-prime; compare also~\cite{schneider2005product,Schneider:14}. Here the following two lemmas will be utilized.

\begin{mlemma}\label{lem:ShiftCoPrimeImpliesEmptyKarrModule}
	Let $\dField{\FF(t)}$ be a \pisiSE-extension of $\dField{\FF}$ with $\sigma(t)=\alpha\,t+\beta$ ($\alpha\in\FF^*$ and $\beta=0$ or $\alpha=1$ and $\beta\in\FF$). Let ${\bs f}=(\Lst{f}{1}{s}) \in (\FF[t]\sm\FF)^{s}$. Suppose that 
	\begin{equation}\label{sta:shiftCoPrime}
	\forall\,i,j\,(1\le i < j \le s): \,\shp(f_{i},f_{j}) = 1
	\end{equation}
	holds and that for $i$ with $1\leq i\leq s$ we have that\footnote{We note that~\eqref{sta:InternalshiftCoPrime} could be also rephrased in terms of Abramov's dispersion~\cite{abramov1971summation,bronstein2000solutions}.}
	\begin{equation}\label{sta:InternalshiftCoPrime}
	\tfrac{\sigma(f_i)}{f_i}\in\FF \,\,\vee\,\, \forall\, k\in\ZZ\setminus\{0\}: \gcd(f_i,\sigma^k(f_i))=1.
	\end{equation}
	Then for all $h\in\FF^{*}$ there does not exist $(\Lst{v}{1}{s}) \in \ZZ^{s} \sm \zvs{s}$ and $g \in \FF(t)^{*}$ with  
	\begin{equation}\label{eqn:genMultTelescoping}
	\frac{\s(g)}{g} = \ProdLst{f}{1}{s}{v}\,h.
	\end{equation}
	In particular, $\KarrModule{{\bs f}}{\FF(t)}=\zvs{s}$.  
\end{mlemma}

\begin{proof}
  Suppose that~\eqref{sta:shiftCoPrime} and~\eqref{sta:InternalshiftCoPrime} hold. Now let $h\in\FF^{*}$ and assume that there are a $g\in\FF(t)^{*}$ and  $(\Lst{v}{1}{s})\in\ZZ^{s}\sm\zvs{s}$ with \eqref{eqn:genMultTelescoping}. Suppose that $\beta=0$ and $g=u\,t^m$ for some $m\in\ZZ$ and some $u\in\FF^{*}$. Then $\frac{\sigma(g)}{g}\in\FF$. Hence $\nu_i=0$ for $1\leq i\leq s$ since the $f_i$ are pairwise co-prime by~\eqref{sta:shiftCoPrime}, a contradiction. Thus we can take a monic irreducible factor, say $p\in\FF[t]\sm\FF$ of $g$ 
	where $p\neq t$ if $\beta=0$. In addition, among all these possible factors we can choose one with the property that for $k>0$, $\s^{k}(p)$ is not a factor in $g$.
	Note that this is possible by Lemma~\ref{lem:KarrsSpecification}. Then $\s(p)$ does not cancel in $\frac{\s(g)}{g}$. Thus $\s(p)\,|\,f_{i}$ for some  $i$  with $1\le i \le s$.
	On the other hand, let $r\le 0$ be minimal such that $\s^{r}(p)$ is the irreducible factor in $g$ with the property that $\s^{r}(p)$ does not occur in $\s(g)$. Note that this is again possible by Lemma~\ref{lem:KarrsSpecification}. Then $\s^{r}(p)$ does not cancel in $\frac{\s(g)}{g}$. Therefore, $\s^{r}(p)\,|\,f_{j}$ for some $j$ with $1\le j\le s$.
	Consequently, $\shp(f_{i},f_{j})\neq1$. By~\eqref{sta:shiftCoPrime} it follows that $i=j$. In particular by~\eqref{sta:InternalshiftCoPrime} it follows that 
	$\sigma(f_i)/f_i\in\FF$.  
	By Lemma~\ref{Lemma:Period} this implies $f_i=w\,t^m$ with $m\in\ZZ$, $w\in\FF^*$ and $\beta=0$. In particular, $p=t$, which we have already excluded. In any case, we arrive at a contradiction and conclude that $v_{1}=\cdots=v_{e}=0$. \qed 
\end{proof}

\noindent Note that condition~\eqref{sta:shiftCoPrime} implies that the $f_i$ are pairwise shift-coprime. In addition condition~\eqref{sta:InternalshiftCoPrime} implies that two different irreducible factors in $f_i$ are shift-coprime. The next lemma considers the other direction.

\begin{mlemma} \label{lem:EmptyKarrModuleAndShiftCoPrime}
	Let $\dField{\FF(t)}$ be a {\dfE} of $\dField{\FF}$ with $t$ transcendental over $\FF$ and $\s(t)=\alpha\,t+\beta$ where $\alpha\in \FF^{*}$ and $\beta\in\FF$. Let ${\bs f}=(\Lst{f}{1}{s})\in (\FF[t]\sm\FF)^{s}$ be irreducible monic polynomials. If there are no $(\Lst{v}{1}{s}) \in \ZZ^{s}\sm\zvs{s}$ and $g \in \FF(t)^{*}$ with 
	\begin{equation}\label{eqn:multTelescoping}
	\tfrac{\s(g)}{g} = \ProdLst{f}{1}{s}{v},
	\end{equation}
	i.e., if $\KarrModule{{\bs f}}{\FF(t)} = \zvs{s}$, then \eqref{sta:shiftCoPrime} holds.
\end{mlemma}

\begin{proof}
	Suppose there are $i,j$ with $1\le i < j\le s$ and $\shp(f_{i},f_{j}) \neq 1$. Since $f_i,f_j$ are irreducible, $f_i\sim f_j$. Thus by Lemma~\ref{lem:ConstructShiftEquivPoly} there is a $g \in\FF(t)^{*}$ with $f_{i}=\frac{\s(g)}{g}\,f_{j}$. Hence $\frac{\s(g)}{g} = f_{i}\,f_{j}^{-1}$ and thus we can find a $(\Lst{v}{1}{s}) \in \ZZ^{s} \sm \zvs{s}$ with \eqref{eqn:genMultTelescoping}. \qed
\end{proof}

\noindent Summarizing, we arrive at the following result. 

\begin{mtheorem}\label{imptThm:ShiftCoPrimePolynomialsAsPiExtension}
	Let $\dField{\FF(t)}$ be a \pisiSE-extension of $\dField{\FF}$. Let ${\bs f} = (\Lst{f}{1}{s})\in(\FF[t]\sm\FF)^{s}$ be irreducible monic polynomials. Then the following statements are equivalent. 
	\begin{enumerate} 
		\item $\forall\,i,j : 1\le i < j \le s, \,\shp(f_{i},f_{j}) = 1$. 
		\item There does not exist $(\Lst{v}{1}{s})\in\ZZ^{s}\sm\zvs{s}$ and $g\in\FF(t)^{*}$ with 
		$\frac{\s(g)}{g} = \ProdLst{f}{1}{s}{v}$,
		i.e., $\KarrModule{{\bs f}}{\FF(t)} = \zvs{s}$.
		\item One can construct a \piE-field extension $\dField{\FF(t)(z_{1})\dots(z_{s})}$ of $\dField{\FF(t)}$ with $\s(z_{i})=f_{i}\,z_{i}$, for $1\le i\le s$.
		\item One can construct a \piE-extension $\dField{\FF(t)[z_{1},z^{-1}_{1}]\dots[z_{s},z^{-1}_{s}]}$ of $\dField{\FF(t)}$ with $\s(z_{i})=f_{i}\,z_{i}$, for $1\le i\le s$.
	\end{enumerate}
\end{mtheorem}

\begin{proof}
Since the $f_i$ are irreducible, the condition~\eqref{sta:InternalshiftCoPrime} always holds. Therefore $(1) \Longrightarrow (2)$ follows from Lemma~\ref{lem:ShiftCoPrimeImpliesEmptyKarrModule}. Further, $(2) \Longrightarrow (1)$ follows from Lemma~\ref{lem:EmptyKarrModuleAndShiftCoPrime}. The equivalences between (2), (3) and (4) follow by Lemma~\ref{lem:transcendentalCriterionForPrdts}. \qed
\end{proof}


\subsection{Construction of  \texorpdfstring{\rpiE}{RPi}-extensions for \texorpdfstring{$\ProdExpr(\KK(n))$}{ProdE(KK(n))}}\label{SubSec:HypergeoemtricCase}

Finally, we combine Theorems~\ref{thm:RPiMonomsRepsKRat} and~\ref{imptThm:ShiftCoPrimePolynomialsAsPiExtension} to obtain a \piE-extension in which expressions from  $\ProdExpr(\KK(n))$ can be rephrased in general. In order to accomplish this task, we will show in Lemma~\ref{lem:rPiExtCombined} that the \piE-monomials of the two constructions in the Subsections~\ref{Sec:ConstantRat} and~\ref{Sec:NonConstantRat} can be merged to one \rpiE-extension.
Before we arrive at this result some preparation steps are needed.

\begin{mlemma}\label{Lemma:AdjoinSigmaE}
	Let $\dField{\FF(t)}$ be a \sigmaE-extension of $\dField{\FF}$ with $\sigma(t)=t+\beta$ and let $\dField{\EE}$ be a \piE-extension of $\dField{\FF}$. Then one can construct a \sigmaE-extension $\dField{\EE(t)}$ of $\dField{\EE}$ with $\sigma(t)=t+\beta$.
\end{mlemma}
\begin{proof}
	Let $\dField{\EE}$ be a \piE-extension of $\dField{\FF}$ with $\EE=\FF(t_1)\dots(t_e)$ and suppose that there is a $g\in\EE$ with $\s(g)=g+\beta$. Let $i$ be minimal such that $g\in\FF(t_1)\dots(t_i)$. Since $\FF(t)$ is a \sigmaE-extension of $\FF$, it follows by part (3) of Theorem~\ref{thm:rpiE-Criterion} that there is no $g\in\FF$ with $\s(g)=g+\beta$. Then~\cite[Lemma~4.1]{karr1985theory} implies that $g$ cannot depend on $t_i$, a contradiction. Thus there is no $g\in\EE$ with $\s(g)=g+\beta$ and by part~(3) of Theorem~\ref{thm:rpiE-Criterion} we get the \sigmaE-extension $\dField{\EE(t)}$ of $\dField{\EE}$ with $\sigma(t)=t+\beta$. \qed
\end{proof}

\noindent As a by-product of the above lemma it follows that the mixed ${\bs q}$-multibasic difference field is built by \piE-monomials and one \sigmaE-monomial.

\begin{mcorollary}\label{cor:MixedMultibasicDiffFieldAsPiExt}
	The mixed ${\bs q}$-multibasic diff.\ ring $\dField{\FF}$ with $\FF=\KK(x)(t_{1})\dots(t_{e})$ from Example~\ref{exa:qMixedDF} is a \pisiSE-extension of $\dField{\KK}$. In particular, $\const\dField{\FF}=\KK$.
\end{mcorollary}

\begin{proof}
       Since the elements $\Lst{q}{1}{e}$ are algebraically independent among each other, there are no $g\in\KK^{*}$ and $(\Lst{v}{1}{e})\in\ZZ^{e}\sm\zvs{e}$ with $1=\frac{\s(g)}{g}=\ProdLst{q}{1}{e}{v}$. Therefore by Lemma~\ref{lem:transcendentalCriterionForPrdts}, $\dField{\EE}$ with $\EE=\KK(t_{1})\dots(t_{e})$ is a \piE-extension of $\dField{\KK}$ with $\s(t_{i})=q_{i}\,t_{i}$ for $1\le i\le e$. Since $\dField{\KK(x)}$ is a \sigmaE-extension of $\dField{\KK}$, we can activate Lemma~\ref{Lemma:AdjoinSigmaE} and can construct the \sigmaE-extension $\dField{\EE(x)}$ of $\dField{\EE}$. Note that $\const\dField{\EE(x)}=\const\dField{\EE}$ also implies that $\const\dField{\KK(x)(t_{1})\dots(t_{e})}=\const\dField{\KK(x)}=\KK$. In particular, the \pE-extension $\dField{\KK(x)(t_1)\dots(t_e)}$ of $\dField{\KK(x)}$ is a \piE-extension. Consequently, $\dField{\FF}$ is a \pisiSE-extension of $\dField{\KK}$.\qed
\end{proof}

\begin{mproposition}\label{pro:FindingGMultExt}
	Let $\dField{\FF(t_1)\dots(t_e)}$ be a \pisiSE-extension of $\dField{\FF}$ with
$\s(t_{i})=\alpha_{i}\,t_{i}+\beta_{i}$ where $\beta_{i}\neq0$ or
$\alpha_{i}=1$. Let $f \in
\FF^{*}$. If there is a $g\in\FF(\Lst{t}{1}{e})^{*}$ with
$\frac{\s(g)}{g}=f$, then $g=\omega\,\ProdLst{t}{1}{e}{v}$ where
$\omega\in\FF^{*}$. In particular,  $v_{i}=0$, if $\beta_{i}\neq0$ (i.e., $t_{i}$ is a \sigmaE-monomial) or $v_{i}\in\ZZ$,
if $\beta_{i}=0$ (i.e., $t_{i}$ is a \piE-monomial).
\end{mproposition}

\begin{proof}
	See \cite[Corollary 2.2.6, pp. 76]{schneider2001symbolic}. \qed
\end{proof}

\begin{mlemma}\label{lem:rPiExtCombined}
	Let $\dField{\KK(x)}$ be the rational difference field with 
	$\s(x)=x+1$ and let $\dField{\KK(x)[z_{1},z_{1}^{-1}]\dots[z_{s},z_{s}^{-1}]}$ be a 
	\piE-extension of $\dField{\KK(x)}$ as given in
	Theorem~\ref{imptThm:ShiftCoPrimePolynomialsAsPiExtension} (item~(4)). Further, let $\KK^{\prime}$ be an algebraic field extension of $\KK$ and let
	$\dField{\KK^{\prime}[y_{1},y_{1}^{-1}]\dots[y_{w},y_{w}^{-1}]}$ be a 
	\piE-extension of $\dField{\KK^{\prime}}$ with $\frac{\s(y_i)}{y_i}\in\KK^{\prime}\setminus\{0\}$. 
	Then the difference ring 
	$\dField{\EE}$ with $\EE=\KK^{\prime}(x)[y_{1}, 
	y^{-1}_{1}]\dots[{y_{w}},y^{-1}_{w}][z_{1},z^{-1}_{1}]\dots[z_{s},z^{-1}_{s}]$ is a \piE-extension of $\dField{\KK^{\prime}(x)}$.
	Furthermore,
	the \aE-extension $\dField{\EE[\vartheta]}$ of $\dField{\EE}$ with 
	$\s(\vartheta)=\zeta\,\vartheta$ of order $\lambda$ is an $R$-extension.
\end{mlemma}

\begin{proof}
	By iterative application of \cite[Corollary $2.6$]{schneider2017summation} it follows that $\dField{\FF}$ is a \piE-field extension of $\dField{\KK^{\prime}}$ with $\FF=\KK^{\prime}(y_{1})\dots(y_{w})$. Note that $\dField{\KK^{\prime}(x)}$ is a \sigmaE-extension of $\dField{\KK^{\prime}}$. Thus by Lemma~\ref{Lemma:AdjoinSigmaE} $\dField{\FF(x)}$ is a \sigmaE-extension of $\dField{\FF}$. We will show that $\dField{\HH}$ with $\HH=\FF(x)(z_{1})\dots(z_{s})$ forms a \piE-extension of $\dField{\FF(x)}$. Since $\dField{\KK(x)[z_{1},z_{1}^{-1}]\dots[z_{s},z_{s}^{-1}]}$ is a \piE-extension of $\dField{\KK(x)}$ as given in Theorem~\ref{imptThm:ShiftCoPrimePolynomialsAsPiExtension} (item~(4)), we conclude that also (item 2) of the theorem, i.e., condition~\eqref{sta:shiftCoPrime} holds. Now suppose that there is a $g\in\FF(x)^*$ and $(l_1,\dots,l_s)\in\ZZ^s$ with $\frac{\sigma(g)}{g}=f_1^{l_1}\dots f_s^{l_s}\in\KK(x)$. By reordering of the generators in $\dField{\FF(x)}$ we get the \piE-extension $\dField{\KK^{\prime}(x)(y_1)\dots(y_w)}$ of $\dField{\KK^{\prime}(x)}$. By Proposition~\ref{pro:FindingGMultExt} we conclude that	$g=q\,y_1^{n_1}\dots y_w^{n_w}$ with $n_1,\dots,n_w\in\ZZ$ and $q\in\KK^{\prime}(x)^*$. Thus $\frac{\sigma(g)}{g}=\frac{\sigma(q)}{q}\alpha_1^{n_1}\dots\alpha_w^{n_w}$ and hence 
	\begin{equation}\label{Equ:RelationInFieldBelow}
	\frac{\sigma(q)}{q}=u\,f_1^{l_1}\dots f_s^{l_s}
	\end{equation}
	for some $u\in{\KK^{\prime}}^*$.	
	Now suppose that $f_i,f_j\in\KK[x]\subset\FF[x]$ with $i\neq j$ are not shift-coprime in $\FF[x]$. Then there are a $k\in\ZZ$ and $v,\tilde{f}_i,\tilde{f}_j\in\FF[x]\setminus\FF$ with $\sigma^k(f_j)=v\,\tilde{f}_j$ and $f_i=v\,\tilde{f}_i$. But this implies that $f_i\,\frac{\tilde{f}_j}{\tilde{f}_i}=\sigma^k(f_j)\in\KK[x]$. Since $f_i,\sigma(f_j)\in\KK[x]$, this implies that $\frac{\tilde{f}_j}{\tilde{f}_i}\in\KK(x)$. Since $f_i,\sigma(f_j)$ are both irreducible in $\KK[x]$ we conclude that $\frac{\tilde{f}_j}{\tilde{f}_i}\in\KK$. Consequently, $f_i$ and $f_j$ are also not shift-coprime in $\KK[x]$, a contradiction. Thus the condition~\eqref{sta:shiftCoPrime} holds not only in $\KK[x]$ but also in $\FF[x]$. Now suppose that $\gcd(f_i,\sigma^k(f_i))\neq1$ holds in $\FF[x]$ for some $k\in\ZZ\setminus\{0\}$. By the same arguments as above, it follows that $\sigma^k(f_i)=u\,f_i$ for some $u\in\KK$. By Lemma~\ref{Lemma:Period} we conclude that $f_i=t$ and $\sigma(t)/t\in\FF$. Therefore also condition~\eqref{sta:InternalshiftCoPrime} holds. Consequently, we can activate Lemma~\ref{lem:ShiftCoPrimeImpliesEmptyKarrModule} and it follows for~\eqref{Equ:RelationInFieldBelow} that $l_1=\dots=l_m=0$. Consequently, we can apply Theorem~\ref{imptThm:ShiftCoPrimePolynomialsAsPiExtension} (equivalence (2) and (3)) and conclude that $\dField{\HH}$ is a \piE-extension of $\dField{\FF(x)}$. Finally, consider the \aE-extension $\dField{\HH[\vartheta]}$ of $\dField{\HH}$  with $\s(\vartheta)=\zeta\,\vartheta$ of order $\lambda$. By Lemma~\ref{Lemma:RExtOverPiSi} it is an \rE-extension. Finally, consider the sub-difference ring $\dField{\HH}$ with $\HH=\KK^{\prime}(x)[y_{1}, y^{-1}_{1}]\dots[{y_{w}},y^{-1}_{w}][z_{1},z^{-1}_{1}]\dots[z_{s},z^{-1}_{s}][\vartheta]$ which is an \apE-extension of $\dField{\KK^{\prime}(x)}$. Since $\const\dField{\HH}=\const\dField{\KK^{\prime}(x)}=\KK^{\prime}$, it is an \rpiE-extension. \qed
\end{proof}

\begin{remark}\label{remk:shiftR-Monomial}
        Take 
        $\dField{\HH}$ with $\HH=\KK^{\prime}(x)[y_{1}, y^{-1}_{1}]\dots[{y_{w}},y^{-1}_{w}][z_{1},z^{-1}_{1}]\dots[z_{s},z^{-1}_{s}][\vartheta]$ as constructed in Lemma~\ref{lem:rPiExtCombined}.       
        Then one can rearrange the generators in $\HH$ and gets the \rpiE-extension
	$\dField{\KK^{\prime}(x)[\vartheta][y_{1}, y^{-1}_{1}]\dots[{y_{w}},y^{-1}_{w}][z_{1},z^{-1}_{1}]\dots[z_{s},z^{-1}_{s}]}$ of $\dField{\KK^{\prime}(x)}$.
\end{remark}

With these considerations we can derive
the following theorem that enables one to construct \rpiE-extension for $\ProdExpr(\KK(n))$.

\begin{mtheorem}\label{imptThm:CombinedShiftCoprimeAndContentsAsRPIExtension}
	Let $\dField{\KK(x)}$ be the rational {\df} with $\s(x)=x+1$ where $\KK=K(\kappa_1,\dots,\kappa_u)$ is a rational function field over a field $K$. Let $\Lst{h}{1}{m}\in\KK(x)^{*}$. Then one can construct an \rpiE-extension $\dField{\AA}$ of $\dField{\KK^{\prime}(x)}$ with 
	\[
	\AA = \KK^{\prime}(x)[\vartheta][y_{1}, y^{-1}_{1}] \dots [{y_{w}}, y^{-1}_{w}][z_{1}, z^{-1}_{1}]\dots[ z_{s}, z^{-1}_{s}]
	\]
	and $\KK^{\prime}=K'(\kappa_1,\dots,\kappa_u)$ where $K^{\prime}$ is an algebraic field extension of $K$ such that
	\begin{itemize}
		\item $\s(\vartheta)=\zeta\,\vartheta$ where $\zeta\in K^{\prime}$ is a  $\lambda$-th root of unity;
		\item $\dfrac{\s(y_{j})}{y_{j}}=\alpha_{j}\in{\KK^{\prime}}\sm\zs$ for $1\le j\le w$ where the $\alpha_j$ are not roots of unity; 
		\item $\dfrac{\s(z_{\nu})}{z_{\nu}}=f_{\nu}\in\KK[x]\sm\KK$ are irreducible\footnote{Instead of irreducibility it would suffice to require that the $f_i$ satisfy property~\eqref{sta:InternalshiftCoPrime} for $1\leq i\leq s$. Restricting to irreducible factors simplifies the proof/construction below. In addition, it also turns the obtained difference ring to a rather robust version. E.g., suppose that one takes $f_1=x(2\,x+1)$ leading to the \piE-monomial $z$ with $\sigma(z)=x\,(2x+1)\,z$. Further, assume that one has to introduce unexpectedly also $x$ and $2\,x+1$ in a later computation. Then one has to split $z$ to the \piE-monomials $z_1,z_2$ with $\sigma(z_1)=x\,z_1$ and $\sigma(z_2)=(2x+1)\,z_2$, i.e., one has to redesign the already constructed \rpiE-extension. In short, irreducible polynomials provide an \rpisiSE-extension which most probably need not be redesigned if other products have to be considered.} and shift co-prime for $1\le\nu\le s$;
	\end{itemize}
	holds with the following property. For $1\leq i\leq m$ one can define 
	\begin{equation}\label{Equ:giFormRat}
		g_{i}=r_{i}\,\vartheta^{\mu_{i}}\,y_{1}^{u_{i,{1}}}\cdots y_{w}^{u_{i,{w}}}\,z_{1}^{v_{i,{1}}}\cdots z_{s}^{v_{i,{s}}}\in\AA
	\end{equation}
		with $0 \leq \mu_{i} \leq \lambda-1$, $u_{i,{1}},\dots,u_{i,{w}}$,$v_{i,{1}},\dots,v_{i,{s}} \in \ZZ$ and  $r_{i}\in\KK(x)^{*}$ such that 
	\begin{equation}\label{Equ:giShiftRat}
	\s(g_{i})=\s(h_{i})\,g_{i}.
	\end{equation}
	If $K$ is strongly $\sigma$-computable, the components of the theorem can be computed.
	\end{mtheorem}
\begin{proof}
	For $1\leq i\leq m$ we can take pairwise different monic irreducible polynomials $\Lst{p}{1}{n}\in\KK[x]\sm\KK$ $\Lst{\gamma}{1}{m}\in\KK^{*}$ and $\Lstc{d}{i}{1}{n}\in\ZZ$ such that $\sigma(h_{i})=\gamma_{i}\,p_{1}^{d_{i,1}}\cdots p_{n}^{d_{i,n}}$ holds. Note that this representation is computable if $K$ is strongly $\s$-computable. By Theorem~\ref{thm:RPiMonomsRepsKRat} it follows that there are a $\lambda$-th root of unity $\zeta\in K^{\prime}$, elements $\bs{\alpha}=(\Lst{\alpha}{1}{w}) \in ({\KK^{\prime}}^{*})^{w}$ with $\KarrModule{\bs{\alpha}}{\KK^{\prime}}=\zvs{w}$ and integer vectors $(u_{i,{1}},\dots,u_{i,{w}})\in\ZZ^{w}$ and $\mu_{i}\in\NN$ with $0\le\mu_{i}<\lambda$ such that $\gamma_{i}=\zeta^{\mu_{i}}\,\alpha_{1}^{u_{i,{1}}}\cdots\alpha_{w}^{u_{i,{w}}}$ holds for all $1 \le i\le m$.  Obviously, the $\alpha_j$ with $1\leq j\leq w$ are not roots of unity. 
	By Lemma~\ref{lem:transcendentalCriterionForPrdts} 
	we get the \piE-extension $\dField{\KK^{\prime}[y_{1},y^{-1}_{1}]\dots[y_{w},y^{-1}_{w}]}$ of $\dField{\KK^{\prime}}$ with $\s(y_{j})=\alpha_{j}\,y_{j}$ for $1\le j\le w$ and we obtain 
	\begin{equation}\label{Equ:ConstantPart}
	a_i=\vartheta^{\mu_{i}}\,y_{1}^{u_{i,{1}}}\cdots y_{w}^{u_{i,{w}}}
	\end{equation}
	with 
	\begin{equation}\label{Equ:ShiftConstantPart}
	\sigma(a_i)=\gamma_i\,a_i
	\end{equation}
	for $1\leq i\leq m.$
	Next we proceed with the non-constant polynomials in $\KK[x]\sm\KK$. Set $\mathcal{I}=\{\Lst{p}{1}{n}\}$. Then there is a partition $\mathcal{P}=\{\mathcal{E}_{1},\dots,\mathcal{E}_{s}\}$ of $\mathcal{I}$ with respect to $\sim_{\s}$, i.e., each $\mathcal{E}_{i}$ contains precisely the shift equivalent elements of $\mathcal{P}$. Take a representative from each equivalence class $\mathcal{E}_{i}$ in $\mathcal{P}$ and collect them in $\mathcal{F}:=\{\Lst{f}{1}{s}\}$. Since each $f_{i}$ is shift equivalent with every element of $\mathcal{E}_{i}$, it follows by Lemma~\ref{lem:ConstructShiftEquivPoly} that for all $h\in\mathcal{E}_{i}$, there is a rational function $r\in\KK(x)^{*}$ with $h=\frac{\s(r)}{r}\,f_{i}$ for $1\le i \le s$. Consequently, we get $r_{i}\in\KK(x)^{*}$ and $v_{i,j}\in\ZZ$ with 
	$p_{1}^{d_{i,1}}\cdots p_{n}^{d_{i,n}}=\frac{\s(r_{i})}{r_{i}}\,f_{1}^{v_{i,1}}\cdots f_{s}^{v_{i,s}}$
	for all $1\le i \le s$.	Further, by this construction, we know that $\shp(f_{i},f_{j})=1$ for $1\le i< j\le s$. Therefore, it follows by Theorem~\ref{imptThm:ShiftCoPrimePolynomialsAsPiExtension} that we can construct the \piE-extension $\dField{\KK(x)[z_{1},z^{-1}_{1}]\dots[z_{s},z^{-1}_{s}]}$ of $\dField{\KK(x)}$ with $\s(z_{i})=f_{i}\,z_{i}$. Now define
	$b_i=r_{i}\,t_{1}^{v_{i,1}}\cdots t_{s}^{v_{i,s}}$.
	Then we get 
	\begin{equation}\label{Equ:Sigmabi}
	\sigma(b_i)=p_{1}^{d_{i,1}}\cdots p_{n}^{d_{i,n}}\,b_i.
	\end{equation}
	Finally, by Lemma~\ref{lem:rPiExtCombined} and Remark~\ref{remk:shiftR-Monomial} we end up at the \rpiE-extension $\dField{\AA}$ of $\dField{\KK'(x)}$ with $\AA = \KK^{\prime}(x)[\vartheta][y_{1},y^{-1}_{1}]\dots[{y_{w}},y^{-1}_{w}][z_{1},z^{-1}_{1}]\dots[z_{e},z^{-1}_{e}]$ with $\s(\vartheta)=\zeta\,\vartheta$, $\s(y_{j})=\alpha_{j}\,y_{j}$ for $1\le j\le w$ and $\s(z_{i})=f_{i}\,z_{i}$ for $1\leq i\leq s$.\\
	Now let $g_i$ be as defined in~\eqref{Equ:giFormRat}. Since $g_i=a_i\,b_i$ with~\eqref{Equ:ShiftConstantPart} and~\eqref{Equ:Sigmabi},
	we conclude that~\eqref{Equ:giShiftRat} holds.
	If $K$ is strongly $\s$-computable, all the ingredients delivered by Theorems~\ref{thm:RPiMonomsReps} and~\ref{imptThm:ShiftCoPrimePolynomialsAsPiExtension} can be computed. This completes the proof.\qed  
	\end{proof}

\begin{mexample}\label{exa:SimpNonConstHyperGeoTermsInDiffRing}
	Let $\KK=K(\kappa)$ be a rational function field over the algebraic number field $K=\QQ(\ii+\sqrt{3},\sqrt{-13})$ and take the rational {\df} $\dField{\KK(x)}$ with $\s(x)=x+1$. Given~\eqref{Equ:RatXh}, we can write
	\begin{align*}
	\sigma(h_1)=\gamma_1\,p_1^{-1},&& \sigma(h_2)=\gamma_2\,p_1\,p_2^{-2},&& \sigma(h_3)=p_1\,p_2^{-5}
	\end{align*}
	where the $\gamma_1,\gamma_2,\gamma_3$ are given in~\eqref{Equ:ContentProd} and where we set $p_1=x+1$, $p_2=x+3$ as our monic irreducible polynomials.	
 Note that $p_1$ and $p_2$ are shift equivalent: $\gcd(p_2,\s^{2}(p_1))=p_2$. Consequently both factors fall into the same equivalence class $\mathcal{E}=\{\s^{k}(x+1)\,|\,k\in\ZZ\}=\{\s^{k}(x+3)\,|\,k\in\ZZ\}$.
Take $p_1=x+1$ as a representative of the  equivalence class $\mathcal{E}$. Then by Lemma~\ref{lem:ConstructShiftEquivPoly}, it follows that there is a $g\in\KK(x)^{*}$ that connects the representatives to all other elements in their respective equivalence classes. In particular with our example we have $x+3 = \frac{\s(g)}{g}\,(x+1)$ where  $g = (x+1)\,(x+2)$. 
By Theorem~\ref{imptThm:ShiftCoPrimePolynomialsAsPiExtension}, it follows that $\dField{\KK(x)[z,z^{-1}]}$ is a \piE-extension of the {\df} $\dField{\KK(x)}$ with $\s(z) = (x+1)\,z$. In this ring, the \piE-monomial $z$ models $n!$.
By Lemma~\ref{lem:rPiExtCombined} the constructed difference rings $\dField{\KK'[\vartheta]\geno{y_{1}}\geno{y_{2}}\geno{y_{3}}\geno{y_{4}}}$ and $\dField{\KK(x)\geno{z}}$ from Example~\ref{exa:SimpProdOverConstantFieldWithParamenters} with
$\KK'= \QQ\big((-1)^{\frac{1}{6}}, \sqrt{13})(\kappa)$
can be merged into a single \rpiE-extension $\dField{\AA}$ where $\AA$ is \eqref{alg:RPiExtension} with the automorphism defined accordingly.
Further note that for $b_1=\frac1z$, $b_2=\frac1{(x+1)^2(x+2)^2\,z}$, $b_3=\frac1{(x+1)^5(x+2)^5z^4}$ we have that $\sigma(b_1)=p_1^{-1}\,b_1$, $\sigma(b_2)=p_1\,p_2^{-2}\,b_2$ and $\sigma(b_3)=p_1\,p_2^{-5}\,b_3$. Taking $a_1,a_2,a_3$ in~\eqref{Equ:GammaShiftRat} with $\sigma(\gamma_i)=a_i\,\gamma_i$ for $i=1,2,3$, we define $g_i=a_i\,b_i$ for $i=1,2,3$ and obtain $\sigma(g_i)=\sigma(h_i)\,g_i$. Note that the $g_i$ are precisely the elements given in~\eqref{eqn:SimpHyperGeoInDiffRing}.
\end{mexample}

Now we are ready to prove Theorem~\ref{thm:ProblemRMHPE} for the special case $\ProdExpr(\KK(n)$. Namely, consider the products 
$$P_{1}(n) = \myProduct{k}{\ell_{1}}{n}{h_{1}(k)},\dots,P_{m}(n) = \myProduct{k}{\ell_{m}}{n}{h_{m}(k)}\ \in\Prod(\KK(n))$$
with $\ell_i\in\NN$ where $\ell_i\geq Z(h_i)$. Further, suppose that we are given the components as claimed in Theorem~\ref{imptThm:CombinedShiftCoprimeAndContentsAsRPIExtension}.\\
$\bullet$ Now take the difference ring embedding  $\tau(\frac{a}{b})=\langle\ev\big(\frac{a}{b}, n \big)\rangle_{n \geq 0}$ for  $a,b\in\KK[x]$ where $\ev$ is defined in~\eqref{Equ:EvRatDefinition}. Then by iterative application of part~(2) of Lemma~\ref{lem:injectiveHom} we can construct the $\KK^{\prime}$-homomorphism $\tau:\AA\to\ringOfEquivSeqs[\KK^{\prime}]$ determined by the homomorphic extension of
\begin{itemize}
\item $\tau(\vartheta)=\langle\zeta^n\rangle_{n\geq0}$, 
\item $\tau(y_i)=\langle\alpha_i^n\rangle_{n\geq0}$ for $1\leq i\leq w$ and 
\item $\tau(z_i)=\langle\prod_{k=\ell'_i}^nf_i(k-1)\rangle_{n\geq0}$ with $\ell'_i=Z(f_i)+1$ for $1\leq i\leq s$.
\end{itemize}
In particular, since $\dField{\AA}$ is an \rpiE-extension of $\dField{\KK'(x)}$, it follows by part (3) of Lemma~\ref{lem:injectiveHom} that $\tau$ is a $\KK^{\prime}$-embedding.\\ 
$\bullet$ Finally, define for $1\leq i\leq m$ the product expression
$$G_i(n)=r_{i}(n)\,(\zeta^n)^{\mu_{i}}\,(\alpha_{1}^n)^{u_{i,{1}}}\cdots (\alpha_{w}^n)^{u_{i,{w}}}\,(\myProduct{k}{\ell'_1}{n}{f_1(k-1)})^{v_{i,{1}}}\cdots (\myProduct{k}{\ell'_s}{n}{f_s(k-1)})^{v_{i,{s}}}$$
from $\Prod(\KK'(n))$ and define $\delta_i=\max(\ell_i,\ell'_1,\dots,\ell'_s,Z(r_i))$.
Observe that $\tau(g_i)=\langle G'_i(n)\rangle_{n\geq0}$ with
\begin{equation}\label{Equ:GiPrime}
G'_i(n)=\begin{cases} 0 &\text{ if }0\leq n<\delta_i\\ G_i(n) & \text{ if }n\geq\delta_i.\end{cases}
\end{equation}
By~\eqref{Equ:giShiftRat} and the fact that $\tau$ is a $\KK'$-embedding, it follows that $S(\tau(g_i))=S(\tau(h_i))\,\tau(g_i)$. In particular, for $n\geq\delta_i$ we have that $G_i(n+1)=h_i(n+1)\,G_i(n)$. By definition, we have $P_i(n+1)=h_i(n+1)\,P_i(n)$ for $n\geq\delta_i\geq\ell_i$. Since $G_i(n)$ and $P_i(n)$ satisfy the same first order recurrence relation, they differ only by a multiplicative constant. Namely, setting $Q_i(n)=c\,G_i(n)$ with $c=\frac{P_i(\delta_i)}{G_i(\delta_i)}\in(\KK^{\prime})^*$ we have that $P_i(\delta_i)=Q_i(\delta_i)$ and thus $P_i(n)=Q_i(n)$ for all $n\geq\delta_i$. This proves part (1) of Theorem~\ref{thm:ProblemRMHPE}.\\
Since $\tau$ is a $\KK'$-embedding, the sequences $$\funcSeqA{\alpha_{1}^{n}}{n},\dots,\funcSeqA{\alpha_{w}^{n}}{n}, \funcSeqA{\myProduct{k}{\ell'_{1}}{n}{f_{1}(k-1)}}{n},\dots,\funcSeqA{\myProduct{k}{\ell'_{s}}{n}{f_{s}(k-1)}}{n}$$ 
are among each other algebraically independent over $\tau\big(\KK'(x)\big)\big[\funcSeqA{\zeta^{n}}{n}\big]$ which proves property $(2)$ of Theorem~\ref{thm:ProblemRMHPE}.

\begin{mexample}[Cont. Example~\ref{exa:SimpNonConstHyperGeoTermsInDiffRing}]\label{exa:constRPiExts}
We have $\sigma(g_i)=\sigma(h_i)\,g_i$ for $i=1,2,3$ where the $h_i$ and $g_i$ are given in~\eqref{Equ:RatXh} and~\eqref{eqn:SimpHyperGeoInDiffRing}, respectively. For the $\KK'$-embedding defined in Example~\ref{exa:ConstructRPiExt} we obtain $c_i\,\tau(g_i)=\langle P_i(n)\rangle_{n\geq0}$ with $P_i(n)=\prod_{k=1}^nh_i(k)$ and $c_1=1$, $c_2=4$ and $c_3=32$. Since there are no poles in the $g_i$ we conclude that for
\begin{align*}
G_1(n)&=\tfrac{\rootOfUnitySeqExp{6}{n}{9}\,\algSeqExp{13}{n}{3}\,\kappa^{n}}{n!},\quad\quad
&G_2(n)=\tfrac{4\,\rootOfUnitySeqExp{6}{n}{11}\,\intSeqExp{7}{n}{2}\,\polySeqExp{\kappa+1}{n}{2}}{(n+1)^{2}\,(n+2)^{2}\,\algSeqExp{13}{n}{3}\,\kappa^{n}\,n!},\\[-0.2cm]
G_3(n)&=\tfrac{32\,\rootOfUnitySeqExp{6}{n}{5}\,\intSeqExp{7}{n}{5}\,\polySeqExp{\kappa+1}{n}{5}}{(n+1)^{5}\,(n+2)^{5}\,\algSeqExp{13}{n}{3}\,\kappa^{n}\,\polySeq{n!}{4}}
\end{align*}
we have $P_i(n)=G_i(n)$ for $n\geq1$. With $P(n)=P_1(n)+P_2(n)+P_3(n)$ (see~\eqref{eqn:HyperGeoPrdtsInDiffRing}) and $Q(n)=G_1(n)+G_2(n)+G_3(n)$ (see~\eqref{eqn:SimpHyperGeoPrdtsSeqSetting}) we get $P(n)=Q(n)$ for $n\geq1$.
\end{mexample}

\section{Construction of
\texorpdfstring{\rpiE}{RPi}-extensions for
\texorpdfstring{$\ProdExpr(\KK(n,\bs{q}^{n}))$}{ProdExpr(K(n,qn))}}\label{Sec:mixedHypergeometricCase}

In this section we extend the results of
Theorem~\ref{imptThm:CombinedShiftCoprimeAndContentsAsRPIExtension} to the case $\ProdExpr(\KK(n,\bs{q}^{n}))$.
As a consequence, we will also prove Theorem~\ref{thm:ProblemRMHPE}.

\subsection{Structural results for nested \texorpdfstring{\piE-extensions}{Pi-extensions}}\label{Sec:NestedPiSigmaStructural}

In the following let $\dField{\FF_e}$ be a \pisiSE-extension of $\dField{\FF_0}$ with $\FF_e=\FF_{0}(\myt_1)\dots(\myt_e)$ with $\sigma(\myt_i)=\alpha_i\,\myt_i+\beta_i$ and $\alpha_i\in\FF_{0}^*$, $\beta_i\in\FF_{0}$ for $1\leq i\leq e$. We set $\FF_i=\FF_{0}(\myt_1)\dots(\myt_{i})$ and thus $\dField{\FF_{i-1}(\myt_i)}$ is a \pisiSE-extension of $\dField{\FF_{i-1}}$ for $1\leq i\leq e$.\\
We will use the following notations. For ${\bs f}=(f_1,\dots,f_s)$ and $h$ we 
write ${\bs f}\wedge h =\left(f_1,\dots,f_s, h\right)$
for the concatenation of $\bs{f}$ and $h$. Moreover, the concatenation
of $\bs{f}$ and $\bs{h}=(h_1,\dots,h_u)$ is denoted by ${\bs f}\wedge{\bs
h}=\left(f_1,\dots,f_s,h_1,\dots,h_u\right)$.

\begin{mlemma}\label{lem:MixedCasePiSigmaConstruction}
	Let $\dField{\FF_{e}}$ be a \pisiSE-extension of $\dField{\FF_0}$ as above.
If the polynomials in
$\bs{f_{i}}\in(\FF_{i-1}[\myt_{i}]\sm\FF_{i-1})^{s_{i}}$ for $1 \le i \le e$ and $s_i\in\NN\setminus\{0\}$ satisfy conditions~\eqref{sta:shiftCoPrime} and~\eqref{sta:InternalshiftCoPrime} with $\FF$ replaced by $\FF_{i-1}$, then
$\KarrModule{\bs{f_{1}}\wedge\dots\wedge\bs{f_{e}}}{\FF_{e}}=\zvs{s}$
where $s = s_{1}+\cdots+s_{e}$.
\end{mlemma}

\begin{proof}
Let
$\bs{v_{1}}\in\ZZ^{s_{1}},\dots,\bs{v_{e}}\in\ZZ^{s_{e}}$ and
$g\in\FF^{*}_{e}$ with
	\begin{align} \label{eqn:MultiMixedCasePMT}
	\frac{\s(g)}{g}={{\bs{f}}^{{\bs{v}}_{\bs{1}}}_{\bs{1}}} \,
{{\bs{f}}^{{\bs{v}}_{\bs{2}}}_{\bs{2}}} \cdots
{{\bs{f}}^{{\bs{v}}_{\bs{e}}}_{\bs{e}}}  .
	\end{align}
Suppose that not all $\bs{v_{i}}$ with $1\leq i\leq e$ are zero-vectors and let $r$ be maximal such that
$\bs{v_{r}} \neq \zv{s_{r}}$. Thus the right hand side of~\eqref{eqn:MultiMixedCasePMT} is in $\FF_r$ and it follows by Proposition~\ref{pro:FindingGMultExt} that $g=\gamma\,\ProdLst{\myt}{r+1}{e}{u}$ with $\gamma\in\FF_{r}^{*}$ and $u_i\in\ZZ$; if $\myt_i$ is a \sigmaE-monomial, then $u_i=0$. Hence
$$\frac{\s(\gamma)}{\gamma}
=\ProdLst{\alpha}{r+1}{e}{-u}\,{{\bs{f}}^{{\bs{v}}_{\bs{1}}}_{\bs{1}}} \cdots
{{\bs{f}}^{{\bs{v}}_{\bs{r-1}}}_{\bs{r-1}}} {{\bs{f}}^{{\bs{v}}_{\bs{r}}}_{\bs{r}}}=h\,\vectorexpsubscript{f}{v}{r}$$ 
with $h=\ProdLst{\alpha}{r+1}{e}{-u}\,\vectorexpsubscript{f}{v}{1} \cdots
\vectorexpsubscript{f}{v}{r-1}\in\FF_{r-1}^*$.
Since conditions~\eqref{sta:shiftCoPrime} and~\eqref{sta:InternalshiftCoPrime} with $\FF$ replaced by $\FF_{r-1}$ hold for these entries, Lemma~\ref{lem:ShiftCoPrimeImpliesEmptyKarrModule} is applicable and we get 
${\bs v_r}=\zv{s_{r}}$, a contradiction.  \qed
\end{proof}

\noindent We can now formulate a generalization of
Theorem~\ref{imptThm:ShiftCoPrimePolynomialsAsPiExtension} for nested \pisiSE-extensions. 

\begin{mtheorem}\label{imptThm:MixedCasePeriodZeroShiftCoPrimePolynomialsAsPiExtension}
	Let $\dField{\FF_{e}}$ be the \pisiSE-extension of $\dField{\FF_0}$ from above. For $1\leq i\leq e$, let
$\bs{f_{i}} =
(\Lstc{f}{i}{1}{s_{i}})\in(\FF_{i-1}[\myt_{i}]\sm\FF_{i-1})^{s_{i}}$ with $s_i\in\NN\setminus\{0\}$ containing irreducible monic polynomials.
Then the following statements are equivalent.
	\begin{enumerate} 
		\item $\shp(f_{i,{j}},f_{i,{k}})=1$ for all $1\le i\le e$ and $1 \le j <
k \le s_{i}$.
		\item There does not exist
$\bs{v_{1}}\in\ZZ^{s_{1}},\,\dots,\,\bs{v_{e}}\in\ZZ^{s_{e}}$ with $\bs{v_1}\wedge\dots\wedge\bs{v_e}\neq\bs{0}_{s}$
and $g \in \FF^{*}_{e}$ such that
		\[
		\dfrac{\s(g)}{g}=\vectorexpsubscript{f}{v}{1}\cdots\vectorexpsubscript{f}{v}{e}
		\]
		holds. That is,
$\KarrModule{\bs{f_{1}}\wedge\dots\wedge\bs{f_{e}}}{\FF_{e}}=\zvs{s}$
where $s=s_{1}+\cdots+s_{e}$.
		\item One can construct a \piE-field extension
$\dField{\FF_{e}(z_{1,{1}})\dots(z_{1,{s_{1}}})\dots(z_{e,{1}})\dots(z_{e,{s_{e}}})}$
of $\dField{\FF_{e}}$ with $\s(z_{i,{k}})=f_{i,{k}}\,z_{i,{k}}$ for $1
\le i \le e$ and $1 \le k \le s_{i}$.
		\item One can construct a \piE-extension $\dField{\EE}$ of $\dField{\FF_e}$ with the ring of Laurent polynomials $\EE=\FF_{e}[z_{1,{1}},z^{-1}_{1,{1}}]\dots[z_{1,{s_{1}}},z^{-1}_{1,{s_{1}}}]\dots[
z_{e,{1}},z^{-1}_{e,{1}}] \dots [z_{e,{s_{e}}}, z^{-1}_{e,{s_{e}}}]$ and $\s(z_{i,{k}}) = f_{i,{k}}\,z_{i,{k}}$ for $1
\le i \le e$ and $1 \le k \le s_{i}$.
	\end{enumerate}
\end{mtheorem}

\begin{proof}
	$(1) \Longrightarrow (2)$: Since the entries in $\bs{f_{i}}$ are shift co-prime and irreducible, conditions~\eqref{sta:shiftCoPrime} and~\eqref{sta:InternalshiftCoPrime} hold ($\FF$ replaced by $\FF_{i-1}$) and thus statement (2) follows by Lemma~\ref{lem:MixedCasePiSigmaConstruction}.

\noindent $(2) \Longrightarrow (3)$:
We prove the statement by induction on the number of \pisiSE-monomials $\myt_1,\dots,\myt_e$. For $e=0$ nothing has to be shown. Now suppose that the implication has been shown for $\FF_{e-1}$, $e\geq0$ and set $\EE=\FF_{e-1}(\twoargstwosamesubscript{z}{1},\dots,\threeargstwosamesubscript{z}{1}{s})\dots(\threeargssubscript{z}{e-1}{1},\dots,\threeargstwosamesubscript{z}{e-1}{s})$. Suppose that $\dField{\EE(\threeargssubscript{z}{e}{1},\dots,\threeargstwosamesubscript{z}{e}{s})}$ is not a \piE-extension of $\dField{\EE}$ and let $\ell$ be minimal with $s_{\ell}<s_{e}$ such that $\dField{\EE(\threeargssubscript{z}{e}{1},\dots,\fourargssubscript{z}{e}{s}{\ell})}$ is a \piE-extension of $\dField{\EE}$. Then by Theorem~\ref{thm:rpiE-Criterion}(1) there are a $v_{e,{s_{\ell}}}\in\ZZ\sm\zs$ and an $\omega \in \EE(\threeargssubscript{z}{e}{1},\dots,\fourargssubscript{z}{e}{s}{j})^{*}$ with $j=\ell-1$ such that $\s(\omega)=f^{v_{e,{s_{\ell}}}}_{e_{s_{\ell}}}\,\omega$ holds. By Proposition~\ref{pro:FindingGMultExt}, $\omega=g\,\fourargsexpsubscript{z}{v}{e}{1}\cdots\fiveargsexpsubscript{z}{v}{e}{s}{j}$ with $(\threeargssubscript{v}{e}{1},\dots,\fourargssubscript{v}{e}{s}{j})\in\ZZ^{s_{j}}$ and $g\in\FF^{*}_{e-1}$. Thus $\frac{\s(g)}{g}=\fourargsexpsubscript{f}{-v}{e}{1}\cdots\fiveargsexpsubscript{f}{-v}{e}{s}{j}\,\fiveargsexpsubscript{f}{v}{e}{s}{\ell}$.

\noindent	$(3) \Longrightarrow (2)$.	We prove the statement by induction on the number of \pisiSE-monomials $\myt_1,\dots,\myt_e$. For the base case $e=0$ nothing has to be shown. Now suppose that the implication has been shown already for $e-1$ \pisiSE-monomials and set $\EE=\FF_{e}(\twoargstwosamesubscript{z}{1},\dots,\threeargstwosamesubscript{z}{1}{s})\dots(\threeargssubscript{z}{e-1}{1},\dots,\threeargstwosamesubscript{z}{e-1}{s})$. Suppose that $\dField{\EE(\threeargssubscript{z}{e}{1},\dots,\threeargstwosamesubscript{z}{e}{s})}$ is a \piE-extension of $\dField{\EE}$ and assume on the contrary that there is a $g\in\FF_{e}^{*}$ and $\bs{v_{e}}\in\ZZ^{s_{e}}\sm\zvs{s_{e}}$ such that $\frac{\s(g)}{g}=\vectorexpsubscript{f}{v}{1}\,\cdots\vectorexpsubscript{f}{v}{e-1}\,\vectorexpsubscript{f}{v}{e}$	holds. Let $j$ be maximal with $v_{e,{j}}\neq 0$ and define $\gamma:=g\,\vectorexpsubscript{z}{-v}{1}\cdots\vectorexpsubscript{z}{-v}{e-1}\,\fourargsexpsubscript{z}{-v}{e}{1}\cdots\fourargsexpsubscript{z}{-v}{e}{j-1}\in\EE(\threeargssubscript{z}{e}{1},\dots,\threeargssubscript{z}{e}{j-1})^{*}$ where $\vectorexpsubscript{z}{-v}{i}=\fourargsexpsubscript{z}{-v}{i}{1}\cdots\fourargsexpsubscript{z}{-v}{i}{s_{i}}$ for $1\leq i<e$ and $g\in\FF^{*}_{e}$. Then $\frac{\s(\gamma)}{\gamma}= \fourargsexpsubscript{f}{v}{e}{j}$ with $v_{e_{j}}\neq 0$; a contradiction since $\dField{\EE(\threeargssubscript{z}{e}{1},\dots,z_{e,{j}})}$ is a \piE-extension of $\dField{\EE(\threeargssubscript{z}{e}{1},\dots,z_{e,{j-1}})}$ by Theorem~\ref{thm:rpiE-Criterion}.

\noindent$(2) \Longrightarrow (1)$. We prove the statement by induction on the number of \pisiSE-monomials $\myt_1,\dots,\myt_e$. For $e=0$ nothing has to be shown. Now assume that the implication holds for the first $e-1$ \pisiSE-monomials. Now suppose that there are $k,\ell$ with $1\le k,\ell\le s_{e}$ and $k\neq\ell$ such that $\shp(\threeargssubscript{f}{e}{k},\threeargssubscript{f}{e}{\ell})\neq1$ holds. Since $\shp(\threeargssubscript{f}{e}{k},\threeargssubscript{f}{e}{\ell})\neq1$ we know that they are shift equivalent and because $\threeargssubscript{f}{e}{k},\,\threeargssubscript{f}{e}{\ell}$ are monic it follows by Lemma~\ref{lem:ConstructShiftEquivPoly} that there is a $g \in \FF^{*}_{e}$ with $\frac{\s(g)}{g}\threeargssubscript{f}{e}{k}=\threeargssubscript{f}{e}{\ell}$ and thus $\frac{\s(g)}{g}=\vectorexpsubscript{f}{v}{1}\,\cdots \vectorexpsubscript{f}{v}{e}$ holds with ${\bs v_{1}}=\cdots{\bs v_{e-1}}=0$ and  ${\bs v_{e}}=(0,\dots,0,v_{e,{k}},0,\dots,0,v_{e,{\ell}},0,\dots,0)\in\ZZ^{s_{e}}\sm\zvs{s_{e}}$ where $v_{e,{k}}=-1$ and $v_{e,{\ell}}=1$.

\noindent$(3) \Longrightarrow (4)$ is obvious and $(4) \Longrightarrow(3)$ follows by \cite[Corollary $2.6$]{schneider2017summation}. \qed
\end{proof}

\subsection{Proof of the main result (Theorem~\ref{thm:ProblemRMHPE})}

Using the structural results for nested \pisiSE-extensions from the previous subsection, we are now in the position to handle the mixed ${\bs q}$-multibasic case. More precisely, we will generalize Theorem~\ref{imptThm:CombinedShiftCoprimeAndContentsAsRPIExtension} from the rational difference field to the mixed $\bs{q}$-multibasic difference field $\dField{\FF}$ with $\bs{q}=(q_1,\dots,q_{e-1})$. Here we assume that $\KK=K(\kappa_1,\dots,\kappa_u)(q_1,\dots,q_{e-1})$ is a rational function field over a field $K$ where $K$ is strongly $\sigma$-computable.
Following the notation from the previous subsection, we set $\FF_0:=\KK$ and $\FF_i=\FF_0(\myt_1)\dots(\myt_i)$ for $1\leq i\leq e$. This means that $\dField{\FF_0(\myt_1)}$ is the \sigmaE-extension of $\dField{\FF_0}$ with $\sigma(\myt_1)=\myt_1+1$ and $\dField{\FF_{i-1}(\myt_i)}$ is the \piE-extension of $\dField{\FF_{i-1}}$ with $\sigma(\myt_i)=q_{i-1}\,\myt_i$ for $2\leq i\leq e$.

As for the rational case we have to merge difference rings coming from different constructions.
Using Theorem~\ref{imptThm:MixedCasePeriodZeroShiftCoPrimePolynomialsAsPiExtension} instead of Theorem~\ref{imptThm:ShiftCoPrimePolynomialsAsPiExtension}, Lemma~\ref{lem:rPiExtCombined} generalizes straightforwardly to Lemma~\ref{lem:rPiExtCombined2}. Thus the proof is omitted here.

\begin{mlemma}\label{lem:rPiExtCombined2}
	Let $\dField{\FF_{e}}$ be the mixed ${\bs q}$-multi-basic difference field with $\FF_0=\KK$ from above. Further, let $\dField{\KK[y_{1},y_{1}^{-1}]\dots[y_{w},y_{w}^{-1}]}$ be a \piE-extension of $\dField{\KK}$ with $\frac{\s(y_i)}{y_i}\in\KK^*$ and $\dField{\FF_{e}[z_{1,{1}},z^{-1}_{1,{1}}]\dots[z_{1,{s_{1}}},z^{-1}_{1,{s_{1}}}]\dots[z_{e,{1}},z^{-1}_{e,{1}}] \dots [z_{e,{s_{e}}}, z^{-1}_{e,{s_{e}}}]}$ be a \piE-extension of $\dField{\FF_0}$ as given in Theorem~\ref{imptThm:MixedCasePeriodZeroShiftCoPrimePolynomialsAsPiExtension} with item (4). Then $\dField{\EE}$ with $\EE=\FF_e[y_{1},y_{1}^{-1}]\dots[y_{w},y_{w}^{-1}][z_{1,{1}},z^{-1}_{1,{1}}]\dots[z_{1,{s_{1}}},z^{-1}_{1,{s_{1}}}]\dots[z_{e,{1}},z^{-1}_{e,{1}}]\dots[z_{e,{s_{e}}}, z^{-1}_{e,{s_{e}}}]$ is	a \piE-extension of $\dField{\FF_e}$ . Furthermore, the \aE-extension $\dField{\EE[\vartheta]}$ of $\dField{\EE}$ with $\s(\vartheta)=\zeta\,\vartheta$ of order $\lambda$ is an \rE-extension.
\end{mlemma}


Gluing everything together, we obtain a generalization of
Theorem~\ref{imptThm:CombinedShiftCoprimeAndContentsAsRPIExtension}. Namely, one obtains an algorithmic construction of an \rpiE-extension in which one can represent a finite set of hypergeometric,  $q$-hypergeometric, ${\bs
q}$-multibasic hypergeometric and mixed ${\bs q}$-multibasic
hypergeometric products.

\begin{mtheorem}\label{imptThm:MixedCaseCombinedPeriodZeroShiftCoprimeAndContentsAsRPIExtension}
Let $\dField{\FF_{e}}$ be a mixed ${\bs q}$-multibasic {\dfE} of $\dField{\FF_{0}}$ with $\FF_{0}=\KK$ where $\KK=K(\kappa_1,\dots,\kappa_u)(q_1,\dots,q_{e-1})$ is a rational function field, $\s(\myt_{1})=\myt_{1}+1$ and $\s(\myt_{\ell})=q_{\ell-1}\,\myt_{\ell}$ for $2\leq\ell\leq e$. Let $\Lst{h}{1}{m}\in\FF_{e}^{*}$. Then one can define an \rpiE-extension $\dField{\AA}$ of
$\dField{\KK^{\prime}}$ with
	\small
		\begin{equation}\label{eqn:mixedHyperLaurentPoly}
		\AA=\KK^{\prime}(\myt_1)\dots(\myt_e)[\vartheta][y_{1},y^{-1}_{1}]\dots[y_{w},y^{-1}_{w}][z_{1,{1}},z^{-1}_{1,{1}}]\dots[z_{1,{s_{1}}},z^{-1}_{1,{s_{1}}}]\dots[z_{e,{1}},z^{-1}_{e,{1}}]\dots[z_{e,{s_{e}}},z^{-1}_{e,{s_{e}}}]
		\end{equation}
	\normalsize
and $\KK^{\prime}=K'(\kappa_1,\dots,\kappa_u)(q_1,\dots,q_{e-1})$ where $K^{\prime}$ is an algebraic field extension of $K$ such that
\begin{itemize}
	\item $\s(\vartheta)=\zeta\,\vartheta$ where $\zeta\in K^{\prime}$ is a $\lambda$-th root of unity.
	\item $\dfrac{\s(y_{j})}{y_{j}}=\alpha_{j}\in{\KK^{\prime}}\sm\zs$ for $1\le j\le w$ where the $\alpha_j$ are not roots of unity; 
	\item $\dfrac{\s(z_{i,{j}})}{z_{i,{j}}}=f_{i,{j}}\in\FF_{i-1}[\myt_{i}]\sm\FF_{i-1}$ are
monic, irreducible and shift co-prime;
	\end{itemize}
	holds with the following property. For $1\leq k\leq m$ one can define\footnote{We remark that this representation is related to the normal form given in~\cite{ZimingLi:11}.} 
		\begin{equation}\label{Equ:giForMixedCase}
		g_k=r_{k}\,\vartheta^{\mu_{k}}\,y_{1}^{u_{k,{1}}}\cdots
y_{1}^{u_{k,{w}}}\,z_{1,{1}}^{v_{k,1,{1}}}\cdots
z_{1,{s_{1}}}^{v_{k,1,{s_{1}}}}\,z_{2,{1}}^{v_{k,2,{1}}} \cdots
z_{2,{s_{2}}}^{v_{k,2,{s_{2}}}}\cdots z_{e,{1}}^{v_{k,e,{1}}}\cdots
z_{e,{s_{e}}}^{v_{k,e,{s_{e}}}}
		\end{equation}
		with $0\le\mu_k\le\lambda-1$, $u_{k,i}\in\ZZ$, $\nu_{k,i,j}\in\ZZ$ and 
$r_{k}\in\FF_{e}^{*}$ such that
$$\s(g_{k})=\s(h_{k})\,g_{k}.$$
	If $K$ is strongly $\sigma$-computable, the components of the theorem can be computed.
\end{mtheorem}

\begin{proof}
Take irreducible monic polynomials $\mathcal{B}=\{\Lst{p}{1}{n}\}\subseteq\FF_{0}[\myt_{1},\myt_{2},\dots,\myt_{e}]$ 
and take
$\Lst{\gamma}{1}{m}\in\FF_{0}^{*}$ such that for each $k$ with $1\leq k\leq m$ we get $\Lstc{d}{k}{1}{n}\in\ZZ$ with $\sigma(h_{k})=\gamma_{i}\,p_{1}^{d_{k,1}}\cdots p_{n}^{d_{k,n}}$.  Following the proof of Theorem~\ref{imptThm:CombinedShiftCoprimeAndContentsAsRPIExtension}, we can construct an \rpiE-extension $\FF_{0}^{\prime}(x)[\vartheta][y_{1}, y^{-1}_{1}] \dots [{y_{w}},y^{-1}_{w}]$ of $\dField{\FF_{0}^{\prime}(x)}$ with constant field $\FF_{0}^{\prime}=K'(\kappa_1,\dots,\kappa_u)(q_1,\dots,q_{e-1})$ where $K'$ is an algebraic extension of $K$ and the automorphism is defined as stated in Theorem~\ref{imptThm:MixedCaseCombinedPeriodZeroShiftCoprimeAndContentsAsRPIExtension} with the following property: we can define $a_k$ of the form~\eqref{Equ:ConstantPart} in this ring with~\eqref{Equ:ShiftConstantPart}.\\
Set $\mathcal{I}_{i}=\big\{\omega\in\mathcal{B}\,|\,\omega\in\KK[\myt_{1},\myt_{2},\dots,\myt_{i}]\sm\KK[\myt_{1},\myt_{2},\dots,\myt_{i-1}]\big\}$ for $1\leq i\leq e$ and define $I=\{1\leq i\leq e \mid \mathcal{I}_{i}\neq\{\}\}$.
Then for each $i\in I$ there is a partition $\mathcal{P}_{i}=\{\mathcal{E}_{i,{1}},\dots,\mathcal{E}_{i,{s_{i}}}\}$ of $\mathcal{I}_{i}$ w.r.t.\ the shift-equivalence of the automorphism defined for each $\myt_{i}$, i.e., each $\mathcal{E}_{i,j}$ with $1\leq j\leq s_i$ and $i\in I$ contains precisely the shift equivalent elements of 
$\mathcal{P}_{i}$. Take a representative from each equivalence class $\mathcal{E}_{i,{j}}$ in $\mathcal{P}_{i}$ and collect them in $\mathcal{F}_{i}:=\{\Lstc{f}{i}{1}{s_{i}}\}$. By construction it follows that property~(1) in Theorem~\ref{imptThm:MixedCasePeriodZeroShiftCoPrimePolynomialsAsPiExtension} holds; here we put all $\myt_i$ with $i\notin I$ in the ground field. Therefore by Theorem~\ref{imptThm:MixedCasePeriodZeroShiftCoPrimePolynomialsAsPiExtension} we obtain the \piE-extension $\dField{\FF_{e}(z_{1,{1}})\dots(z_{1,{s_{1}}})\dots(z_{e,{1}})\dots(z_{e,{s_{e}}})}$ of $\dField{\FF_{e}}$ with $\s(z_{i,{k}})=f_{i,{k}}\,z_{i,{k}}$ for all $i\in I$ and $1 \le k \le s_{i}$ with $s_i\in\NN\setminus\{0\}$; for $i\notin I$ we set $s_i=0$. By Lemma~\ref{lem:rPiExtCombined2} and Remark~\ref{remk:shiftR-Monomial}, $\dField{\AA}$ with~\eqref{eqn:mixedHyperLaurentPoly}
is an \rpiE-extension of $\dField{\FF_{0}^{\prime}(\myt_{1})\dots(\myt_{e})}$. Let $i,j$ with $i\in I$ and $1\le j \le s_{i}$. Since each $f_{i,j}$ is shift equivalent with every element of $\mathcal{E}_{i,j}$, it follows by Lemma~\ref{lem:ConstructShiftEquivPoly} that for all $h\in\mathcal{E}_{i,j}$, there is a rational function $0\neq r\in\FF_{i}\sm\FF_{i-1}$ with $h=\frac{\s(r)}{r}\,f_{i,j}$. 
Putting everything together we obtain for each $k$ with $1\leq k\leq m$, an $0\neq r_k\in\FF_{e}$ and ${\bs v_{k,i}}=(v_{k,i,1},\dots,v_{k,i,s_i})\in\ZZ^{s_{i}}$ with $p_{1}^{d_{k,1}}\cdots p_{n}^{d_{k,n}}=\frac{\s(r_{k})}{r_{k}}\,\vectorexpsubscript{f}{v}{1}\cdots\vectorexpsubscript{f}{v}{e}$. Note that for $$b_k:=r_k\,z_{1,{1}}^{v_{k,1,{1}}}\cdots z_{1,{s_{1}}}^{v_{k,1,{s_{1}}}}\,z_{2,{1}}^{v_{k,2,{1}}} \cdots z_{2,{s_{2}}}^{v_{k,2,{s_{2}}}}\cdots z_{e,{1}}^{v_{k,e,{1}}}\cdots z_{e,{s_{e}}}^{v_{k,e,{s_{e}}}}\in\AA$$
we have that $\s(b_k)=p_{1}^{d_{k,1}}\cdots p_{n}^{d_{k,n}}\,b_k.$ Now let $g_k\in\AA$ be as defined in~\eqref{Equ:giForMixedCase}. Since $g_k=a_k\,b_k$ where $a_k$  equals~\eqref{Equ:ConstantPart} and has the property~\eqref{Equ:ShiftConstantPart}, we get $\s(g_k)=\s(h_k)\,g_k$.
	The proof of the computational part is the same as that of
Theorem~\ref{imptThm:CombinedShiftCoprimeAndContentsAsRPIExtension}. \qed
\end{proof}

We are now ready to complete the proof for Theorem~\ref{thm:ProblemRMHPE}. To link to the notations used there, we set $\bs{q}=(q_1,\dots,q_{e-1})$ and set further $(x,t_1,\dots,t_{e-1})=(\myt_1,\dots,\myt_e)$, in particular we use the shortcut $\bs{t}=(\myt_2,\dots,\myt_{e-1})$. 
Suppose we are given the products~\eqref{eqn:MixedHypergeoPdts} and that we are given the components as stated in Theorem~\ref{imptThm:MixedCaseCombinedPeriodZeroShiftCoprimeAndContentsAsRPIExtension}. 
Then we follow the strategy as in Subsection~\ref{SubSec:HypergeoemtricCase}.\\
$\bullet$ Take the $\KK'$-embedding $\tau:\KK'(x,\bs{t})\to\ringOfEquivSeqs[\KK^{\prime}]$ where
$\tau(\frac{a}{b})=\langle\ev\big(\frac{a}{b}, n \big)\rangle_{n \geq 0}$ for  $a,b\in\KK^{\prime}[x,\bs{t}]$ is defined by~\eqref{Equ:EvqMixed}. Then by iterative application of part~(2) of Lemma~\ref{lem:injectiveHom} we can construct the $\KK'$-homomorphism $\tau:\AA\to\ringOfEquivSeqs[\KK^{\prime}]$ determined by the homomorphic extension with
\begin{itemize}
\item $\tau(\vartheta)=\langle\zeta^n\rangle_{n\geq0}$, 
\item $\tau(y_i)=\langle\alpha_i^n\rangle_{n\geq0}$ for $1\leq i\leq w$ and 
\item $\tau(z_{i,j})=\langle\displaystyle\prod_{k=\ell'_{i,j}}^nf_{i,j}(k-1,{\bs q}^{k-1})\rangle_{n\geq0}$ with $\ell'_{i,j}=Z(f_{i,j})+1$ for $1\leq i\leq e$, $1\leq j\leq s_i$.
\end{itemize}
In particular, since $\dField{\AA}$ is an \rpiE-extension of $\dField{\KK'(x,\bs{t})}$, it follows by part~(3) of Lemma~\ref{lem:injectiveHom} that $\tau$ is a $\KK^{\prime}$-embedding.\\ 
$\bullet$ Finally, define for $1\leq i\leq m$ the product expressions
\begin{align*}
 G_{i}(n)=&r_{i}(n)\,(\zeta^n)^{\mu_{i}}\,(\alpha_{1}^n)^{u_{i,{1}}}\cdots (\alpha_{w}^n)^{u_{i,{w}}}\\
 &\Big(\myProduct{k}{\ell'_{1,1}}{n}{f_{1,1}(k-1,{\bs q}^{k-1})}\Big)^{v_{i,{1},1}}\cdots \Big(\myProduct{k}{\ell'_{1,s_1}}{n}{f_{1,s_1}(k-1,{\bs q}^{k-1})}\Big)^{v_{i,1,{s_1}}}\dots\\[-0.3cm]
&\Big(\myProduct{k}{\ell'_{e,1}}{n}{f_{e,1}(k-1,{\bs q}^{k-1})}\Big)^{v_{i,{e},1}}\cdots \Big(\myProduct{k}{\ell'_{e,s_e}}{n}{f_{e,s_e}(k-1,{\bs q}^{k-1})}\Big)^{v_{i,e,{s_e}}}
\end{align*}
and
define $\delta_{i}=\max(\ell_{i},\ell'_{1,1},\dots,\ell'_{e,s_e},Z(r_{i}))$. Then observe that $\tau(g_i)=\langle G'_i(n)\rangle_{n\geq0}$ with~\eqref{Equ:GiPrime}. Now set $Q_i(n):=c\,G_i(n)$ with $c=\frac{P_i(\delta_i)}{G_i(\delta_i)}\in\KK^{\prime}$. Then as for the proof of the rational case we conclude that $P_i(n)=Q_i(n)$ for all $n\geq\delta_i$. This proves part (1) of Theorem~\ref{thm:ProblemRMHPE}. Since $\tau$ is a $\KK'$-embedding, the sequences \vspace*{-0.25cm} 
$$\funcSeqA{\alpha_{1}^{n}}{n},\dots,\funcSeqA{\alpha_{w}^{n}}{n}, \funcSeqA{\myProduct{k}{\beta_{1}}{n}{f_{1,1}(k-1,{\bs q}^{k-1})}}{n},\dots,\funcSeqA{\myProduct{k}{\beta_{e,s_e}}{n}{f_{e,s_e}(k-1,{\bs q}^{k-1})}}{n}\vspace*{-0.25cm}$$ 
are among each other algebraically independent over $\tau\big(\KK'(x)\big)\big[\funcSeqA{\zeta^{n}}{n}\big]$
which proves property $(2)$ of Theorem~\ref{thm:ProblemRMHPE}.

\begin{mexample}\label{exa:MixedHyperGeoPrdtsInSeqSetting}
	Let $\KK = K(q_{1}, q_{2})$ be the rational function field over the algebraic number field
	$K=\QQ(\sqrt{-3},\sqrt{-13})$, and consider the mixed ${\bs q}=(q_1,q_2)$-multibasic
	hypergeometric product expression\vspace*{-0.25cm}
	\begin{equation}\label{eqn:MixedHyperGeoPrdtsInSeqSetting}
	\hspace*{-0.3cm}P(n)=\myProduct{k}{1}{n}{\tfrac{
			\sqrt{-13}\,\big(k\,q_{1}^{k}+1\big)}{k^{2}\,
			\big(q_{1}^{k+1}\,q_{2}^{k+1}+k+1\big)}} +
	\myProduct{k}{1}{n}{\tfrac{k^{2}\,\big(k+q_{1}^{k}\,q_{2}^{k}\big)^{2}}{\sqrt{-3}\,(k+1)^{2}}}+
	\myProduct{k}{1}{n}{\tfrac{169\,\big(k\,q_{1}^{k}\,q_{2}^{k}+q_{2}^{k}+k\,q_{1}^{k}+1\big)}{\big(k\,q_{1}^{k+2}+2\,q_{1}^{k+2}+1\big)\,k^{2}}}. \vspace*{-0.25cm}
	\end{equation}
	Now take the mixed ${\bs q}$-multibasic {\dfE}
	$\dField{\KK(x)(t_{1})(t_{2})}$ of $\dField{\KK}$ with  $\s(x)=x+1$,
	$\s(t_{1})=q_{1}\,t_{1}$ and $\s(t_{2})=q_{2}\,t_{2}$. 
        Note that $h_1(k,q_1^k,q_2^k)$, $h_2(k,q_1^k,q_2^k)$ and $h_3(k,q_1^k,q_2^k)$ with         
        $$h_1=\tfrac{\sqrt{-13}\,(x\,t_1+1)}{x^{2}\,(q_{1}\,t_1\,q_{2}\,t_2+x+1)},\,
			h_2=\tfrac{x^{2}\,(x+t_1\,t_2)^{2}}{\sqrt{-3}\,(x+1)^{2}},\,
			h_3=\tfrac{169\,(x\,t_1\,t_2+t_2+x\,t_1+1}{(x\,q_{1}^{2}t_1+2\,q_{1}^{2}t_2+1)\,x^{2}}\in\KK(x,t_{1}, t_{2})$$ 
	are the multiplicands of the above products, respectively. Applying Theorem~\ref{imptThm:MixedCaseCombinedPeriodZeroShiftCoprimeAndContentsAsRPIExtension} we
	construct the algebraic number field extension $\KK'=\QQ\big((-1)^{\frac{1}{2}},\sqrt{3},\sqrt{13}\big)$ of $\KK$ and take the \pisiSE-extension $\dField{\FF'}$ of $\dField{\KK'}$ with $\FF'=\KK'(x)(t_1)(t_2)$ where $\sigma(x)=x+1$, $\sigma(t_1)=q_1\,t_1$ and $\sigma(t_2)=q_2\,t_2$. On top of this mixed multibasic difference field over $\KK'$ we construct the 
		\rpiE-extension $\dField{\AA}$
		with $\AA=\FF'[\vartheta]\geno{y_{1}}\geno{y_{2}}\geno{z_{1}}\geno{z_{2}}\geno{z_{3}}\geno{z_{4}}$ where the \rE-monomial $\vartheta$ with
		$\s(\vartheta)=(-1)^{\frac{1}{2}}\,\vartheta$ and the \piE-monomials $y_1,y_2,y_3$ with $\s(y_{1})=\sqrt{3}\,y_{1}$ and $\s(y_{2})=\sqrt{13}\,y_{2}$ are used to scope the content of the polynomials in $h_1,h_2,h_3$. Furthermore, the \piE-monomials $z_1,z_2,z_3,z_4$ with $\s(z_{1})=(x+1)\,z_{1}$,
		$\s(z_{2})=\big((x+1)\,q_{1}\,t_{1}+1\big)\,z_{2}$,
		$\s(z_{3})=(q_{2}\,t_{2}+1)\,z_{3}$,
		$\s(z_{4})=(q_{2}\,q_{1}\,t_{2}\,t_{1}+x+1)\,z_{4}$ are used to handle the monic polynomials in $h_1,h_2,h_3$. 
	These \piE-monomials are constructed in an iterative fashion as worked out in the proof of Theorem~\ref{imptThm:MixedCaseCombinedPeriodZeroShiftCoprimeAndContentsAsRPIExtension}. In particular, within this construction we derive	
	\fontsize{9.2}{0}\selectfont
	       \begin{equation*}
		Q =
		\underbrace{\dfrac{(q_{2}\,q_{1}+1)\,\vartheta\,y_{2}\,z_{2}}{\big(q_{2}\,q_{1}\,t_{2}\,t_{1}+x+1\big)\,z_{1}^{2}\,z_{4}}}_{=:g_1}+\,\underbrace{\dfrac{\vartheta^{3}\,z_{4}^{2}}{\big(x+1\big)^{2}\,y_{1}}}_{=:g_2}+\,\underbrace{\dfrac{\big(q_{1}+1\big)\,\big(2\,q_{1}^{2}+1\big)\,y_{2}^{4}\,z_{3}}{\big((x+1)\,q_{1}\,t_{1}+1\big)\,\big((x+2)\,q_{1}^{2}\,t_{1}+1\big)\,z_{1}^{2}}}_{=:g_3}
		\end{equation*}
\normalsize	
 such that $\sigma(g_i)=\sigma(h_i)\,g_i$ holds for $i=1,2,3$.\\ 
	Now take the $\KK'$-embedding $\tau:\KK'(x,\bs{t})\to\ringOfEquivSeqs[\KK^{\prime}]$ where
$\tau(\frac{a}{b})=\langle\ev\big(\frac{a}{b}, n \big)\rangle_{n \geq 0}$ for  $a,b\in\KK^{\prime}[x,\bs{t}]$ is defined by~\eqref{Equ:EvqMixed}. Then by iterative application of part~(2) of Lemma~\ref{lem:injectiveHom} we can construct the $\KK^{\prime}$-embedding $\tau:\AA\to\ringOfEquivSeqs[\KK^{\prime}]$ determined by the homomorphic extension of $\tau(\vartheta) = \constSeqA{\big((-1)^{\frac{1}{2}}\big)}{\hspace*{-0.06cm}n}$, $\tau(y_{1})=\langle\big(\sqrt{3}\big)^{n}\rangle_{n\ge0}$, $\tau(y_{2})=\langle\big(\sqrt{13}\big)^{n}\rangle_{n\ge0}$, $\tau(z_{1})=\langle n!\rangle_{n\ge0}$, $\tau(z_{2})=\langle \prod_{k=1}^{n}(k\,q_{1}^{k}+1) \rangle_{n\ge0}$, $\tau(z_{3})=\langle \prod_{k=1}^{n}(q_{2}^{k}+1) \rangle_{n\ge0}$ and $\tau(z_{4})=\langle \prod_{k=1}^{n}(q_{2}^{k}\,q_{1}^{k}+k) \rangle_{n\ge0}$. By our construction we can conclude that $\tau(g_1)$, $\tau(g_2)$ and $\tau(g_3)$ equal the sequences produced by the three products in~\eqref{eqn:MixedHyperGeoPrdtsInSeqSetting}, respectively. In particular, $\tau(Q)=\langle P(n)\rangle_{n\geq0}$. 
Furthermore, if we define 	
	\fontsize{9.2}{0}\selectfont
		\begin{multline*}
		Q(n) = \dfrac{(q_{2}\,q_{1}+1)}{\big(q_{2}^{n+1}\,q_{1}^{n+1}+n+1
			\big)}\,\rootOfUnitySeq{2}{n}\,\algSeq{13}{n}\,\frac{1}{\polySeq{n!}{2}}\,\myProduct{k}{1}{n}{\big(k\,q_{1}^{k}+1\big)}\,\myProduct{k}{1}{n}{\frac{1}{\big(q_{2}^{k}\,q_{1}^{k}+k\big)}}\\[-0.2cm]
			+\,\dfrac{1}{\big(n+1\big)^{2}}\,\rootOfUnitySeqExp{2}{n}{3}\algSeqExp{3}{n}{-1}\left(\myProduct{k}{1}{n}{\big(q_{2}^{k}\,q_{1}^{k}+k\big)}\right)^{2}\\[-0.2cm]
			+\,\dfrac{\big(q_{1}+1\big)\,\big(2\,q_{1}^{2}+1\big)}{\big((n+1)\,q_{1}^{n+1}+1\big)\,\big((n+2)\,q_{1}^{n+2}+1\big)}\,\algSeqExp{13}{n}{4}\,\frac{1}{\polySeq{n!}{2}}\,\myProduct{k}{1}{n}{\big(q_{2}^{k}+1\big)}
		\end{multline*}
\normalsize
		then we can guarantee that $P(n)=Q(n)$ for all $n\geq1$. The sequences generated by
$\algSeq{3}{n},\,{\algSeq{13}{n}},\,{n!},\,
			\myProduct{k}{1}{n}{\big(k\,q^{k}_{1}+1\big)},\,
			\myProduct{k}{1}{n}{\big(q^{k}_{2}+1\big)},\,
			\myProduct{k}{1}{n}{\big(q^{k}_{2}\,q^{k}_{1}+k \big)}$ 
		are algebraically independent among each other over $\tau(\KK'(x,\bs{t}))[\langle ((-1)^{\frac12})^n\rangle_{n\geq0}]$ by construction.
\end{mexample}

\section{Conclusion}\label{Sec:Conclusion}

We extended the earlier work~\cite{schneider2005product,Schneider:14} substantially and showed that
any expression in terms of hypergeometric products $\ProdExpr(\KK(n))$ 
can be formulated in an \rpisiSE-extension if the original constant field $\KK$ satisfies certain algorithmic properties.
This is in particular the case if $\KK=K(\kappa_1,\dots,\kappa_u)$ is a rational function field over an
algebraic number field $K$. In addition, we extended this machinery for the class of mixed $\bs{q}$-multibasic hypergeometric products.
Internally, we rely on Ge's algorithm~\cite{ge1993algorithms} that solves the orbit problem in $K$ and we utilize heavily 
results from difference ring theory~\cite{schneider2010parameterized,Singer:08,schneider2016difference,schneider2017summation}.
This product machinery implemented in Ocansey's package \texttt{NestedProducts} in combination with the summation machinery available in \texttt{Sigma}~\cite{Sigma} yields a complete summation
toolbox in which nested sums defined over $\ProdExpr(\KK(n,\bs{q^n}))$ can be represented and simplified using the summation paradigms of telescoping, creative telescoping and recurrence solving~\cite{petkovvsek1996b,Sigma}.

\medskip

\noindent\textbf{Acknowledgement.} We would like to thank 
Michael Karr for his valuable remarks to improve the article.


\providecommand{\bysame}{\leavevmode\hbox to3em{\hrulefill}\thinspace}
\providecommand{\MR}{\relax\ifhmode\unskip\space\fi MR }
\providecommand{\MRhref}[2]{%
  \href{http://www.ams.org/mathscinet-getitem?mr=#1}{#2}
}
\providecommand{\href}[2]{#2}

\end{document}